\documentclass[twocolumn,notitlepage,superscriptaddress,nofootinbib,tightenlines]{revtex4-1}
\usepackage{amsmath,verbatim,latexsym,amssymb,graphicx,indentfirst,mathrsfs,mathtools,amsthm,bbm,bm,hyperref,url,cancel,subcaption,yfonts}
\usepackage[font=small,labelfont=bf,format=plain,justification=raggedright,singlelinecheck=false]{caption}
\usepackage{verbatim,indentfirst}
\usepackage[title,titletoc]{appendix}
\usepackage[shortlabels]{enumitem}

\newtheorem{theorem}{Theorem}

\newtheorem{corollary}{Corollary}
\newtheorem{definition}{Definition}
\newtheorem{example}{Example}
\newtheorem{property}{Property}

\usepackage{color}
\usepackage[dvipsnames]{xcolor}
\usepackage[normalem]{ulem}

\usepackage{soul}

\newcommand{\W}{ \mathcal{W} }  
\newcommand{\Dim}{ d }  
\newcommand{\DegenW}{ \alpha }  
\newcommand{\DegenV}{ \lambda }  
\newcommand{\Sites}{N}  
\newcommand{\Sys}{S}  
\newcommand{\td}{t_{\rm d}}  
\newcommand{\Lyap}{\lambda_{\rm L}}  
\newcommand*{\OurKD}[1]{\tilde{A}_{#1}}  
\newcommand*{\SumKD}[1]{\tilde{ \mathscr{A} }_{#1}}  
\newcommand*{\ProjW}[1]{\Pi^{ \W }_{#1}}  
\newcommand*{\ProjWt}[1]{\Pi^{ \W(t) }_{#1}}  
\newcommand*{\ProjV}[1]{\Pi^{ V }_{#1}}  

\newcommand{\NondegW}{ \tilde{\W} }
\newcommand{\NondegV}{ \tilde{V} }
\newcommand{\reg}{ {\rm reg} } 
\newcommand{\TOC}{F_\toc} 
\newcommand{\GW}{ G_\W }
\newcommand{\GV}{ G_V }

\newcommand{\Charac}{ \mathcal{G} }  
\newcommand{\U}{ \mathcal{U} }  
\newcommand{\weak}{ {\rm weak} }
\newcommand{\target}{ {\rm target} }

\newcommand{\Protocol}{\mathcal{P}}  
\newcommand{\ProtocolA}{\mathscr{P}_\Amp}  
\newcommand{\toc}{{\rm TOC}}
\newcommand*{\TOCKD}[1]{\tilde{A}^\toc_{#1}}  
\newcommand{\A}{\mathcal{A}} 
\newcommand{\B}{ \mathcal{B} } 
\newcommand{\C}{ \mathcal{C} } 
\newcommand{\K}{ \mathcal{K} } 
\newcommand{\Oper}{\mathcal{O}} 
\newcommand{\Ops}{\mathscr{K}} 
\newcommand{\Opsb}{\bar{\mathscr{K}}} 
\newcommand{\Gest}{ \Gamma_{\rm est} } 
\newcommand{\Amp}{A} 

\newcommand*{\bra}[1]{\langle #1\rvert}
\newcommand*{\ket}[1]{\lvert #1 \rangle}

\newcommand*{\ketbra}[2]{\lvert #1 \rangle\!\langle #2 \rvert}
\newcommand*{\expval}[1]{\left\langle  #1  \right\rangle}

\newcommand{\Min}{ {\rm min} }
\newcommand{\Max}{ {\rm max} }
\newcommand{\inter}{ {\rm int} }   

\newcommand{\Tr}{{\rm Tr}}   
\def\id{\mathbbm{1}}   
\newcommand{\kB}{k_\mathrm{B}}  
\newcommand{\Hil}{\mathcal{H}}  
\newcommand{\Basis}{\mathcal{S}}  

\newcommand{\1}{ {(1)} }
\newcommand{\2}{ {(2)} }
\newcommand{\3}{ {(3)} }

\newcommand{\ParenK}{{(\Ops)}}
\newcommand{\ParenKB}{{( \Opsb )}}

\newcommand{\LParen}{ \bm{(} }
\newcommand{\RParen}{ \bm{)} }
\newcommand*{\Set}[1]{\left\{  #1  \right\}}
\newcommand*{\Unit}[1]{ \bm{ \hat{ #1 }} }  

\renewcommand\th{ {\rm th} }

\usepackage{graphicx}
\usepackage{epstopdf}
\newcommand{\caphead}[1]{{\bf #1}}

\makeatletter
\newcommand\footnoteref[1]{\protected@xdef\@thefnmark{\ref{#1}}\@footnotemark}
\makeatother

\begin{document}
\title{The quasiprobability behind the out-of-time-ordered correlator}
\author{Nicole~Yunger~Halpern\footnote{E-mail: nicoleyh@caltech.edu}}
\affiliation{Institute for Quantum Information and Matter, Caltech, Pasadena, CA 91125, USA}
\author{Brian~Swingle}
\affiliation{Department of Physics, Massachusetts Institute of Technology, Cambridge, MA 02139, USA}
\affiliation{Department of Physics, Harvard University, Cambridge, MA 02138, USA}
\affiliation{Department of Physics, Brandeis University, Waltham, MA 02453, USA}

\author{Justin~Dressel}\affiliation{Institute for Quantum Studies, Chapman University, Orange, CA 92866, USA}
\affiliation{Schmid College of Science and Technology, Chapman University, Orange, CA 92866, USA}
\date{\today}

%
%

%
%
\begin{abstract}
Two topics, evolving rapidly in separate fields, were combined recently: The out-of-time-ordered correlator (OTOC) signals quantum-information scrambling in many-body systems. The Kirkwood-Dirac (KD) quasiprobability represents operators in quantum optics. The OTOC was shown to equal a moment of a summed quasiprobability [Yunger Halpern, \emph{Phys. Rev. A} \textbf{95}, 012120 (2017)]. That quasiprobability, we argue, is an extension of the KD distribution. We explore the quasiprobability's structure from experimental, numerical, and theoretical perspectives. First, we simplify and analyze the weak-measurement and interference protocols for measuring the OTOC and its quasiprobability. We decrease, exponentially in system size, the number of trials required to infer the OTOC from weak measurements. We also construct a circuit for implementing the weak-measurement scheme. Next, we calculate the quasiprobability (after coarse-graining) numerically and analytically: We simulate a transverse-field Ising model first. Then, we calculate the quasiprobability averaged over random circuits, which model chaotic dynamics. The quasiprobability, we find, distinguishes chaotic from integrable regimes. We observe nonclassical behaviors: The quasiprobability typically has negative components. It becomes nonreal in some regimes. The onset of scrambling breaks a symmetry that bifurcates the quasiprobability, as in classical-chaos pitchforks. Finally, we present mathematical properties. We define an extended KD quasiprobability that generalizes the KD distribution. The quasiprobability obeys a Bayes-type theorem, for example, that exponentially decreases the memory required to calculate weak values, in certain cases. A time-ordered correlator analogous to the OTOC, insensitive to quantum-information scrambling, depends on a quasiprobability closer to a classical probability. This work not only illuminates the OTOC's underpinnings, but also generalizes quasiprobability theory and motivates immediate-future weak-measurement challenges.
\end{abstract}

{\let\newpage\relax\maketitle}

%
%
Two topics have been flourishing independently:
the out-of-time-ordered correlator (OTOC)
and the Kirkwood-Dirac (KD) quasiprobability distribution.
The OTOC signals chaos,
and the dispersal of information through entanglement,
in quantum many-body systems~\cite{Shenker_Stanford_14_BHs_and_butterfly,Shenker_Stanford_14_Multiple_shocks,Shenker_Stanford_15_Stringy,Roberts_15_Localized_shocks,Roberts_Stanford_15_Diagnosing,Maldacena_15_Bound}.
Quasiprobabilities represent quantum states
as phase-space distributions represent statistical-mechanical states~\cite{Carmichael_02_Statistical}.
Classical phase-space distributions are restricted to positive values;
quasiprobabilities are not.
The best-known quasiprobability is the Wigner function.
The Wigner function can become negative;
the KD quasiprobability, negative and nonreal~\cite{Kirkwood_33_Quantum,Dirac_45_On,Lundeen_11_Direct,Lundeen_12_Procedure,Bamber_14_Observing,Mirhosseini_14_Compressive,Dressel_15_Weak}.
Nonclassical values flag contextuality,
a resource underlying quantum-computation speedups~\cite{Spekkens_08_Negativity,Ferrie_11_Quasi,Kofman_12_Nonperturbative,Dressel_14_Understanding,Howard_14_Contextuality,Dressel_15_Weak,Delfosse_15_Wigner}.
Hence the KD quasiprobability, like the OTOC,
reflects nonclassicality.

Yet disparate communities use these tools:
The OTOC $F(t)$ features in quantum information theory,
high-energy physics, and condensed matter. Contexts include black holes within AdS/CFT duality~\cite{Shenker_Stanford_14_BHs_and_butterfly,Maldacena_98_AdSCFT,Witten_98_AdSCFT,Gubser_98_AdSCFT}, weakly interacting field theories~\cite{Stanford_15_WeakCouplingChaos,Patel_16_ChaosCritFS,Chowdhury_17_ONChaos,Patel_17_DisorderMetalChaos}, spin models~\cite{Shenker_Stanford_14_BHs_and_butterfly,HosurYoshida_16_Chaos}, and the Sachdev-Ye-Kitaev model~\cite{Sachdev_93_Gapless,Kitaev_15_Simple}. The KD distribution features in quantum optics.
Experimentalists have inferred the quasiprobability
from weak measurements of photons~\cite{Bollen_10_Direct,Lundeen_11_Direct,Lundeen_12_Procedure,Bamber_14_Observing,Mirhosseini_14_Compressive,Suzuki_16_Observation,Piacentini_16_Measuring,Thekkadath_16_Direct}
and superconducting qubits~\cite{White_16_Preserving,Groen_13_Partial}.

The two tools were united in~\cite{YungerHalpern_17_Jarzynski}.
The OTOC was shown to equal a moment of
a summed quasiprobability, $\OurKD{\rho}$:
\begin{align}
   \label{eq:JarzLike}
   F (t)  =  \frac{ \partial^2 }{ \partial \beta  \,  \partial \beta' }
   \expval{ e^{ - ( \beta W + \beta' W' ) } }
   \Bigg\rvert_{ \beta, \beta' = 0 }   \, .
\end{align}
$W$ and $W'$ denote measurable random variables
analogous to thermodynamic work; and $\beta, \beta'  \in  \mathbb{R}$.
The average $\expval{ . }$ is with respect to
a sum of quasiprobability values $\OurKD{\rho} ( . )$.
Equation~\eqref{eq:JarzLike} resembles Jarzynski's Equality,
a fluctuation relation in nonequilibrium statistical mechanics~\cite{Jarzynski_97_Nonequilibrium}.
Jarzynski cast a useful, difficult-to-measure free-energy difference $\Delta F$
in terms of the characteristic function of a probability.
Equation~\eqref{eq:JarzLike} casts the useful, difficult-to-measure OTOC
in terms of the characteristic function of a summed quasiprobability.\footnote{
For a thorough comparison of Eq.~\eqref{eq:JarzLike}
with Jarzynski's equality,
see the two paragraphs that follow the proof in~\cite{YungerHalpern_17_Jarzynski}.
}
The OTOC has recently been linked to thermodynamics also in~\cite{Campisi_16_Thermodynamics,Tsuji_16_Out}.

Equation~\eqref{eq:JarzLike} motivated
definitions of quantities
that deserve study in their own right.
The most prominent quantity is
the quasiprobability $\OurKD{\rho}$.
$\OurKD{\rho}$ is more fundamental than $F(t)$:
$\OurKD{\rho}$ is a distribution that consists of many values.
$F(t)$ equals a combination of those values---a
derived quantity, a coarse-grained quantity.
$\OurKD{\rho}$ contains more information than $F(t)$.
This paper spotlights $\OurKD{\rho}$
and related quasiprobabilities ``behind the OTOC.''

$\OurKD{\rho}$, we argue, is an extension of the KD quasiprobability.
Weak-measurement tools used to infer KD quasiprobabilities
can be applied to infer $\OurKD{\rho}$ from experiments~\cite{YungerHalpern_17_Jarzynski}.
Upon measuring $\OurKD{\rho}$, one can recover the OTOC.
Alternative OTOC-measurement proposals rely on
Lochshmidt echoes~\cite{Swingle_16_Measuring}, interferometry~\cite{Swingle_16_Measuring,Yao_16_Interferometric,YungerHalpern_17_Jarzynski,Bohrdt_16_Scrambling},
clocks~\cite{Zhu_16_Measurement},
particle-number measurements of ultracold atoms~\cite{Danshita_16_Creating,Tsuji_17_Exact,Bohrdt_16_Scrambling},
and two-point measurements~\cite{Campisi_16_Thermodynamics}.
Initial experiments have begun the push
toward characterizing many-body scrambling:
OTOCs of an infinite-temperature four-site NMR system
have been measured~\cite{Li_16_Measuring}.
OTOCs of symmetric observables have been measured
with infinite-temperature trapped ions~\cite{Garttner_16_Measuring} and in nuclear spin chains~\cite{Wei_16_NuclearSpinOTOC}.
Weak measurements offer a distinct toolkit,
opening new platforms and regimes to OTOC measurements.
The weak-measurement scheme in~\cite{YungerHalpern_17_Jarzynski}
is expected to provide a near-term challenge for
superconducting qubits~\cite{White_16_Preserving,Hacohen_16_Quantum,Rundle_16_Quantum,Takita_16_Demonstration,Kelly_15_State,Heeres_16_Implementing,Riste_15_Detecting},
trapped ions~\cite{Gardiner_97_Quantum,Choudhary_13_Implementation,Lutterbach_97_Method,Debnath_16_Nature,Monz_16_Realization,Linke_16_Experimental,Linke_17_Experimental},
ultracold atoms~\cite{Browaeys_16_Experimental}, cavity quantum electrodynamics (QED)~\cite{Guerlin_07_QND,Murch_13_SingleTrajectories}, and perhaps NMR~\cite{Xiao_06_NMR,Dawei_14_Experimental}.

We investigate the quasiprobability $\OurKD{\rho}$
that ``lies behind'' the OTOC.
The study consists of three branches:
We discuss experimental measurements,
calculate (a coarse-grained) $\OurKD{\rho}$,
and explore mathematical properties.
Not only does quasiprobability theory shed new light on the OTOC.
The OTOC also inspires questions about quasiprobabilities
and motivates weak-measurement experimental challenges.

The paper is organized as follows.
In a technical introduction, we review the KD quasiprobability, the OTOC,
the OTOC quasiprobability $\OurKD{\rho}$,
and schemes for measuring $\OurKD{\rho}$.
We also introduce our set-up and notation.
All the text that follows the technical introduction is new
(never published before, to our knowledge).

Next, we discuss experimental measurements.
We introduce a coarse-graining $\SumKD{\rho}$ of $\OurKD{\rho}$.
The coarse-graining involves a ``projection trick'' that decreases,
exponentially in system size, the number of trials required
to infer $F(t)$ from weak measurements.
We evaluate pros and cons of
the quasiprobability-measurement schemes in~\cite{YungerHalpern_17_Jarzynski}.
We also compare our schemes
with alternative $F(t)$-measurement schemes~\cite{Swingle_16_Measuring,Yao_16_Interferometric,Zhu_16_Measurement}.
We then present a circuit for weakly measuring
a qubit system's $\SumKD{\rho}$.
Finally, we show how to infer the coarse-grained $\SumKD{\rho}$
from alternative OTOC-measurement schemes
(e.g.,~\cite{Swingle_16_Measuring}).

Sections~\ref{section:Numerics} and~\ref{section:Brownian} feature
calculations of $\SumKD{\rho}$.
First, we numerically simulate a transverse-field Ising model.
$\SumKD{\rho}$ changes significantly, we find,
over time scales relevant to the OTOC.
The quasiprobability's behavior distinguishes
nonintegrable from integrable Hamiltonians.
The quasiprobability's negativity and nonreality remains robust
with respect to substantial quantum interference.
We then calculate an average, over Brownian circuits, of $\SumKD{\rho}$.
Brownian circuits model chaotic dynamics:
The system is assumed to evolve, at each time step,
under random two-qubit couplings~\cite{Brown_13_Scrambling,Hayden_07_Black,Sekino_08_Fast,Lashkari_13_Towards}.

A final ``theory'' section concerns mathematical properties
and physical interpretations of $\OurKD{\rho}$.
$\OurKD{\rho}$ shares some, though not all,
of its properties with the KD distribution.
The OTOC motivates a generalization of
a Bayes-type theorem obeyed by the KD distribution~\cite{Aharonov_88_How,Johansen_04_Nonclassical,Hall_01_Exact,Hall_04_Prior,Dressel_15_Weak}.
The generalization exponentially shrinks the memory required
to compute weak values, in certain cases.
The OTOC also motivates a generalization of
decompositions of quantum states $\rho$.
This decomposition property may help experimentalists assess
how accurately they prepared the desired initial state
when measuring $F(t)$.
A time-ordered correlator $F_\toc(t)$ analogous to $F(t)$, we show next,
depends on a quasiprobability that can reduce to a probability.
The OTOC quasiprobability lies farther from classical probabilities
than the TOC quasiprobability,
as the OTOC registers quantum-information scrambling
that $F_\toc(t)$ does not.
Finally, we recall that the OTOC encodes three time reversals.
OTOCs that encode more are
moments of sums of ``longer'' quasiprobabilities.
We conclude with theoretical and experimental opportunities.

We invite readers to familiarize themselves with the technical review,
then to dip into the sections that interest them most.
The technical review is intended to introduce
condensed-matter, high-energy, and quantum-information readers
to the KD quasiprobability
and to introduce quasiprobability and weak-measurement readers
to the OTOC.
Armed with the technical review,
experimentalists may wish to focus on Sec.~\ref{section:Measuring}
and perhaps Sec.~\ref{section:Numerics}.
Adherents of abstract theory may prefer Sec.~\ref{section:Theory}.
The computationally minded may prefer
Sections~\ref{section:Numerics} and~\ref{section:Brownian}.
The paper's modules (aside from the technical review) are independently accessible.

\onecolumngrid
\twocolumngrid

\section{Technical introduction}
\label{section:Tech_intro}

This review consists of three parts.
In Sec.~\ref{section:Intro_to_KD}, we overview the KD quasiprobability.
Section~\ref{section:SetUp} introduces our set-up and notation.
In Sec.~\ref{section:OTOC_review}, we review the OTOC and
its quasiprobability $\OurKD{\rho}$.
We overview also the weak-measurement and interference schemes
for measuring $\OurKD{\rho}$ and $F(t)$.

The quasiprobability section (\ref{section:Intro_to_KD}) provides background for
quantum-information, high-energy, and condensed-matter readers.
The OTOC section (\ref{section:OTOC_review}) targets
quasiprobability and weak-measurement readers.
We encourage all readers to study
the set-up (\ref{section:SetUp}), as well as
$\OurKD{\rho}$ and the schemes for measuring $\OurKD{\rho}$ (\ref{section:Review_OTOC_quasiprob}).

\subsection{The KD quasiprobability in quantum optics}
\label{section:Intro_to_KD}

The Kirkwood-Dirac quasiprobability is defined as follows.
Let $\Sys$ denote a quantum system
associated with a Hilbert space $\Hil$.
Let $\Set{ \ket{ a } }$ and $\Set{ \ket{ f } }$
denote orthonormal bases for $\Hil$.
Let $\mathcal{B} ( \Hil )$ denote the set of bounded operators
defined on $\Hil$, and let $\Oper  \in  \mathcal{B} ( \Hil )$.
The KD quasiprobability
\begin{align}
   \label{eq:KD_rho}   
   \OurKD{\Oper}^\1 ( a, f )  :=
   \langle f | a \rangle  \langle a |  \Oper  | f \rangle  \, ,
\end{align}
regarded as a function of $a$ and $f$,
contains all the information in $\Oper$,
if $\langle a | f \rangle \neq 0$ for all $a, f$.
Density operators $\Oper = \rho$ are often focused on
in the literature and in this paper.
This section concerns the context, structure,
and applications of $\OurKD{\Oper}^\1 ( a, f )$.

We set the stage with phase-space representations of quantum mechanics,
alternative quasiprobabilities, and historical background.
Equation~\eqref{eq:KD_rho}
facilitates retrodiction, or inference about the past,
reviewed in Sec.~\ref{section:KD_Retro}.
How to decompose an operator $\Oper$
in terms of KD-quasiprobability values
appears in Sec.~\ref{section:KD_Coeffs}.
The quasiprobability has mathematical properties
reviewed in Sec.~\ref{section:KDProps}.

Much of this section parallels Sec.~\ref{section:Theory},
our theoretical investigation of the OTOC quasiprobability.
More background appears in~\cite{Dressel_15_Weak}.

\subsubsection{Phase-space representations, alternative quasiprobabilities, and history}

Phase-space distributions form a mathematical toolkit
applied in Liouville mechanics~\cite{Landau_80_Statistical}.
Let $\Sys$ denote a system of $6 \Sites$ degrees of freedom (DOFs).
An example system consists of $\Sites$ particles,
lacking internal DOFs,
in a three-dimensional space.
We index the particles with $i$ and let $\alpha = x, y, z$.
The $\alpha^\th$ component $q_i^\alpha$
of particle $i$'s position
is conjugate to
the $\alpha^\th$ component $p_i^\alpha$
of the particle's momentum.
The variables $q_i^\alpha$ and $p_i^\alpha$ label
the axes of \emph{phase space}.

Suppose that the system contains many DOFs: $\Sites \gg 1$.
Tracking all the DOFs is difficult.
Which phase-space point $\Sys$ occupies,
at any instant, may be unknown.
The probability that, at time $t$, $\Sys$ occupies
an infinitesimal volume element
localized at $(q_1^x, \ldots, p_N^z)$ is
$\rho( \{ q_i^\alpha \} , \{ p_i^\alpha \}; t ) \, d^{3N} q  \:  d^{3N} p$.
The \emph{phase-space distribution}
$\rho( \{ q_i^\alpha \} , \{ p_i^\alpha \}; t )$
is a probability density.

$q_i^\alpha$ and $p_i^\alpha$
seem absent from quantum mechanics (QM), \emph{prima facie}.
Most introductions to QM cast quantum states in terms of
operators, Dirac kets $\ket{ \psi }$, and wave functions $\psi (x)$.
Classical variables are relegated to measurement outcomes
and to the classical limit.
Wigner, Moyal, and others represented QM
in terms of phase space~\cite{Carmichael_02_Statistical}.
These representations are used most in quantum optics.

In such a representation, a \emph{quasiprobability} density replaces
the statistical-mechanical probability density $\rho$.\footnote{
We will focus on discrete quantum systems,
motivated by a spin-chain example.
Discrete systems are governed by quasiprobabilities,
which resemble probabilities.
Continuous systems are governed by
quasiprobability densities,
which resemble probability densities.
Our quasiprobabilities can be replaced with quasiprobability densities,
and our sums can be replaced with integrals,
in, e.g., quantum field theory.}
Yet quasiprobabilities violate axioms of probability~\cite{Ferrie_11_Quasi}.
Probabilities are nonnegative, for example.
Quasiprobabilities can assume negative values,
associated with nonclassical physics such as contextuality~\cite{Spekkens_08_Negativity,Ferrie_11_Quasi,Kofman_12_Nonperturbative,Dressel_14_Understanding,Dressel_15_Weak,Delfosse_15_Wigner},
and nonreal values.
Relaxing different axioms leads to different quasiprobabilities.
Different quasiprobabilities correspond also to
different orderings of noncommutative operators~\cite{Dirac_45_On}.
The best-known quasiprobabilities include
the Wigner function, the Glauber-Sudarshan $P$ representation,
and the Husimi $Q$ function~\cite{Carmichael_02_Statistical}.

The KD quasiprobability resembles a little brother of theirs,
whom hardly anyone has heard of~\cite{Banerji_07_Exploring}.
Kirkwood and Dirac defined the quasiprobability independently
in 1933~\cite{Kirkwood_33_Quantum} and 1945~\cite{Dirac_45_On}.
Their finds remained under the radar for decades.
Rihaczek rediscovered the distribution in 1968,
in classical-signal processing~\cite{Rihaczek_68_Signal,Cohen_89_Time}.
(The KD quasiprobability is sometimes called
``the Kirkwood-Rihaczek distribution.'')
The quantum community's attention has revived recently.
Reasons include experimental measurements, mathematical properties,
and applications to retrodiction and state decompositions.

\subsubsection{Bayes-type theorem and retrodiction with
the KD quasiprobability}
\label{section:KD_Retro}

Prediction is inference about the future.
\emph{Retrodiction} is inference about the past. 
One uses the KD quasiprobability to infer about a time $t'$,
using information about an event that occurred before $t'$
and information about an event that occurred after $t'$.
This forward-and-backward propagation
evokes the OTOC's out-of-time ordering.

We borrow notation from,
and condense the explanation in,~\cite{Dressel_15_Weak}.
Let $\Sys$ denote a discrete quantum system.
Consider preparing $\Sys$ in
a state $\ket{ i }$ at time $t = 0$.
Suppose that $\Sys$ evolves under a time-independent Hamiltonian
that generates the family $U_t$ of unitaries.
Let $F$ denote an observable
measured at time $t'' > 0$.
Let $F = \sum_f  f  \ketbra{f }{ f }$ be the eigendecomposition,
and let $f$ denote the outcome.

Let $\A  =  \sum_a a  \ketbra{a}{a}$ be
the eigendecomposition of an observable
that fails to commute with $F$.
Let $t'$ denote a time in $(0, t'')$.
Which value can we most reasonably attribute
to the system's time-$t'$ $\A$,
knowing that $\Sys$ was prepared in $\ket{i}$
and that the final measurement yielded $f$?

Propagating the initial state forward to time $t'$ yields
$\ket{ i' }  :=  U_{t'}  \ket{i}$.
Propagating the final state backward yields
$\ket{f'}  :=  U^\dag_{t'' - t'} \ket{f}$.
Our best guess
about $\A$ is the \emph{weak value}~\cite{Ritchie_91_Realization,Hall_01_Exact,Johansen_04_Nonclassical,Hall_04_Prior,Pryde_05_Measurement,Dressel_11_Experimentals,Groen_13_Partial}
\begin{align}
   \label{eq:WeakVal}
   \A_\weak (i, f )  :=  \Re  \left(
   \frac{ \langle f' | \A | i' \rangle }{  \langle f' | i' \rangle }
   \right)  \, .
\end{align}
The real part of a complex number $z$
is denoted by $\Re (z)$.
The guess's accuracy is quantified with
a distance metric (Sec.~\ref{section:TA_retro})
and with comparisons to weak-measurement data.

Aharonov \emph{et al.} discovered weak values
in 1988~\cite{Aharonov_88_How}.
Weak values be \emph{anomalous}, or \emph{strange}:
$\A_\weak$ can exceed the greatest eigenvalue $a_\Max$ of $\A$
and can dip below the least eigenvalue $a_\Min$.
Anomalous weak values concur with
negative quasiprobabilities and nonclassical physics~\cite{Kofman_12_Nonperturbative,Dressel_14_Understanding,Pusey_14_Anomalous,Dressel_15_Weak,Waegell_16_Confined}.
Debate has surrounded weak values' role in quantum mechanics~\cite{Ferrie_14_How,Vaidman_14_Comment,Cohen_14_Comment,Aharonov_14,Sokolovski_14_Comment,Brodutch_15_Comment,Ferrie_15_Ferrie}.

The weak value $\A_\weak$, we will show,
depends on the KD quasiprobability.
We replace the $\A$ in Eq.~\eqref{eq:WeakVal}
with its eigendecomposition.
Factoring out the eigenvalues yields
\begin{align}
  \label{eq:WeakVal2}
   \A_\weak( i , f )  =   \sum_a  a  \,  \Re  \left(
   \frac{ \langle f' | a \rangle \langle a | i' \rangle }{
            \langle f' | i' \rangle }  \right) \, .
\end{align}
The weight $\Re ( . )$ is a \emph{conditional quasiprobability}.
It resembles a conditional probability---the
likelihood that, if $\ket{i}$ was prepared
and the measurement yielded $f$,
$a$ is the value most reasonably attributable to $\A$.
Multiplying and dividing the argument
by $\langle i' | f' \rangle$ yields
\begin{align}
   \label{eq:CondQuasi}
   \tilde{p} ( a | i, f )  :=   \frac{
   \Re  \left(  \langle f' | a \rangle \langle a | i' \rangle  \langle i' | f' \rangle
                      \right) }{
   |  \langle f' | i' \rangle  |^2 }     \, .
\end{align}
Substituting  into Eq.~\eqref{eq:WeakVal2} yields
\begin{align}
   \label{eq:WeakVal3}
   \A_\weak( i , f )  =   \sum_a  a  \,
   \tilde{p} ( a | i, f )  \, .
\end{align}

Equation~\eqref{eq:WeakVal3} illustrates why
negative quasiprobabilities concur with anomalous weak values.
Suppose that $\tilde{p} ( a | i, f )  \geq  0  \;  \:  \forall a$.
The triangle inequality, followed by the Cauchy-Schwarz inequality, implies
\begin{align}
   | \A_\weak( i , f ) |
   & \leq  \left\lvert  \sum_a  a  \,  \tilde{p} ( a | i, f )  \right\rvert \\
   & \leq  \sum_a  |a |  \cdot  |  \tilde{p} ( a | i, f )  |  \\
   \label{eq:Anom_help1}
   & \leq  |  a_\Max |  \sum_a    |  \tilde{p} ( a | i, f )  |  \\
   \label{eq:Anom_help2}
   & =  |  a_\Max |  \sum_a    \tilde{p} ( a | i, f ) \\
   & =  | a_\Max |  \, .
\end{align}
The penultimate equality follows from $\tilde{p}( a | i, f )  \geq  0$.
Suppose, now, that the quasiprobability contains
a negative value $\tilde{p} ( a_- | i, f ) < 0$.
The distribution remains normalized.
Hence the rest of the $\tilde{p}$ values sum to $>1$.
The RHS of~\eqref{eq:Anom_help1}
exceeds $|  a_\Max |$.
%

The numerator of Eq.~\eqref{eq:CondQuasi} is
the \emph{Terletsky-Margenau-Hill (TMH) quasiprobability}~\cite{Terletsky_37_Limiting,Margenau_61_Correlation,Johansen_04_Nonclassical,Johansen_04_Nonclassicality}.
The TMH distribution is the real part of a complex number.
That complex generalization,
\begin{align}
   \label{eq:KD}
   \langle f' | a \rangle    \langle a | i' \rangle   \langle i' | f' \rangle \, ,
\end{align}
is the KD quasiprobability~\eqref{eq:KD_rho}.

We can generalize the retrodiction argument to arbitrary states $\rho$~\cite{Wiseman_02_Weak}.
Let $\mathcal{D} ( \Hil )$ denote the set of density operators
(unit-trace linear positive-semidefinite operators)
defined on $\mathcal{H}$.
Let $\rho  =  \sum_i  p_i  \ketbra{ i }{ i }
\in  \mathcal{D} ( \Hil )$
be a density operator's eigendecomposition.
Let $\rho'  :=  U_{t'} \rho U_{t'}^\dag$.
The weak value Eq.~\eqref{eq:WeakVal} becomes
\begin{align}
   \label{eq:WeakVal_rho}
   \A_\weak ( \rho, f )  :=  \Re \left(  \frac{
   \langle f' | \A \rho' | f' \rangle }{
   \langle f' | \rho' | f' \rangle }  \right) \, .
\end{align}
Let us eigendecompose $\A$ and factor out $\sum_a a$.
The eigenvalues are weighted by the conditional quasiprobability
\begin{align}
   \tilde{p} ( a | \rho, f )  =  \frac{  \Re \left(
   \langle f' | a \rangle \langle a | \rho' | f' \rangle  \right) }{
   \langle f' | \rho' | f' \rangle } \, .
\end{align}
The numerator is the TMH quasiprobability for $\rho$.
The complex generalization
\begin{align}
   \label{eq:KD_rho_2}
   \OurKD{\rho}^\1 ( a, f )  =
   \langle f' | a \rangle  \langle a | \rho' | f' \rangle
\end{align}
is the KD quasiprobability~\eqref{eq:KD_rho} for $\rho$.\footnote{
The $A$ in the quasiprobability $\OurKD{\rho}$
should not be confused with
the observable $\A$.}
We rederive~\eqref{eq:KD_rho_2}, via
an operator decomposition, next.

\subsubsection{Decomposing operators in terms of
KD-quasiprobability coefficients}
\label{section:KD_Coeffs}

The KD distribution can be interpreted
not only in terms of retrodiction,
but also in terms of operation decompositions~\cite{Lundeen_11_Direct,Lundeen_12_Procedure}.
Quantum-information scientists decompose
qubit states in terms of Pauli operators.
Let $\bm{ \sigma } =
\sigma^x \hat{ \mathbf{ x } }  +  \sigma^y  \hat{ \mathbf{ y } }
+  \sigma^z  \hat{ \mathbf{ z } }$
denote a vector of the one-qubit Paulis.
Let $\hat{ \mathbf{n} }  \in  \mathbb{R}^3$ denote a unit vector.
Let $\rho$ denote any state of a \emph{qubit},
a two-level quantum system.
$\rho$ can be expressed as
$\rho  =  \frac{1}{2}
\left( \id  +  \hat{ \mathbf{n} }  \cdot  \bm{ \sigma }  \right) \, .$
The identity operator is denoted by $\id$.
The $\Unit{n}$ components $n_\ell$ constitute decomposition coefficients.
The KD quasiprobability consists of coefficients
in a more general decomposition.

Let $\Sys$ denote a discrete quantum system
associated with a Hilbert space $\mathcal{H}$.
Let $\Set{ \ket{ f } }$ and
$\Set{  \ket{a} }$ denote orthonormal bases for $\mathcal{H}$.
Let $\Oper \in \mathcal{B} ( \mathcal{H} )$ denote a bounded operator
defined on $\mathcal{H}$.
Consider operating on each side of $\Oper$
with a resolution of unity:
\begin{align}
   \Oper  & =  \id  \Oper  \id
   =  \left(  \sum_a  \ketbra{ a }{ a }  \right)  \Oper
       \left(  \sum_f  \ketbra{ f }{ f }  \right) \\
   & \label{eq:DecomposeHelp1}
   =  \sum_{a, f }  \ketbra{ a }{ f }  \;  \langle a | \Oper | f \rangle \, .
\end{align}
Suppose that every element of $\Set{ \ket{a} }$ has a nonzero overlap
with every element of $\Set{ \ket{f} }$:
\begin{align}
   \label{eq:OverlapCond}
   \langle f | a \rangle  \neq  0  \qquad \forall a, f \, .
\end{align}
Each term in Eq.~\eqref{eq:DecomposeHelp1}
can be multiplied and divided by the inner product:
\begin{align}
   \label{eq:StateDecomp}
   \Oper  =  \sum_{a , f }
   \frac{ \ketbra{ a }{ f } }{  \langle f | a \rangle }  \;
   \langle f | a \rangle   \langle a | \Oper | f \rangle \, .
\end{align}

Under condition~\eqref{eq:OverlapCond},
$\Set{ \frac{ \ketbra{ a }{ f } }{  \langle f | a \rangle } }$
forms an orthonormal basis for $\mathcal{B} ( \mathcal{ H } ) \, .$
[The orthonormality is with respect to
the Hilbert-Schmidt inner product.
Let $\Oper_1 ,  \Oper_2  \in  \mathcal{B} ( \Hil )$.
The operators have the Hilbert-Schmidt inner product
$( \Oper_1 ,  \,  \Oper_2  )  =  \Tr ( \Oper_1^\dag \Oper_2 )$.]
The KD quasiprobability
$\langle f | a \rangle   \langle a | \Oper | f \rangle$
consists of the decomposition coefficients.

Condition~\eqref{eq:OverlapCond} is usually assumed to hold~\cite{Lundeen_11_Direct,Lundeen_12_Procedure,Thekkadath_16_Direct}.
In~\cite{Lundeen_11_Direct,Lundeen_12_Procedure}, for example,
$\Set{ \ketbra{ a }{ a } }$ and $\Set{ \ketbra{ f }{ f } }$ manifest as
the position and momentum eigenbases
$\Set{ \ket{ x } }$ and $\Set{ \ket{ p } }$.
Let $\ket{ \psi }$ denote a pure state.
Let $\psi(x)$ and $\tilde{\psi} (p)$ represent $\ket{ \psi }$
relative to the position and momentum eigenbases.
The KD quasiprobability for
$\rho = \ketbra{ \psi }{ \psi }$ has the form
\begin{align}
   \label{eq:KD_ex}
   \OurKD{ \ketbra{\psi}{\psi} }^\1 ( p, x )
   & = \langle x | p \rangle  \langle p | \psi \rangle \langle \psi | x \rangle \\
   & =  \frac{ e^{ - i x p / \hbar } }{ \sqrt{ 2 \pi \hbar } }  \;
   \tilde{\psi} ( p )  \,  \psi^* ( x )  \, .
\end{align}
The OTOC motivates a relaxation of condition~\eqref{eq:OverlapCond}
(Sec.~\ref{section:TA_Coeffs}).
[Though assumed
in the operator decomposition~\eqref{eq:StateDecomp},
and assumed often in the literature,
condition~\eqref{eq:OverlapCond} need not hold
in arbitrary KD-quasiprobability arguments.]

\subsubsection{Properties of the KD quasiprobability}
\label{section:KDProps}

The KD quasiprobability shares some, but not all,
of its properties with other quasiprobabilities.
The notation below is defined as it has been
throughout Sec.~\ref{section:Intro_to_KD}.

\begin{property}   \label{prop:Complex}
The KD quasiprobability $\OurKD{\Oper}^\1 ( a, f )$ maps
$\mathcal{B} ( \mathcal{H} )   \times  \Set{ a }  \times  \Set{ f }$
to $\mathbb{C} \, .$
The domain is a composition of the set
$\mathcal{B} ( \mathcal{H} )$ of bounded operators
and two sets of real numbers.
The range is the set $\mathbb{C}$ of complex numbers,
not necessarily the set $\mathbb{R}$ of real numbers.
\end{property}

The Wigner function assumes only real values.
Only by dipping below zero can the Wigner function
deviate from classical probabilistic behavior.
The KD distribution's negativity has the following physical significance:
Imagine projectively measuring two (commuting) observables,
$\A$ and $\B$, simultaneously.
The measurement has some probability $p(a ; b)$
of yielding the values $a$ and $b$.
Now, suppose that $\A$ does not commute with $\B$.
No joint probability distribution $p(a ; b)$ exists.
Infinitely precise values cannot be ascribed
to noncommuting observables simultaneously.
Negative quasiprobability values are not observed directly:
Observable phenomena are modeled by
averages over quasiprobability values.
Negative values are visible only on scales
smaller than the physical coarse-graining scale.
But negativity causes observable effects,
visible in sequential measurements.
Example effects include anomalous weak values~\cite{Aharonov_88_How,Kofman_12_Nonperturbative,Dressel_14_Understanding,Pusey_14_Anomalous,Dressel_15_Weak,Waegell_16_Confined}
and violations of Leggett-Garg inequalities~\cite{Leggett_85_Quantum,Emary_14_LGI}.

Unlike the Wigner function, the KD distribution can assume nonreal values.
Consider measuring two noncommuting observables sequentially.
How much does the first measurement
affect the second measurement's outcome?
This disturbance is encoded in
the KD distribution's imaginary component~\cite{Hofmann_12_Complex,Dressel_12_Significance,Hofmann_14_Derivation,Hofmann_14_Sequential}.

\begin{property}
\label{prop:SumToProb}
Summing $\OurKD{\rho}^\1 ( a, f )$ over $a$
yields a probability distribution.
So does summing $\OurKD{\rho}^\1 ( a, f )$ over $f$.
\end{property} \noindent
Consider substituting $\Oper = \rho$ into Eq.~\eqref{eq:KD_rho}.
Summing over $a$ yields $\langle f | \rho | f \rangle$.
This inner product equals a probability, by Born's Rule.

\begin{property}
\label{prop:Discrete}
The KD quasiprobability is defined as in Eq.~\eqref{eq:KD_rho}
regardless of whether $\Set{ a }$ and $\Set{ f }$ are discrete.
\end{property} \noindent
The KD distribution and the Wigner function
were defined originally for continuous systems.
Discretizing the Wigner function
is less straightforward~\cite{Ferrie_11_Quasi,Delfosse_15_Wigner}.

\begin{property}
\label{prop:Bayes}
The KD quasiprobability obeys an analog of Bayes' Theorem,
Eq.~\eqref{eq:CondQuasi}.
\end{property}

Bayes' Theorem governs the conditional probability $p(f | i)$
that an event $f$ will occur, given that an event $i$ has occurred.
$p(f | i)$ is expressed in terms of
the conditional probability $p(i | f )$
and the absolute probabilities $p(i)$ and $p(f)$:
\begin{align}
   \label{eq:BayesThm}
   p( f | i )  =  \frac{ p ( i | f )  \:  p ( f ) }{ p ( i ) }  \, .
\end{align}

Equation~\eqref{eq:BayesThm} can be expressed in terms of
jointly conditional distributions.
Let $p ( a | i , f )$ denote the probability that an event $a$ will occur,
given that an event $i$ occurred
and that $f$ occurred subsequently.
$p( a, f | i )$ is defined similarly.
What is the joint probability $p( i, f, a )$ that $i$, $f$, and $a$ will occur?
We can construct two expressions:
\begin{align}
   \label{eq:ToBayes}
   p( i , f , a )  =  p ( a | i, f ) \,  p( i , f )
   =  p ( a , f | i )  \,  p ( i )  \, .
\end{align}
The joint probability $p ( i , f )$ equals $p ( f | i )  \,  p ( i )$.
This $p( i )$ cancels with the $p ( i )$ on
the right-hand side of Eq.~\eqref{eq:ToBayes}.
Solving for $p( a | i, f )$ yields
Bayes' Theorem for jointly conditional probabilities,
\begin{align}
   \label{eq:BayesThm2}
   p( a | i, f )  =  \frac{ p ( a, f | i ) }{ p ( f | i ) } \, .
\end{align}

Equation~\eqref{eq:CondQuasi} echoes Eq.~\eqref{eq:BayesThm2}.
The KD quasiprobability's Bayesian behavior~\cite{Hofmann_14_Derivation,Bamber_14_Observing}
has been applied to quantum state tomography~\cite{Lundeen_11_Direct,Lundeen_12_Procedure,Hofmann_14_Sequential,Salvail_13_Full,Malik_14_Direct,Howland_14_Compressive,Mirhosseini_14_Compressive}
and to quantum foundations~\cite{Hofmann_12_Complex}.

Having reviewed the KD quasiprobability,
we approach the extended KD quasiprobability behind the OTOC.
We begin by concretizing our set-up,
then reviewing the OTOC.

\subsection{Set-up}
\label{section:SetUp}

This section concerns the set-up and notation
used throughout the rest of this paper.
Our framework is motivated by the OTOC,
which describes quantum many-body systems.
Examples include black holes~ \cite{Shenker_Stanford_14_BHs_and_butterfly,Kitaev_15_Simple},
the Sachdev-Ye-Kitaev model~\cite{Sachdev_93_Gapless,Kitaev_15_Simple},
other holographic systems~\cite{Maldacena_98_AdSCFT,Witten_98_AdSCFT,Gubser_98_AdSCFT}
and spin chains.
We consider a system $\Sys$ associated with
a Hilbert space $\Hil$ of dimensionality $\Dim$.
The system evolves under a Hamiltonian $H$
that might be nonintegrable or integrable.
$H$ generates the time-evolution operator $U  :=  e^{ - i H t } \, .$

We will have to sum or integrate over spectra.
For concreteness, we sum, supposing that $\Hil$ is discrete.
A spin-chain example, discussed next, motivates our choice.
Our sums can be replaced with integrals
unless, e.g., we evoke spin chains explicitly.

We will often illustrate with a one-dimensional (1D) chain
of spin-$\frac{1}{2}$ degrees of freedom.
Figure~\ref{fig:Spin_chain} illustrates the chain,
simulated numerically in Sec.~\ref{section:Numerics}.
Let $\Sites$ denote the number of spins.
This system's $\Hil$ has dimensionality $\Dim  =  2^\Sites$.

%
%
\begin{figure}[h]
\centering
\includegraphics[width=.35\textwidth]{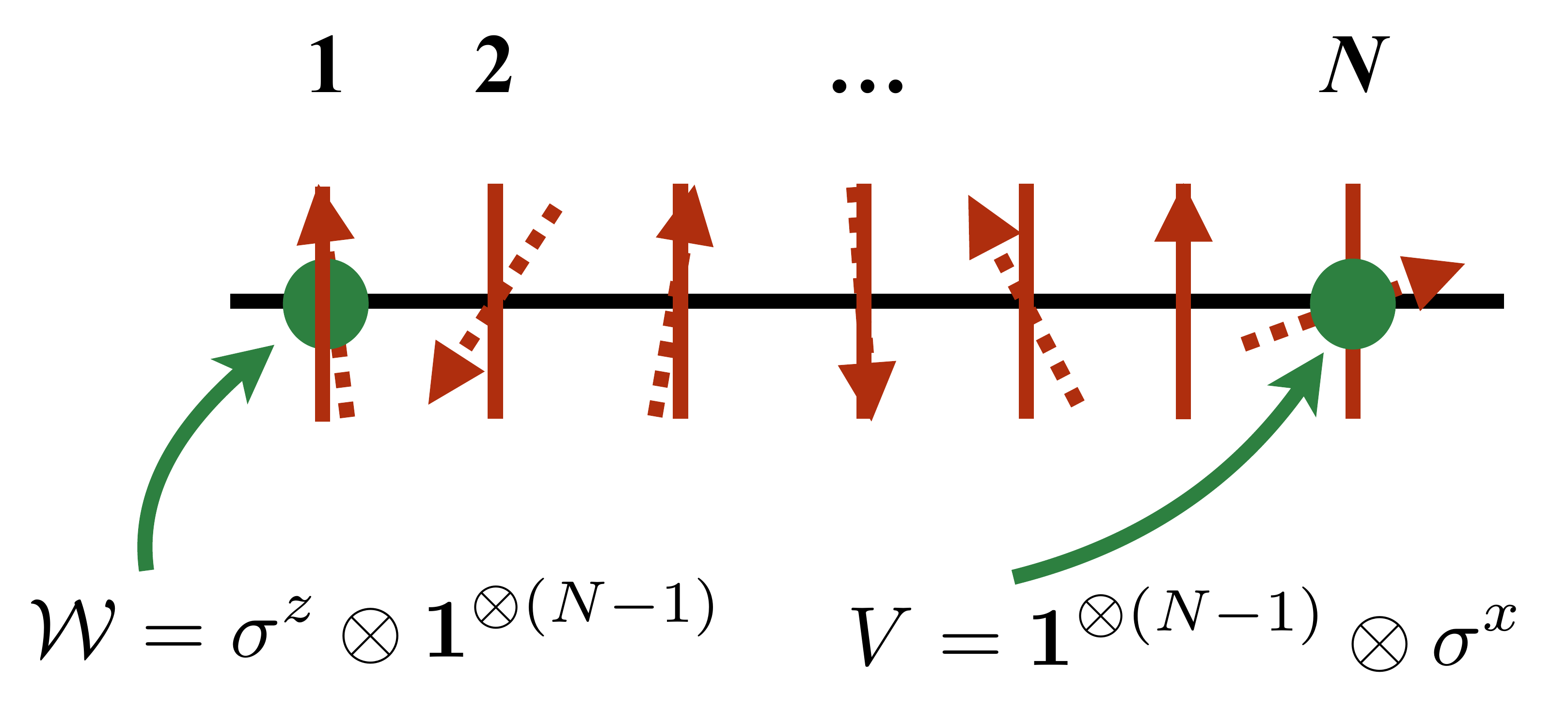}
\caption{\caphead{Spin-chain example:}
A spin chain exemplifies the quantum many-body systems
characterized by the out-of-time-ordered correlator (OTOC).
We illustrate with a one-dimensional chain
of $\Sites$ spin-$\frac{1}{2}$ degrees of freedom. 
The vertical red bars mark the sites.
The dotted red arrows illustrate how spins can point in arbitrary directions.
The OTOC is defined in terms of
local unitary or Hermitian operators $\W$ and $V$.
Example operators include single-qubit Paulis $\sigma^x$ and $\sigma^z$
that act nontrivially on opposite sides of the chain.}
\label{fig:Spin_chain}
\end{figure}

We will often suppose that $\Sys$ occupies, or is initialized to, a state
\begin{align}
   \label{eq:Rho}
   \rho  =  \sum_j  p_j  \ketbra{j}{j}   \in  \mathcal{D} (\Hil)  \, .
\end{align}
The set of density operators defined on $\Hil$
is denoted by $\mathcal{D} ( \Hil )$, as in Sec.~\ref{section:Intro_to_KD}.
Orthonormal eigenstates are indexed by $j$;
eigenvalues are denoted by $p_j$.
Much literature focuses on
temperature-$T$ thermal states $e^{ - H / T } / Z$.
(The partition function $Z$ normalizes the state.)
We leave the form of $\rho$ general,
as in~\cite{YungerHalpern_17_Jarzynski}.

The OTOC is defined in terms of local operators $\W$ and $V$.
In the literature, $\W$ and $V$ are assumed to be unitary and/or Hermitian.
Unitarity suffices for deriving
the results in~\cite{YungerHalpern_17_Jarzynski},
as does Hermiticity.
Unitarity and Hermiticity are assumed there,
and here, for convenience.\footnote{
Measurements of $\W$ and $V$ are discussed
in~\cite{YungerHalpern_17_Jarzynski} and here.
Hermitian operators $\GW$ and $\GV$ generate $\W$ and $V$.
If $\W$ and $V$ are not Hermitian,
$\GW$ and $\GV$ are measured instead of $\W$ and $V$.}
In our spin-chain example,
the operators manifest as one-qubit Paulis
that act nontrivially on opposite sides of the chain, e.g.,
$\W  =  \sigma^z \otimes \id^{ \otimes ( \Sites - 1 ) }$,
and $V  =  \id^{ \otimes ( \Sites - 1 ) } \otimes \sigma^x$.
In the Heisenberg Picture, $\W$ evolves as
$\W(t)  :=  U^\dag  \W  U \, .$

The operators eigendecompose as
\begin{align}
   \W  =  \sum_{w_\ell,  \DegenW_{w_\ell} }
   w_\ell  \ketbra{ w_\ell,  \DegenW_{w_\ell} }{ w_\ell,  \DegenW_{w_\ell} }
\end{align}
and
\begin{align}
   V  =  \sum_{ v_\ell,  \DegenV_{v_\ell} }
   v_\ell  \ketbra{ v_\ell,  \DegenV_{v_\ell} }{ v_\ell,  \DegenV_{v_\ell} } \, .
\end{align}
The eigenvalues are denoted by $w_\ell$ and $v_\ell$.
The degeneracy parameters are denoted by
$\DegenW_{ w_\ell }$ and $\DegenV_{v_\ell}$.
Recall that $\W$ and $V$ are local.
In our example, $\W$ acts nontrivially on just
one of $\Sites \gg 1$ qubits.
Hence $\W$ and $V$ are
exponentially degenerate in $\Sites$.
The degeneracy parameters can be measured:
Some nondegenerate Hermitian operator $\NondegW$
has eigenvalues in a one-to-one correspondence with
the $\DegenW_{ w_\ell }$'s.
A measurement of $\W$ and $\NondegW$
outputs a tuple $(w_\ell, \DegenW_{ w_\ell } )$.
We refer to such a measurement as ``a $\NondegW$ measurement,''
for conciseness.
Analogous statements concern $V$ and a Hermitian operator $\NondegV$.
Section~\ref{section:ProjTrick} introduces a trick
that frees us from bothering with degeneracies.

\subsection{The out-of-time-ordered correlator}
\label{section:OTOC_review}

Given two unitary operators $\mathcal{W}$ and $V$, the
out-of-time-ordered correlator is defined as
\begin{align}
   \label{eq:OTOC_Def}
   F(t) := \langle \mathcal{W}^\dagger(t) V^\dagger \mathcal{W}(t) V \rangle
   \equiv \text{Tr} \LParen  \rho \mathcal{W}^\dagger(t) V^\dagger
   \mathcal{W}(t) V  \RParen  \, .
\end{align}
This object reflects the degree of noncommutativity of $V$ and the Heisenberg operator $\mathcal{W}(t)$. More precisely, the OTOC appears in the expectation value of the squared magnitude of the commutator $[\mathcal{W}(t),V]$,
\begin{align}
C(t) := \langle [\mathcal{W}(t),V]^\dagger [\mathcal{W}(t),V] \rangle
= 2 - 2 \Re \LParen  F(t)  \RParen  \, .
\end{align}
Even if $\mathcal{W}$ and $V$ commute, the Heisenberg operator $\mathcal{W}(t)$ generically does not commute with $V$ at sufficiently late times.

An analogous definition involves Hermitian $\mathcal{W}$ and $V$.
The commutator's square magnitude becomes
\begin{align}
C(t) = - \langle [\mathcal{W}(t),V]^2\rangle.
\end{align}
This squared commutator involves TOC (time-ordered-correlator) and OTOC terms. The TOC terms take the forms
$\langle V \mathcal{W}(t) \mathcal{W}(t) V \rangle$ and
$\langle \mathcal{W}(t) V V \mathcal{W}(t) \rangle$.
[Technically, $\langle V \mathcal{W}(t) \mathcal{W}(t) V \rangle$
is time-ordered. $\langle \mathcal{W}(t) V V \mathcal{W}(t) \rangle$
behaves similarly.]

The basic physical process reflected by the OTOC is the spread of Heisenberg operators with time. Imagine starting with a simple $\mathcal{W}$, e.g., an operator acting nontrivially on just one spin in a many-spin system. Time-evolving yields $\mathcal{W}(t)$. The operator has grown if $\mathcal{W}(t)$ acts nontrivially on more spins than $\W$ does.
The operator $V$ functions as a probe for testing whether the action of $\mathcal{W}(t)$ has spread to the spin on which $V$ acts nontrivially.

Suppose $\mathcal{W}$ and $V$ are unitary and commute.
At early times, $\mathcal{W}(t)$ and $V$ approximately commute.
Hence $F(t) \approx 1$, and $C(t) \approx 0$.
Depending on the dynamics, at later times, $\mathcal{W}(t)$ may significantly fail to commute with $V$.
In a chaotic quantum system, $\W(t)$ and $V$ generically do not commute at late times, for most choices of $\W$ and $V$.

The analogous statement for Hermitian $\mathcal{W}$ and $V$ is that $F(t)$ approximately equals the TOC terms at early times. At late times, depending on the dynamics, the commutator can grow large. The time required for the TOC terms to approach their equilibrium values is called the \emph{dissipation time} $t_{\text{d}}$. This time parallels the time required for a system to reach local thermal equilibrium. The time scale on which the commutator grows to be order-one is called the \emph{scrambling time} $t_*$. The scrambling time parallels the time over which a drop of ink spreads across a container of water.

Why consider the commutator's square modulus?
The simpler object $\langle [\mathcal{W}(t),V]\rangle$
often vanishes at late times,
due to cancellations between states
in the expectation value. Physically, the vanishing of $\langle [\mathcal{W}(t),V]\rangle$ signifies that perturbing the system with $V$ does not significantly change the expectation value of $\mathcal{W}(t)$. This physics is expected for a chaotic system, which effectively loses its memory of its initial conditions. In contrast, $C(t)$ is the expectation value of a positive operator (the magnitude-squared commutator). The cancellations that zero out $\langle [\mathcal{W}(t),V]\rangle$ cannot zero out $\expval{ | [ \W(t) , V ] |^2 }$.

Mathematically, the diagonal elements of
the matrix that represents $[\mathcal{W}(t),V]$
relative to the energy eigenbasis
can be small.
$\expval{ [\mathcal{W}(t),V] }$, evaluated on a thermal state,
would be small.
Yet the matrix's off-diagonal elements can boost
the operator's Frobenius norm,
$\sqrt{  \Tr \left( | [\mathcal{W}(t),V] |^2 \right)  }$,
which reflects the size of $C(t)$.

We can gain intuition about the manifestation of chaos in $F(t)$
from a simple quantum system
that has a chaotic semiclassical limit.
Let $\mathcal{W} = q$ and $\mathcal{V} = p$
for some position $q$ and momentum $p$:
\begin{align}
C(t) = - \langle [q(t),p]^2 \rangle \sim \hbar^2 e^{2 \Lyap t}  \, .
\end{align}
This $\Lyap$ is a classical Lyapunov exponent. The final expression follows from the Correspondence Principle:
Commutators are replaced with $i \hbar$ times the corresponding Poisson bracket. The Poisson bracket of $q(t)$ with $p$ equals the derivative of the final position with respect to the initial position. This derivative reflects the butterfly effect in classical chaos, i.e., sensitivity to initial conditions. The growth of $C(t)$, and the deviation of $F(t)$ from the TOC terms, provide a quantum generalization of the butterfly effect.

Within this simple quantum system, the analog of the dissipation time may be regarded as $t_{\text{d}} \sim \Lyap^{-1}$. The analog of the scrambling time is $t_* \sim \Lyap^{-1} \ln \frac{\Omega}{\hbar}$. The $\Omega$ denotes some measure of the accessible phase-space volume. Suppose that the phase space is large in units of $\hbar$. The scrambling time is much longer than the dissipation time: $t_*  \gg  t_{\rm d}$.
Such a parametric separation between the time scales
characterizes the systems that interest us most.

In more general chaotic systems, the value of $t_*$ depends on whether the interactions are geometrically local and on $\mathcal{W}$ and $V$. Consider, as an example, a spin chain
 governed by a local Hamiltonian.
Suppose that $\mathcal{W}$ and $V$ are local operators
that act nontrivially on spins separated by a distance $\ell$. The scrambling time is generically proportional to $\ell$. For this class of local models,
$\ell/t_*$ defines a velocity $v_{\rm B}$ called the \emph{butterfly velocity}. Roughly, the butterfly velocity reflects how quickly initially local Heisenberg operators grow in space.

Consider a system in which $\td$ is separated parametrically from $t_*$.
The rate of change of $F(t)$
[rather, a regulated variation on $F(t)$]
was shown to obey a nontrivial bound.
Parameterize the OTOC as
$F(t) \sim \text{TOC} - \epsilon  \, e^{\Lyap t}$.
The parameter $\epsilon \ll 1$ encodes the separation of scales.
The exponent $\Lyap$ obeys $\Lyap \leq 2 \pi \kB T$ in thermal equilibrium at temperature $T$~\cite{Maldacena_15_Bound}. $\kB$ denotes Boltzmann's constant. Black holes in the AdS/CFT duality saturate this bound, exhibiting maximal chaos~\cite{Shenker_Stanford_14_BHs_and_butterfly,Kitaev_15_Simple}.

More generally, $\Lyap$ and $v_{\rm B}$ control the operators' growth and the spread of chaos. The OTOC has thus attracted attention for a variety of reasons, including (but not limited to) the possibilities of nontrivial bounds on quantum dynamics, a new probe of quantum chaos, and a signature of black holes in AdS/CFT.

\subsection{Introducing the quasiprobability
behind the OTOC}
\label{section:Review_OTOC_quasiprob}

$F(t)$ was shown, in~\cite{YungerHalpern_17_Jarzynski},
to equal a moment of a summed quasiprobability.
We review this result, established in four steps:
A quantum probability amplitude $A_\rho$
is reviewed in Sec.~\ref{section:Review_A} .
Amplitudes are combined to form the quasiprobability $\OurKD{\rho}$
in Sec.~\ref{section:Review_OTOC_quasiprob_sub}.
Summing $\OurKD{\rho}( . )$ values,
with constraints, yields a complex distribution
$P (W, W')$ in Sec.~\ref{section:Intro_PWWPrime}.
Differentiating $P(W, W')$ yields the OTOC.
$\OurKD{\rho}$ can be inferred experimentally
from a weak-measurement scheme
and from interference.
We review these schemes in Sec.~\ref{section:Intro_weak_meas}.

A third quasiprobability is introduced in Sec.~\ref{section:ProjTrick},
the \emph{coarse-grained quasiprobability} $\SumKD{\rho}$.
$\SumKD{\rho}$ follows from summing values of $\OurKD{\rho}$.
$\SumKD{\rho}$ has a more concise description than $\OurKD{\rho}$.
Also, measuring $\SumKD{\rho}$ requires fewer resources
(e.g., trials)
than measuring $\OurKD{\rho}$.
Hence Sections~\ref{section:Measuring}-\ref{section:Brownian}
will spotlight $\SumKD{\rho}$.
$\OurKD{\rho}$ returns to prominence in
the proofs of Sec.~\ref{section:Theory} and in
opportunities detailed in Sec.~\ref{section:Outlook}.
Different distributions suit different investigations.
Hence the presentation of three distributions in this thorough study:
$\OurKD{\rho}$, $\SumKD{\rho}$, and $P(W, W')$.

\subsubsection{Quantum probability amplitude $A_\rho$}
\label{section:Review_A}

The OTOC quasiprobability $\OurKD{\rho}$
is defined in terms of probability amplitudes $A_\rho$.
The $A_\rho$'s are defined in terms of the following process, $\ProtocolA$:
\begin{enumerate}[(1)]
   \item   Prepare $\rho$.
   \item   Measure the $\rho$ eigenbasis, $\Set{ \ketbra{j}{j} }$.
   \item   Evolve $\Sys$ forward in time under $U$.
   \item   Measure $\NondegW$.
   \item   Evolve $\Sys$ backward under $U^\dag$.
   \item   Measure $\NondegV$.
   \item   Evolve $\Sys$ forward under $U$.
   \item   Measure $\NondegW$.
\end{enumerate}
Suppose that the measurements yield the outcomes
$j$,  $(w_1, \DegenW_{w_1})$,  $(v_1, \DegenV_{v_1})$,
and  $(w_2, \DegenW_{w_2})$.
Figure~\ref{fig:Protocoll_Trial1} illustrates this process.
The process corresponds to the probability amplitude\footnote{
We order the arguments of $A_\rho$ differently
than in~\cite{YungerHalpern_17_Jarzynski}.
Our ordering here parallels
our later ordering of the quasiprobability's argument.
Weak-measurement experiments motivate
the quasiprobability arguments' ordering.
This motivation is detailed in Footnote~\ref{footnote:Order}.}
\begin{align}
   \label{eq:Amp}
   & A_\rho( j  ;  w_1,  \DegenW_{w_1}  ;  v_1, \DegenV_{v_1}  ;
                   w_2, \DegenW_{w_2}  )
   :=  \langle w_2, \DegenW_{w_2} | U | v_1, \DegenV_{v_1} \rangle
   \nonumber \\ & \qquad \times
   \langle  v_1, \DegenV_{v_1}  |  U^\dag  |  w_1,  \DegenW_{w_1}   \rangle
   \langle  w_1,  \DegenW_{w_1}   |  U  |  j  \rangle
   \sqrt{ p_j } \, .
\end{align}

%
%
\begin{figure}[h]
\centering
\begin{subfigure}{0.49\textwidth}
\centering
\includegraphics[width=.9\textwidth]{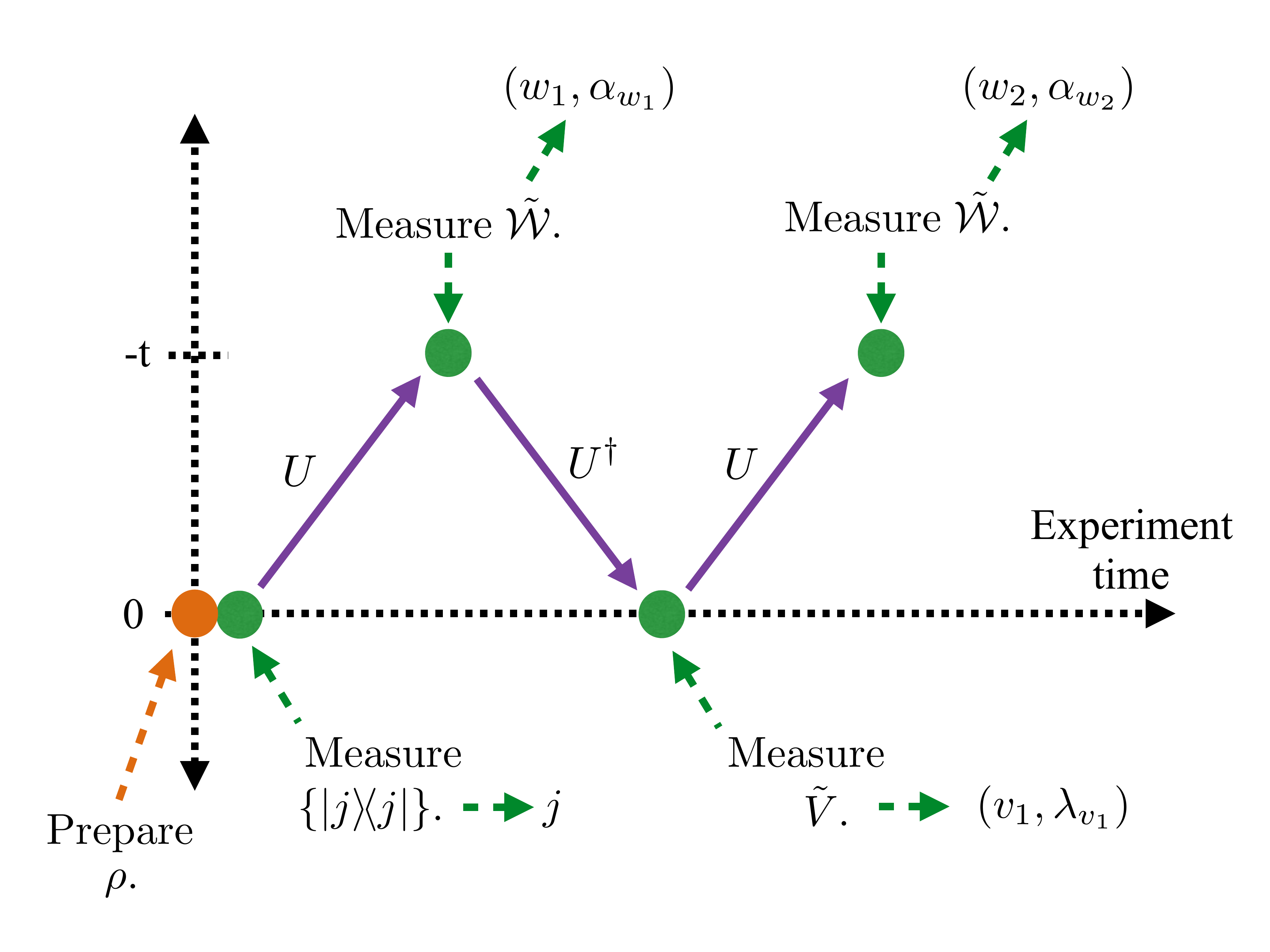}
\caption{}
\label{fig:Protocoll_Trial1}
\end{subfigure}
\begin{subfigure}{.49\textwidth}
\centering
\includegraphics[width=.9\textwidth]{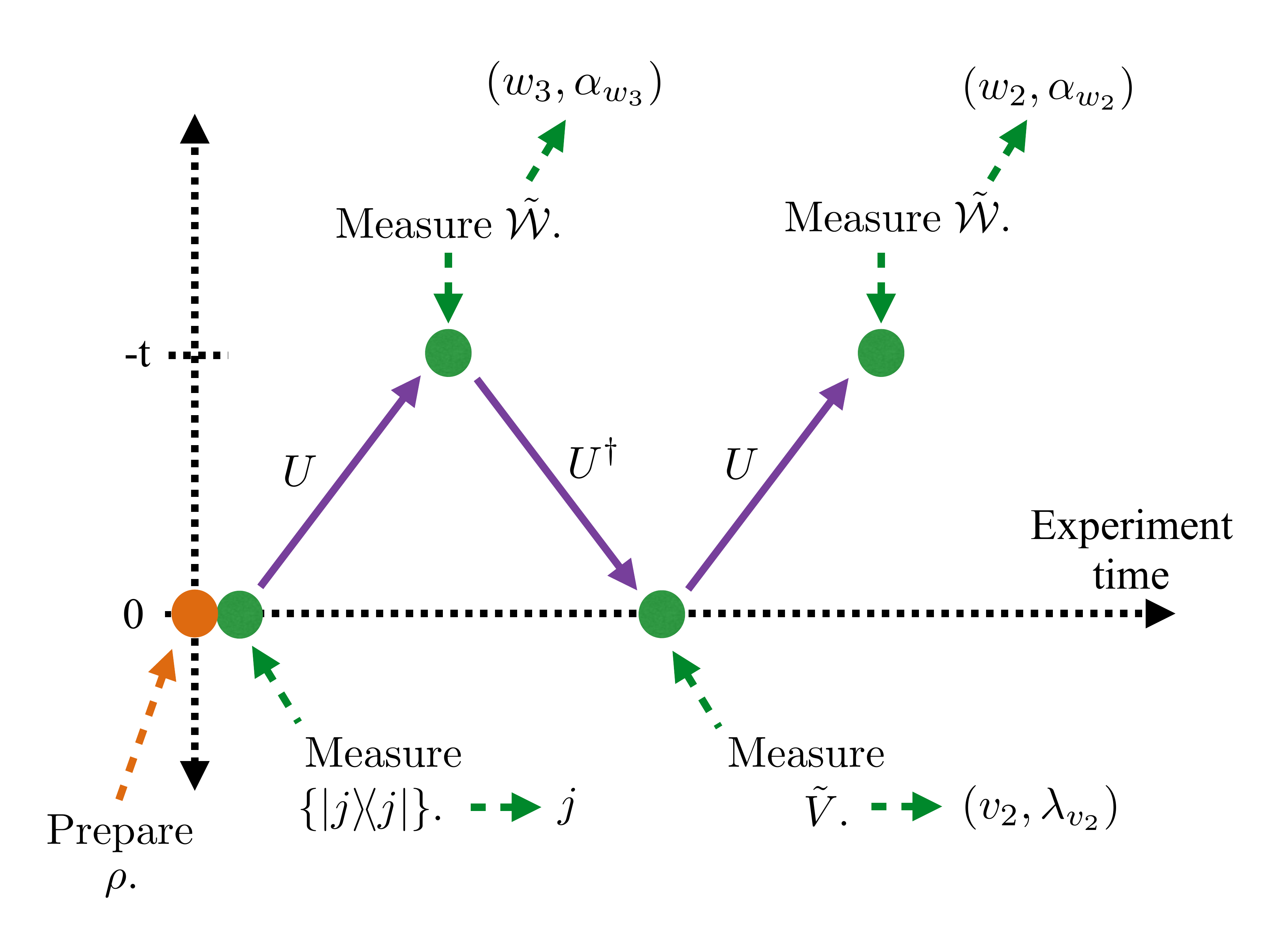}
\caption{}
\label{fig:Protocoll_Trial2}
\end{subfigure}
\caption{\caphead{Quantum processes
described by the probability amplitudes $\Amp_\rho$
in the out-of-time-ordered correlator (OTOC):}
These figures, and parts of this caption, appear in~\cite{YungerHalpern_17_Jarzynski}.
The OTOC quasiprobability $\OurKD{\rho}$ results from summing products
$A_\rho^*( . ) A_\rho( . )$.
Each $A_\rho( . )$ denotes a probability amplitude [Eq.~\eqref{eq:Amp}],
so each product resembles a probability.
But the amplitudes' arguments differ---the
amplitudes correspond to different quantum processes---because
the OTOC operators $\W(t)$ and $V$
fail to commute, typically.
Figure~\ref{fig:Protocoll_Trial1} illustrates
the process described by the $A_\rho( . )$;
and Fig.~\ref{fig:Protocoll_Trial2}, the process described by the $A_\rho^*( . )$.
Time, as measured by a laboratory clock, increases from left to right.
Each process begins with the preparation of
the state $\rho = \sum_j p_j \ketbra{j}{j}$
and a measurement of the state's eigenbasis.
Three evolutions ($U$, $U^\dag$, and $U$) then alternate with
three measurements of observables
($\NondegW$, $\NondegV$, and $\NondegW$).
Figures~\ref{fig:Protocoll_Trial1} and~\ref{fig:Protocoll_Trial2}
are used to define $\OurKD{\rho}$,
rather than showing protocols for measuring $\OurKD{\rho}$.}
\label{fig:Protocoll}
\end{figure}

We do not advocate for performing $\ProtocolA$ in any experiment.
$\ProtocolA$ is used to define $A_\rho$
and to interpret $A_\rho$ physically.
Instances of $A_\rho$ are combined into $\OurKD{\rho}$.
A weak-measurement protocol can be used
to measure $\OurKD{\rho}$ experimentally.
An interference protocol can be used to measure
$A_\rho$ (and so $\OurKD{\rho}$) experimentally.

\subsubsection{The fine-grained OTOC quasiprobability $\OurKD{\rho}$}
\label{section:Review_OTOC_quasiprob_sub}

The quasiprobability's definition is constructed as follows.
Consider a realization of $\ProtocolA$ that yields the outcomes
$j$, $( w_3, \DegenW_{w_3} )$, $(  v_2, \DegenV_{v_2} )$, and
$( w_2 , \DegenW_{w_2} )$.
 Figure~\ref{fig:Protocoll_Trial2} illustrates this realization.
The initial and final measurements yield
the same outcomes as in the~\eqref{eq:Amp} realization.
We multiply the complex conjugate of
the second realization's amplitude
by the first realization's probability amplitude.
Then, we sum over $j$ and $(w_1, \DegenW_{w_1})$:\footnote{
Familiarity with tensors might incline one to sum
over the $(w_2,  \DegenW_{w_2})$
shared by the trajectories.
But we are not invoking tensors.
More importantly, summing over $(w_2,  \DegenW_{w_2})$
introduces a $\delta_{ v_1 v_2 }  \delta_{ \DegenV_{v_1} \DegenV_{v_2} }$
that eliminates one $( v_\ell,  \DegenV_{v_\ell} )$ degree of freedom.
The resulting quasiprobability would not ``lie behind'' the OTOC.
One could, rather than summing over $( w_1,  \DegenW_{w_1} )$,
sum over $(w_3,  \DegenW_{w_3})$.
Either way, one sums over one trajectory's first $\NondegW$ outcome.
We sum over $(w_1,  \DegenW_{w_1})$ to maintain consistency with~\cite{YungerHalpern_17_Jarzynski}.}$^,$\footnote{   \label{footnote:Order}
In~\cite{YungerHalpern_17_Jarzynski},
the left-hand side's arguments are ordered differently
and are condensed into the shorthand
$(w, v,  \DegenW_w, \DegenV_v)$.
Experiments motivate our reordering:
Consider inferring $\OurKD{\rho} ( a , b , c , d )$ from
experimental measurements.
In each trial, one (loosely speaking) weakly measures
$a$, then $b$, then $c$;
and then measures $d$ strongly.
As the measurements are ordered, so are the arguments.}
\begin{align}
   \label{eq:TADef}
   & \OurKD{\rho} ( v_1,  \DegenV_{v_1} ;  w_2,  \DegenW_{w_2} ;
   v_2,  \DegenV_{v_2}  ;  w_3,  \DegenW_{w_3}  )
   \nonumber  \\ &
   :=  \sum_{j , (w_1, \DegenW_{w_1} ) }
   A_{\rho}^* ( j  ;  w_3,  \DegenW_{w_3}  ;  v_2, \DegenV_{v_2}  ;
                   w_2, \DegenW_{w_2} )
  \nonumber  \\ & \qquad \qquad \qquad \times
  A_{\rho} ( j  ;  w_1,  \DegenW_{w_1}  ;  v_1, \DegenV_{v_1}  ;
                   w_2, \DegenW_{w_2} ) \, .
\end{align}

Equation~\eqref{eq:TADef} resembles a probability
but differs due to the noncommutation of $\W(t)$ and $V$.
We illustrate this relationship in two ways.

Consider a 1D quantum system, e.g., a particle on a line.
We represent the system's state with
a wave function $\psi(x)$.
The probability density at point $x$ equals $\psi^* (x)  \,  \psi(x)$.
The $A^*_\rho  \,  A_\rho$ in Eq.~\eqref{eq:TADef} echoes $\psi^* \psi$.
But the argument of the $\psi^*$ equals
the argument of the $\psi$.
The argument of the $A^*_\rho$ differs from
the argument of the $A_\rho$,
because $\W(t)$ and $V$ fail to commute.

Substituting into Eq.~\eqref{eq:TADef} from Eq.~\eqref{eq:Amp} yields
\begin{align}
   \label{eq:TAForm}
   & \OurKD{\rho} ( v_1,  \DegenV_{v_1} ;  w_2,  \DegenW_{w_2} ;
   v_2,  \DegenV_{v_2}  ;  w_3,  \DegenW_{w_3} )
   \nonumber \\ &
   =  \langle  w_3 , \DegenW_{w_3}  |  U  |
                   v_2, \DegenV_{v_2}  \rangle
   \langle  v_2, \DegenV_{v_2}  |  U^\dag  |
               w_2,  \DegenW_{w_2}  \rangle
   \nonumber \\ & \qquad \times
   \langle  w_2,  \DegenW_{w_2}  |  U  |  v_1,  \DegenV_{v_1}  \rangle
   \langle  v_1,  \DegenV_{v_1}  |  \rho  U^\dag  |  w_3 , \DegenW_{w_3}  \rangle \, .
\end{align}
A simple example illustrates how
$\OurKD{\rho}$ nearly equals a probability.
Suppose that an eigenbasis of $\rho$ coincides with
$\Set{ \ketbra{ v_\ell,  \DegenV_{v_\ell} }{ v_\ell,  \DegenV_{v_\ell} } }$
or  with  $\Set{ U^\dag  \ketbra{ w_\ell , \DegenW_{w_\ell} }{
                        w_\ell , \DegenW_{w_\ell} } U }$.
Suppose, for example, that
\begin{align}
   \label{eq:WRho}
   \rho =  \rho_{V}  :=
   \sum_{ v_\ell , \DegenV_{v_\ell} }
   p_{ v_\ell , \DegenV_{v_\ell} }
   \ketbra{ v_\ell , \DegenV_{v_\ell} }{ v_\ell , \DegenV_{v_\ell} }  \, .
\end{align}
One such $\rho$ is the infinite-temperature Gibbs state
$\id / \Dim$.
Another example is easier to prepare:
Suppose that $\Sys$ consists of $\Sites$ spins
and that $V = \sigma^x_\Sites$.
One $\rho_V$ equals a product of $\Sites$ $\sigma^x$ eigenstates.
Let $( v_2 ,  \DegenV_{v_2} )  =  ( v_1,  \DegenV_{v_1} )$.
[An analogous argument follows from
$( w_3 , \DegenW_{w_3} ) = ( w_2, \DegenW_{w_2} )$.]
Equation~\eqref{eq:TAForm} reduces to
\begin{align}
   \label{eq:Reduce_to_p}
   & | \langle  w_2,  \DegenW_{w_2}  |  U  |
        v_1,  \DegenV_{v_1}  \rangle |^2  \,
   |  \langle  w_3,  \DegenW_{w_3}  |  U  |
       v_1,  \DegenV_{v_1}  \rangle |^2  \,
   p_{ v_1 , \DegenV_{v_1} } \, .
\end{align}
Each square modulus equals a conditional probability.
$p_{ v_1,  \DegenV_{v_1} }$ equals the probability that,
if $\rho$ is measured with respect to
$\Set{ \ketbra{ v_\ell , \DegenV_{v_\ell} }{
                        v_\ell , \DegenV_{v_\ell} } }$,
outcome $( v_1,  \DegenV_{v_1} )$ obtains.

In this simple case,
certain quasiprobability values
equal probability values---the
quasiprobability values that satisfy
$( v_2,   \DegenV_{v_2} )  =  ( v_1,  \DegenV_{v_1} )$
or $( w_3 ,  \DegenW_{w_3} )  =  ( w_2,  \DegenW_{w_2} )$.
When both conditions are violated, typically,
the quasiprobability value does not
equal a probability value.
Hence not all the OTOC quasiprobability's values
reduce to probability values.
Just as a quasiprobability lies behind the OTOC,
quasiprobabilities lie behind
time-ordered correlators (TOCs).
Every value of a TOC quasiprobability
reduces to a probability value
in the same simple case (when $\rho$ equals, e.g., a $V$ eigenstate)
(Sec.~\ref{section:OTOC_TOC}).

\subsubsection{Complex distribution $P(W, W')$}
\label{section:Intro_PWWPrime}

$\OurKD{\rho}$ is summed, in~\cite{YungerHalpern_17_Jarzynski},
to form a complex distribution $P(W, W')$.
Let $W :=  w_3^*  v_2^*$ and $W'  :=  w_2  v_1$ denote
random variables calculable from measurement outcomes.
If $\W$ and $V$ are Paulis,
$(W, W')$ can equal $(1, 1), (1, -1), (-1, 1),$ or $(-1, -1)$.

$W$ and $W'$ serve, in the Jarzynski-like equality~\eqref{eq:JarzLike},
analogously to thermodynamic work $W_\th$ in Jarzynski's equality.
$W_\th$ is a random variable,
inferable from experiments,
that fluctuates from trial to trial.
So are $W$ and $W'$.
One infers a value of $W_\th$
by performing measurements and processing the outcomes.
The \emph{two-point measurement scheme} (TPMS)
illustrates such protocols most famously.
The TPMS has been used to derive
quantum fluctuation relations~\cite{Tasaki00}.
One prepares the system in a thermal state,
measures the Hamiltonian, $H_i$, projectively;
disconnects the system from the bath;
tunes the Hamiltonian to $H_f$;
and measures $H_f$ projectively.
Let $E_i$ and $E_f$ denote the measurement outcomes.
The work invested in the Hamiltonian tuning
is defined as $W_\th  :=  E_f  -  E_i$.
Similarly, to infer $W$ and $W'$,
one can measure $\W$ and $V$
as in Sec.~\ref{section:Intro_weak_meas},
then multiply the outcomes.

Consider fixing the value of $(W, W')$.
For example, let  $(W, W')  =  (1, -1)$.
Consider the octuples
$( v_1,  \DegenV_{v_1} ;  w_2,  \DegenW_{w_2}  ;
    v_2,  \DegenV_{v_2}   ;  w_3,  \DegenW_{w_3} )$
that satisfy the constraints $W = w_3^*  v_2^*$
and $W'  =  w_2  v_1$.
Each octuple corresponds to
a quasiprobability value $\OurKD{\rho} ( . )$.
Summing these quasiprobability values yields
\begin{align}
   \label{eq:PWWPrime}
   & P(W, W')  :=  \sum_{ \substack{
   ( v_1,  \DegenV_{v_1} ) , ( w_2,  \DegenW_{w_2} ) ,
   ( v_2,  \DegenV_{v_2} ) , ( w_3,  \DegenW_{w_3} ) } }
   \\  \nonumber  &
   \OurKD{\rho} ( v_1,  \DegenV_{v_1} ;  w_2,  \DegenW_{w_2} ;
   v_2,  \DegenV_{v_2}  ;  w_3,  \DegenW_{w_3} )  \:
   \delta_{W ( w_3^*  v_2^* ) }    \delta_{ W' (w_2  v_1) } \, .
\end{align}
The Kronecker delta is represented by $\delta_{ab}$.
$P(W, W')$ functions analogously to
the probability distribution, in the fluctuation-relation paper~\cite{Jarzynski_97_Nonequilibrium},
over values of thermodynamic work.

The OTOC equals a moment of $P(W, W')$ [Eq.~\eqref{eq:JarzLike}],
which equals a constrained sum over $\OurKD{\rho}$~\cite{YungerHalpern_17_Jarzynski}.
Hence our labeling of $\OurKD{\rho}$ as a
``quasiprobability behind the OTOC.''
Equation~\eqref{eq:PWWPrime} expresses
the useful, difficult-to-measure $F(t)$
in terms of a characteristic function of a (summed) quasiprobability,
as Jarzynski~\cite{Jarzynski_97_Nonequilibrium} expresses
a useful, difficult-to-measure
free-energy difference $\Delta F$
in terms of a characteristic function of a probability.
Quasiprobabilities reflect nonclassicality (contextuality)
as probabilities do not;
so, too, does $F(t)$ reflect nonclassicality (noncommutation)
as $\Delta F$ does not.

The definition of $P$ involves arbitrariness:
The measurable random variables, and $P$,
may be defined differently.
Alternative definitions, introduced in Sec.~\ref{section:HigherOTOCs},
extend more robustly to
OTOCs that encode more time reversals.
All possible definitions share two properties:
(i) The arguments $W$, etc. denote random variables
inferable from measurement outcomes.
(ii) $P$ results from summing $\OurKD{\rho}( . )$ values
subject to constraints $\delta_{ab}$.

$P(W, W')$ resembles a work distribution
constructed by Solinas and Gasparinetti (S\&G)~\cite{Jordan_chat,Solinas_chat}.
They study fluctuation-relation contexts,
rather than the OTOC.
S\&G propose a definition for
the work performed on a quantum system~\cite{Solinas_15_Full,Solinas_16_Probing}.
The system is coupled weakly to detectors
at a protocol's start and end.
The couplings are represented by constraints
like $\delta_{W ( w_3^* v_2^*) }$ and $\delta_{W' ( w_2 v_1 ) }$.
Suppose that the detectors measure the system's Hamiltonian.
Subtracting the measurements' outcomes
yields the work performed during the protocol.
The distribution over possible work values
is a quasiprobability.
Their quasiprobability is a Husimi $Q$-function,
whereas the OTOC quasiprobability is a KD distribution~\cite{Solinas_16_Probing}.
Related frameworks appear in~\cite{Alonso_16_Thermodynamics,Miller_16_Time,Elouard_17_Role}.
The relationship between those thermodynamics frameworks
and our thermodynamically motivated OTOC framework
merits exploration.

\subsubsection{Weak-measurement and interference schemes
for inferring $\OurKD{\rho}$}
\label{section:Intro_weak_meas}

$\OurKD{\rho}$ can be inferred from weak measurements
and from interference, as shown in~\cite{YungerHalpern_17_Jarzynski}.
Section~\ref{section:MeasSumFromOTOC} shows
how to infer a coarse-graining of $\OurKD{\rho}$
from other OTOC-measurement schemes (e.g.,~\cite{Swingle_16_Measuring}).
We focus mostly on the weak-measurement scheme here.
The scheme is simplified
in Sec.~\ref{section:Measuring}.
First, we briefly review the interference scheme.

The interference scheme in~\cite{YungerHalpern_17_Jarzynski}
differs from other interference schemes for measuring $F(t)$~\cite{Swingle_16_Measuring,Yao_16_Interferometric,Bohrdt_16_Scrambling}:
From the~\cite{YungerHalpern_17_Jarzynski} interference scheme,
one can infer not only $F(t)$, but also $\OurKD{\rho}$.
Time need not be inverted ($H$ need not be negated) in any trial.
The scheme is detailed in Appendix~B of~\cite{YungerHalpern_17_Jarzynski}.
The system is coupled to an ancilla prepared in a superposition
$\frac{1}{ \sqrt{2} } \, ( \ket{0}  +  \ket{1} )$.
A unitary, conditioned on the ancilla, rotates the system's state.
The ancilla and system are measured projectively.
From many trials' measurement data,
one infers $\langle a | \U | b \rangle$,
wherein $\U = U$ or $U^\dag$ and
$a, b =  ( w_\ell, \DegenW_{w_\ell} ), ( v_m , \DegenV_{v_m} )$.
These inner products are multiplied together
to form $\OurKD{\rho}$ [Eq.~\eqref{eq:TAForm}].
If $\rho$ shares neither the $\NondegV$
nor the $\NondegW(t)$ eigenbasis,
quantum-state tomography is needed to infer
$\langle v_1 ,  \DegenV_{v_1}  |  \rho  U^\dag  |
   w_3,  \DegenW_{w_3}  \rangle$.

The weak-measurement scheme
is introduced in Sec. II B 3 of~\cite{YungerHalpern_17_Jarzynski}.
A simple case, in which $\rho = \id / \Dim$,
is detailed in Appendix~A of~\cite{YungerHalpern_17_Jarzynski}.
Recent weak measurements~\cite{Bollen_10_Direct,Lundeen_11_Direct,Lundeen_12_Procedure,Bamber_14_Observing,Mirhosseini_14_Compressive,White_16_Preserving,Piacentini_16_Measuring,Suzuki_16_Observation,Thekkadath_16_Direct},
some used to infer KD distributions,
inspired our weak $\OurKD{\rho}$-measurement proposal.
We review weak measurements,
a Kraus-operator model for measurements,
and the $\OurKD{\rho}$-measurement scheme.

\emph{Review of weak measurements:}
Measurements can alter quantum systems' states.
A weak measurement barely disturbs the measured system's state.
In exchange, the measurement
provides little information about the system.
Yet one can infer much by
performing many trials and processing the outcome statistics.

Extreme disturbances result from strong measurements~\cite{NielsenC10}.
The measured system's state collapses onto a subspace.
For example, let $\rho$ denote the initial state.
Let $\A  =  \sum_a  a  \ketbra{a}{a}$ denote
the measured observable's eigendecomposition.
A strong measurement has a probability $\langle a | \rho | a \rangle$
of projecting $\rho$ onto $\ket{a}$.

One can implement a measurement with an ancilla.
Let $X  =  \sum_x  x \ketbra{x}{x}$ denote an ancilla observable.
One correlates $\A$ with $X$ via an interaction unitary.
Von Neumann modeled such unitaries with
$V_\inter  :=  e^{ - i  \tilde{g} \,  \A \otimes X }$~\cite{vonNeumann_32_Mathematische,Dressel_15_Weak}.
The parameter $\tilde{g}$ signifies the interaction strength.\footnote{
$\A$ and $X$ are dimensionless: To form them,
we multiply dimensionful observables by
natural scales of the subsystems.
These scales are incorporated into $\tilde{g}$.}
An ancilla observable---say,
$Y  =  \sum_y y \ketbra{y}{y}$---is measured strongly.

The greater the $\tilde{g}$, the stronger
the correlation between $\A$ and $Y$.
$\A$ is measured strongly if
it is correlated with $Y$ maximally,
if a one-to-one mapping interrelates
the $y$'s and the $a$'s.
Suppose that the $Y$ measurement yields $y$.
We say that an $\A$ measurement has yielded
some outcome $a_y$.

Suppose that $\tilde{g}$ is small.
$\A$ is correlated imperfectly with $Y$.
The $Y$-measurement outcome, $y$,
provides incomplete information about $\A$.
The value most reasonably attributable to $\A$ remains $a_y$.
But a subsequent measurement of $\A$
would not necessarily yield $a_y$.
In exchange for forfeiting information about $\A$,
we barely disturb the system's initial state.
We can learn more about $\A$ by measuring $\A$ weakly
in each of many trials, then processing measurement statistics.

\emph{Kraus-operator model for measurement:}
Kraus operators~\cite{NielsenC10} model
the system-of-interest evolution
induced by a weak measurement.
Let us choose the following form for $\A$.
Let $V  =  \sum_{ v_\ell,  \DegenV_{v_\ell} }  v_\ell
\ketbra{ v_\ell,  \DegenV_{v_\ell} }{ v_\ell,  \DegenV_{v_\ell} }
=  \sum_{ v_\ell }  v_\ell  \,  \ProjV{v_\ell}$
denote an observable of the system.
$\ProjV{v_\ell}$ projects onto the $v_\ell$ eigenspace.
Let $\A  =  \ketbra{ v_\ell,  \DegenV_{v_\ell} }{ v_\ell,  \DegenV_{v_\ell} }$.
Let $\rho$ denote the system's initial state,
and let $\ket{ D }$ denote the detector's initial state.

Suppose that the $Y$ measurement yields $y$.
The system's state evolves under the Kraus operator
\begin{align}
    \label{eq:Kraus_form00}
    M_y  & =  \langle y | V_\inter | D \rangle  \\
    & =  \langle y |
    \exp \left( - i \tilde{g}
                    \ketbra{ v_\ell,  \DegenV_{v_\ell} }{ v_\ell,  \DegenV_{v_\ell} }
                    \otimes  X  \right)
   | D \rangle      \\
    & = \label{eq:Kraus_form0}
    \langle y | D \rangle  \,  \id
    \nonumber \\ & \quad
    +  \langle y |
        \left( e^{ - i \tilde{g} X }  -  \id  \right)
        | D \rangle  \,
    \ketbra{ v_\ell,  \DegenV_{v_\ell} }{ v_\ell,  \DegenV_{v_\ell} }  \,
\end{align}
as $\rho  \mapsto
\frac{ M_y  \rho  M_y^\dag }{
         \Tr \left(  M_y \rho  M_y^\dag  \right) }  \, .$
The third equation follows from Taylor-expanding the exponential,
then replacing the projector's square with the projector.\footnote{
\label{footnote:MeasPauli}
Suppose that each detector observable
(each of $X$ and $Y$) has
at least as many eigenvalues as $V$.
For example, let $Y$ represent a pointer's position
and $X$ represent the momentum.
Each $X$ eigenstate can be coupled to
one $V$ eigenstate.
$\A$ will equal $V$, and
$V_\inter$ will have the form $e^{ - i \tilde{g} V \otimes X }$.
Such a coupling makes efficient use of the detector:
Every possible final pointer position $y$ correlates with
some $(v_\ell ,  \DegenV_{v_\ell} )$.
Different $\ketbra{ v_\ell,  \DegenV_{v_\ell} }{ v_\ell,  \DegenV_{v_\ell} }$'s
need not couple to different detectors.
Since a weak measurement of $V$ provides information about
one $( v_\ell,  \DegenV_{v_\ell} )$
as well as a weak measurement of
$\ketbra{ v_\ell,  \DegenV_{v_\ell} }{ v_\ell,  \DegenV_{v_\ell} }$ does,
we will sometimes call a weak measurement of
$\ketbra{ v_\ell,  \DegenV_{v_\ell} }{ v_\ell,  \DegenV_{v_\ell} }$
``a weak measurement of $V$,'' for conciseness.

The efficient detector use
trades off against mathematical simplicity,
if $\A$ is not a projector:
Eq.~\eqref{eq:Kraus_form00} fails to simplify to
Eq.~\eqref{eq:Kraus_form0}.
Rather, $V_\inter$ should be approximated to
some order in $\tilde{g}$.
The approximation is (i) first-order
if a KD quasiprobability is being inferred
and (ii) third-order
if the OTOC quasiprobability is being inferred.

If $\A$ is a projector,
Eq.~\eqref{eq:Kraus_form00} simplifies to
Eq.~\eqref{eq:Kraus_form0} even if
$\A$ is degenerate,
e.g., $\A  =  \ProjV{v_\ell}$.
Such an $\A$ assignment will prove natural
in Sec.~\ref{section:Measuring}:
Weak measurements of eigenstates
$\ketbra{ v_\ell,  \DegenV_{v_\ell} }{ v_\ell,  \DegenV_{v_\ell} }$
are replaced with less-resource-consuming
weak measurements of $\ProjV{v_\ell}$'s.

Experimentalists might prefer measuring Pauli operators $\sigma^\alpha$
(for $\alpha = x, y, z$)
to measuring projectors $\Pi$ explicitly.
Measuring Paulis suffices, as
the eigenvalues of $\sigma^\alpha$ map,
bijectively and injectively, onto
the eigenvalues of $\Pi$
(Sec.~\ref{section:Measuring}).
Paulis square to the identity, rather than to themselves:
$\left( \sigma^\alpha \right)^2  =  \id$.
Hence Eq.~\eqref{eq:Kraus_form0} becomes
\begin{align}
   \langle y | \cos \left(  \tilde{g} X \right)  |  D \rangle  \,  \id
   - i \langle y |
     \sin \left( \tilde{g} X \right)
     | D \rangle  \,  \sigma^\alpha  \,   .
\end{align}
}
We reparameterize the coefficients as
$\langle y | D \rangle  \equiv  p(y)  \,  e^{ i \phi }$,
wherein $p(y)  :=  |  \langle y | D \rangle |$, and
$\langle y |  \left( e^{ - i \tilde{g} X }  -  \id  \right) | D \rangle
\equiv   g(y)  \,   e^{ i \phi }$.
An unimportant global phase is denoted by $e^{ i \phi }$.
We remove this global phase from the Kraus operator,
redefining $M_y$ as
\begin{align}
   \label{eq:Kraus_form}
    M_y  & =    \sqrt{ p( y ) }  \:  \id
    +  g( y )  \,  \ketbra{ v_\ell,  \DegenV_{v_\ell} }{
                                 v_\ell,  \DegenV_{v_\ell} }  \, .
\end{align}

The coefficients have the following significances.
Suppose that the ancilla did not couple to the system.
The $Y$ measurement would have
a baseline probability $p( y )$ of outputting $y$.
The dimensionless parameter $g( y )  \in  \mathbb{C}$
is derived from $\tilde{g}$.
We can roughly interpret $M_y$ statistically:
In any given trial, the coupling has a probability $p( y )$
of failing to disturb the system
(of evolving $\rho$ under $\id$)
and a probability $| g( y ) |^2$ of projecting $\rho$ onto
$\ketbra{ v_\ell,  \DegenV_{v_\ell} }{
                                 v_\ell,  \DegenV_{v_\ell} }$.

\emph{Weak-measurement scheme for inferring
the OTOC quasiprobability $\OurKD{\rho}$:}
Weak measurements have been used to measure
KD quasiprobabilities~\cite{Bollen_10_Direct,Lundeen_11_Direct,Lundeen_12_Procedure,Bamber_14_Observing,Mirhosseini_14_Compressive,White_16_Preserving,Suzuki_16_Observation,Thekkadath_16_Direct}.
These experiments' techniques can be applied
to infer $\OurKD{\rho}$ and, from $\OurKD{\rho}$, the OTOC.
Our scheme involves
three sequential weak measurements per trial
(if $\rho$ is arbitrary) or two
[if $\rho$ shares the $\NondegV$ or the $\NondegW(t)$ eigenbasis,
e.g., if $\rho = \id / \Dim$].
The weak measurements alternate with time evolutions
and precede a strong measurement.

We review the general and simple-case protocols.
A projection trick, introduced in Sec.~\ref{section:ProjTrick},
reduces exponentially the number of trials required
to infer about $\OurKD{\rho}$ and $F(t)$.
The weak-measurement and interference protocols
are analyzed in Sec.~\ref{section:Advantages}.
A circuit for implementing the weak-measurement scheme
appears in Sec.~\ref{section:Circuit}.

Suppose that $\rho$ does not share the $\NondegV$
or the $\NondegW(t)$ eigenbasis.
One implements the following protocol, $\Protocol$:
\begin{enumerate}[(1)]
   \item  Prepare $\rho$.
   \item  Measure $\NondegV$ weakly.
             (Couple the system's $\NondegV$ weakly to
             some observable $X$ of a clean ancilla.
             Measure $X$ strongly.)
   \item  Evolve the system forward in time under $U$.
   \item  Measure $\NondegW$ weakly.
             (Couple the system's $\NondegW$ weakly to
             some observable $Y$ of a clean ancilla.
             Measure $Y$ strongly.)
   \item  Evolve the system backward under $U^\dag$.
   \item  Measure $\NondegV$ weakly.
             (Couple the system's $\NondegV$ weakly to
             some observable $Z$ of a clean ancilla.
             Measure $Z$ strongly.)
   \item  Evolve the system forward under $U$.
   \item  Measure $\NondegW$ strongly.
\end{enumerate}
$X$, $Y$, and $Z$ do not necessarily denote Pauli operators.
Each trial yields three ancilla eigenvalues ($x$, $y$, and $z$)
and one $\NondegW$ eigenvalue ($w_3, \DegenW_{w_3}$).
One implements $\Protocol$ many times.
From the measurement statistics, one infers the probability
$\mathscr{P}_\weak ( x ; y ; z ; w_3, \DegenW_{w_3} )$
that any given trial will yield the outcome quadruple
$( x ; y ; z ; w_3, \DegenW_{w_3} )$.

From this probability, one infers the quasiprobability
$\OurKD{\rho} ( v_1,  \DegenV_{v_1} ;  w_2,  \DegenW_{w_2} ;
  v_2,  \DegenV_{v_2}  ;  w_3,  \DegenW_{w_3}  )$.
The probability has the form
\begin{align}
   \label{eq:P_weak}
   & \mathscr{P}_\weak ( x ; y ; z ; w_3, \DegenW_{w_3} )
   =  \bra{ w_3,  \DegenW_{w_3} }  U
   M_z  U^\dag  M_y  U  M_x
   \nonumber \\ & \qquad \qquad \qquad \quad \times
   \rho  M_x^\dag  U^\dag  M_y^\dag  U  M_z^\dag  U^\dag
   \ket{ w_3,  \DegenW_{w_3} } \, .
\end{align}
We integrate over $x$, $y$, and $z$,
to take advantage of all measurement statistics.
We substitute in for the Kraus operators from Eq.~\eqref{eq:Kraus_form},
then multiply out.
The result appears in Eq.~(A7) of~\cite{YungerHalpern_17_Jarzynski}.
Two terms combine into $\propto \Im \LParen \OurKD{\rho} ( . ) \RParen$.
The other terms form independently measurable ``background'' terms.
To infer $\Re \LParen \OurKD{\rho} ( . ) \RParen$,
one performs $\Protocol$ many more times,
using different couplings
(equivalently, measuring different detector observables).
Details appear in Appendix~A of~\cite{YungerHalpern_17_Jarzynski}.

To infer the OTOC, one multiplies
each quasiprobability value
$\OurKD{\rho} ( v_1,  \DegenV_{v_1} ;  w_2,  \DegenW_{w_2}  ;
    v_2,  \DegenV_{v_2}   ;  w_3,  \DegenW_{w_3} )$
by the eigenvalue product
$v_1 w_2 v_2^* w_3^*$.
Then, one sums over the eigenvalues
and the degeneracy parameters:
\begin{align}
   \label{eq:RecoverF1}
   F(t)   & =  \sum_{ ( v_1,  \DegenV_{v_1} ) ,
   ( w_2,  \DegenW_{w_2} ) ,
   ( v_2,  \DegenV_{v_2} ) ,  ( w_3,  \DegenW_{w_3} ) }
   %
   v_1 w_2 v_2^* w_3^*
   \\  \nonumber &  \quad  \times
   \OurKD{\rho} ( v_1,  \DegenV_{v_1}  ;  w_2,  \DegenW_{w_2}  ;
   v_2,  \DegenV_{v_2}  ;  w_3,  \DegenW_{w_3} )   \, .
\end{align}
Equation~\eqref{eq:RecoverF1} follows from Eq.~\eqref{eq:JarzLike}.
Hence inferring the OTOC from the weak-measurement scheme---inspired by Jarzynski's equality---requires a few steps more than
inferring a free-energy difference $\Delta F$ from
Jarzynski's equality~\cite{Jarzynski_97_Nonequilibrium}.
Yet such quasiprobability reconstructions are performed
routinely in quantum optics.

$\W$ and $V$ are local.
Their degeneracies therefore scale with the system size.
If $\Sys$ consists of $\Sites$ spin-$\frac{1}{2}$ degrees of freedom,
$| \DegenW_{w_\ell} |,  | \DegenV_{v_\ell} |  \sim  2^\Sites$.
Exponentially many $\OurKD{\rho}( . )$ values must be inferred.
Exponentially many trials must be performed.
We sidestep this exponentiality in Sec.~\ref{section:ProjTrick}:
One measures eigenprojectors of the degenerate $\W$ and $V$,
rather than of the nondegenerate $\NondegW$ and $\NondegV$.
The one-dimensional
$\ketbra{ v_\ell,  \DegenV_{v_\ell} }{ v_\ell,  \DegenV_{v_\ell} }$
of Eq.~\eqref{eq:Kraus_form0} is replaced with $\ProjV{v_\ell}$.
From the weak measurements, one infers the coarse-grained quasiprobability
$\sum_{\rm degeneracies} \OurKD{\rho}( . )
=:  \SumKD{\rho} ( . )$.
Summing $\SumKD{\rho}( . )$ values yields the OTOC:
\begin{align}
          \label{eq:RecoverF2}
          F(t)  & =  \sum_{ v_1, w_2, v_2, w_3 }
          v_1 w_2 v_2^* w_3^*  \;
          \SumKD{\rho} ( v_1, w_2, v_2, w_3 ) \, .
\end{align}
Equation~\eqref{eq:RecoverF2} follows from
performing the sums over the degeneracy parameters
$\DegenW$ and $\DegenV$ in Eq.~\eqref{eq:RecoverF1}.

Suppose that $\rho$ shares the $\NondegV$ or the $\NondegW(t)$ eigenbasis.
The number of weak measurements reduces to two.
For example, suppose that $\rho$ is
the infinite-temperature Gibbs state $\id / \Dim$.
The protocol $\Protocol$ becomes
\begin{enumerate}[(1)]
   \item  Prepare a $\NondegW$ eigenstate
            $\ket{ w_3,  \DegenW_{w_3} }$.
   \item  Evolve the system backward under $U^\dag$.
   \item  Measure $\NondegV$ weakly.
   \item  Evolve the system forward under $U$.
   \item  Measure $\NondegW$ weakly.
   \item  Evolve the system backward under $U^\dag$.
   \item  Measure $\NondegV$ strongly.
\end{enumerate}

In many recent experiments, only one weak measurement
is performed per trial~\cite{Bollen_10_Direct,Lundeen_11_Direct,Bamber_14_Observing}.
A probability $\mathscr{P}_\weak$ must be approximated
to first order in the coupling constant $g(x)$.
Measuring $\OurKD{\rho}$ requires
two or three weak measurements per trial.
We must approximate $\mathscr{P}_\weak$ to second or third order.
The more weak measurements performed sequentially,
the more demanding the experiment.
Yet sequential weak measurements have been performed recently~\cite{Piacentini_16_Measuring,Suzuki_16_Observation,Thekkadath_16_Direct}.
The experimentalists aimed to reconstruct density matrices
and to measure non-Hermitian operators.
The OTOC measurement provides
new applications for their techniques.

\section{Experimentally measuring $\OurKD{\rho}$
and the coarse-grained $\SumKD{\rho}$}
\label{section:Measuring}

Multiple reasons motivate measurements of
the OTOC quasiprobability $\OurKD{\rho}$.
$\OurKD{\rho}$ is more fundamental than the OTOC $F(t)$,
$F(t)$ results from combining values of $\OurKD{\rho}$.
$\OurKD{\rho}$ exhibits behaviors
not immediately visible in $F(t)$,
as shown in Sections~\ref{section:Numerics} and~\ref{section:Brownian}.
$\OurKD{\rho}$ therefore holds interest in its own right.
Additionally, $\OurKD{\rho}$ suggests
new schemes for measuring the OTOC.
One measures the possible values of $\OurKD{\rho} ( . )$, then
combines the values to form $F(t)$.
Two measurement schemes are detailed in~\cite{YungerHalpern_17_Jarzynski}
and reviewed in Sec.~\ref{section:Intro_weak_meas}.
One scheme relies on weak measurements;
one, on interference.
We simplify, evaluate, and augment these schemes.

First, we introduce a ``projection trick'':
Summing over degeneracies turns one-dimensional projectors
(e.g., $\ketbra{ w_\ell,  \DegenW_{w_\ell} }{ w_\ell,  \DegenW_{w_\ell} }$)
into projectors onto degenerate eigenspaces (e.g., $\ProjW{w_\ell}$).
The \emph{coarse-grained OTOC quasiprobability} $\SumKD{\rho}$ results.
This trick decreases exponentially the number of trials required
to infer the OTOC from weak measurements.\footnote{
The summation preserves interesting properties of the quasiprobability---nonclassical negativity and nonreality, as well as intrinsic time scales.
We confirm this preservation via numerical simulation
in Sec.~\ref{section:Numerics}.}
Section~\ref{section:Advantages} concerns pros and cons of
the weak-measurement and interference schemes
for measuring $\OurKD{\rho}$ and $F(t)$.
We also compare those schemes with
alternative schemes for measuring $F(t)$.
Section~\ref{section:Circuit} illustrates a circuit for implementing
the weak-measurement scheme.
Section~\ref{section:MeasSumFromOTOC}
shows how to infer $\SumKD{\rho}$ not only from
the measurement schemes in Sec.~\ref{section:Intro_weak_meas},
but also with alternative OTOC-measurement proposals
(e.g.,~\cite{Swingle_16_Measuring})
(if the eigenvalues of $\W$ and $V$ are $\pm 1$).

\subsection{The coarse-grained OTOC quasiprobability $\SumKD{\rho}$
and a projection trick}
\label{section:ProjTrick}

$\W$ and $V$ are local.
They manifest, in our spin-chain example, as one-qubit Paulis
that nontrivially transform opposite ends of the chain.
The operators' degeneracies grows exponentially with
the system size $\Sites$:
$| \DegenW_{w_\ell} | ,  \:  | \DegenV_{v_m} |  \sim  2^\Sites$.
Hence the number of $\OurKD{\rho} ( . )$ values grows exponentially.
One must measure exponentially many numbers
to calculate $F(t)$ precisely via $\OurKD{\rho}$.
We circumvent this inconvenience
by summing over the degeneracies in $\OurKD{\rho} ( . ) $,
forming the coarse-grained quasiprobability $\SumKD{\rho} ( . )$.
$\SumKD{\rho} ( . )$ can be measured in numerical simulations,
experimentally via weak measurements,
and (if the eigenvalues of $\W$ and $V$ are $\pm 1$)
experimentally with other $F(t)$-measurement set-ups
(e.g.,~\cite{Swingle_16_Measuring}).

   The \emph{coarse-grained OTOC quasiprobability}
   results from marginalizing $\OurKD{\rho} ( . )$ over its degeneracies:
   \begin{align}
      \label{eq:Sum_KD_def}
      & \SumKD{\rho} ( v_1, w_2, v_2, w_3 )  :=
      \sum_{ \DegenV_{v_1},  \DegenW_{w_2} ,
                  \DegenV_{v_2} ,  \DegenW_{w_3} }
      \nonumber \\ & \quad
      \OurKD{\rho} ( v_1,  \DegenV_{v_1} ;  w_2,  \DegenW_{w_2} ;
                              v_2,  \DegenV_{v_2}  ;  w_3,  \DegenW_{w_3} )   \, .
   \end{align}

Equation~\eqref{eq:Sum_KD_def} reduces to a more practical form.
Consider substituting into Eq.~\eqref{eq:Sum_KD_def}
for $\OurKD{\rho} ( . )$ from Eq.~\eqref{eq:TAForm}.
The right-hand side of Eq.~\eqref{eq:TAForm} equals a trace.
Due to the trace's cyclicality,
the three rightmost factors can be shifted leftward:
\begin{align}
   & \SumKD{\rho} ( v_1, w_2, v_2, w_3 )
   =  \sum_{ \substack{  \DegenV_{v_1},  \DegenW_{w_2} ,  \\
                    \DegenV_{v_2} ,  \DegenW_{w_3} } }
   \Tr \Big( \rho U^\dag
   \ketbra{ w_3,  \DegenW_{w_3} }{ w_3,  \DegenW_{w_3} }  U
   \nonumber \\ & \times
   \ketbra{ v_2,  \DegenV_{v_2} }{  v_2,  \DegenV_{v_2} }
   U^\dag  \ketbra{ w_2,  \DegenW_{w_2} }{  w_2,  \DegenW_{w_2} } U
   \ketbra{ v_1,  \DegenV_{v_1} }{ v_1,  \DegenV_{v_1} }
   \Big) \, .
\end{align}
The sums are distributed throughout the trace:
\begin{align}
   \label{eq:Corase_help1}
   & \SumKD{\rho} ( v_1, w_2, v_2, w_3 )
   =  \Tr \Bigg( \rho
   \Bigg[  U^\dag  \sum_{ \DegenW_{w_3} }
      \ketbra{ w_3,  \DegenW_{w_3} }{ w_3,  \DegenW_{w_3} }  U  \Bigg]
   \nonumber \\ & \times
   \Bigg[  \sum_{ \DegenV_{v_2} }
      \ketbra{ v_2,  \DegenV_{v_2} }{  v_2,  \DegenV_{v_2} }  \Bigg]
   \Bigg[  U^\dag  \sum_{ \DegenW_{w_2} }
      \ketbra{ w_2,  \DegenW_{w_2} }{  w_2,  \DegenW_{w_2} }  U  \Bigg]
   \nonumber \\ & \times
   \Bigg[  \sum_{ \DegenV_{v_1} }
      \ketbra{ v_1,  \DegenV_{v_1} }{ v_1,  \DegenV_{v_1} }  \Bigg]
   \Bigg)  \, .
\end{align}
Define
 \begin{align}
      \label{eq:ProjW}
      \ProjW{w_\ell}  :=  \sum_{ \DegenW_{w_\ell} }
      \ketbra{ w_\ell,  \DegenW_{w_\ell} }{ w_\ell,  \DegenW_{w_\ell} }
   \end{align}
   as the projector onto the $w_\ell$ eigenspace of $\W$,
   \begin{align}
      \label{eq:ProjWt}
      \ProjWt{ w_\ell }  :=  U^\dag  \ProjW{ w_\ell }  U
   \end{align}
   as the projector onto the $w_\ell$ eigenspace of $\W(t)$, and
   \begin{align}
      \label{eq:ProjV}
      \ProjV{v_\ell}  :=  \sum_{ \DegenV_{v_\ell} }
      \ketbra{ v_\ell,  \DegenV_{v_\ell} }{ v_\ell,  \DegenV_{v_\ell} }
   \end{align}
   as the projector onto the $v_\ell$ eigenspace of $V$.
We substitute into Eq.~\eqref{eq:Corase_help1},
then invoke the trace's cyclicality:
\begin{align}
      \label{eq:SumKD_simple2} \boxed{
      \SumKD{\rho} ( v_1, w_2, v_2, w_3 )
      =  \Tr  \Big(  \ProjWt{ w_3 }  \ProjV{v_2}
                            \ProjWt{w_2}  \ProjV{v_1}  \rho  \Big)  }  \, .
\end{align}

Asymmetry distinguishes Eq.~\eqref{eq:SumKD_simple2}
from Born's Rule and from expectation values.
Imagine preparing $\rho$,
measuring $V$ strongly, evolving $\Sys$ forward under $U$,
measuring $\W$ strongly, evolving $\Sys$ backward under $U^\dag$,
measuring $V$ strongly, evolving $\Sys$ forward under $U$,
and measuring $\W$.
The probability of obtaining the outcomes
$v_1, w_2, v_2$, and $w_3$, in that order, is
\begin{align}
   \label{eq:BornCompare}
   \Tr \Big(  & \ProjWt{ w_3 }  \ProjV{ v_2 }   \ProjWt{ w_2 }  \ProjV{ v_1 }
              \rho  \ProjV{ v_1 }  \ProjWt{ w_2 }   \ProjV{ v_2 }  \ProjWt{ w_3 }  \Big) \, .
\end{align}
The operator $\ProjWt{ w_3 }  \ProjV{ v_2 }  \ProjWt{ w_2 }  \ProjV{ v_1 }$
conjugates $\rho$ symmetrically.
This operator multiplies $\rho$ asymmetrically in Eq.~\eqref{eq:SumKD_simple2}.
Hence $\SumKD{\rho}$ does not obviously equal a probability.

Nor does $\SumKD{\rho}$ equal an expectation value.
Expectation values have the form
$\Tr ( \rho \A )$, wherein $\A$ denotes a Hermitian operator.
The operator leftward of the $\rho$ in Eq.~\eqref{eq:SumKD_simple2}
is not Hermitian.
Hence $\SumKD{\rho}$ lacks two symmetries of
familiar quantum objects:
the symmetric conjugation in Born's Rule
and the invariance, under Hermitian conjugation,
of the observable $\A$ in an expectation value.

The right-hand side of Eq.~\eqref{eq:SumKD_simple2} can be measured
numerically and experimentally.
We present numerical measurements in Sec.~\ref{section:Numerics}.
The weak-measurement scheme follows
from Appendix~A of~\cite{YungerHalpern_17_Jarzynski},
reviewed in Sec.~\ref{section:Intro_weak_meas}:
Section~\ref{section:Intro_weak_meas} features projectors onto
one-dimensional eigenspaces, e.g.,
$\ketbra{ v_1, \DegenV_{v_1} }{ v_1, \DegenV_{v_1} }$.
Those projectors are replaced with
$\Pi$'s onto higher-dimensional eigenspaces.
Section~\ref{section:MeasSumFromOTOC} details
how $\SumKD{\rho}$ can be inferred
from alternative OTOC-measurement schemes.

\subsection{Analysis of the quasiprobability-measurement schemes and
comparison with other OTOC-measurement schemes}
\label{section:Advantages}

%
%
\begin{table*}[t]
\begin{center}
\begin{tabular}{|c|c|c|c|c|c|}
   \hline
   &  Weak-
   &  Yunger Halpern
   &  Swingle
   &  Yao
   &  Zhu
   \\
   &  measurement
   &  interferometry
   &  \emph{et al.}
   &  \emph{et al.}
   &  \emph{et al.}
   \\  \hline \hline
        Key tools
   &   Weak
   &   Interference
   &   Interference,
   &   Ramsey interfer.,
   &   Quantum
   \\
   &  measurement
   &
   &  Lochschmidt echo
   &  R\'{e}nyi-entropy meas.
   &  clock
   \\  \hline
        What's inferable
   &   (1) $F(t)$, $\OurKD{\rho}$,
   &   $F^\ParenKB(t)$,  $\OurKD{\rho}^\K$,
   &   $F(t)$
   &   Regulated
   &   $F(t)$
   \\
        from the mea-
   &   \& $\rho$ or
   &   \&  $\rho  \; \:  \forall \K$
   &
   &   correlator
   &
   \\
        surement?
   &   (2) $F(t)$ \& $\SumKD{\rho}$
   &
   &
   &   $F_\reg(t)$
   &
   \\  \hline
        Generality
   &   Arbitrary
   &   Arbitrary
   &   Arbitrary
   &   Thermal:
   &   Arbitrary
   \\
        of $\rho$
   &   $\rho \in \mathcal{D} ( \Hil )$
   &   $\rho \in \mathcal{D} ( \Hil )$
   &   $\rho \in \mathcal{D} ( \Hil )$
   &   $e^{ - H / T } / Z$
   &   $\rho  \in \mathcal{D} (\Hil)$
   \\  \hline
        Ancilla
   &  Yes
   &  Yes
   &   Yes for $\Re \LParen F(t) \RParen$,
   &   Yes
   &   Yes
   \\
      needed?
   &
   &
   &   no for $| F(t) |^2$
   &
   &
   \\  \hline
        Ancilla coup-
   &   No
   &   Yes
   &   Yes
   &   No  
   &   Yes
   \\
       ling global?
   &
   &
   &
   &
   &
   \\ \hline
         How long must
   &   1 weak
   &   Whole
   &   Whole
   &   Whole
   &   Whole
   \\
        ancilla stay
   &   measurement
   &   protocol
   &   protocol
   &   protocol
   &  protocol
   \\
       coherent?
   &
   &
   &
   &
   &
   \\ \hline
       \# time
   &  2
   &  0
   &  1
   &  0
   &  2 (implemented
   \\
      reversals
   &
   &
   &
   &
   &  via ancilla)
   \\ \hline
      \# copies of $\rho$
   &  1
   &  1
   &  1
   &  2
   &  1
   \\
      needed / trial
   &
   &
   &
   &
   &
   \\ \hline
        Signal-to-
   &   To be deter-
   &   To be deter-
   &   Constant
   &   $\sim e^{ - \Sites }$
   &   Constant
   \\
         noise ratio
   &   mined~\cite{Swingle_Resilience}
   &   mined~\cite{Swingle_Resilience}
   &   in $\Sites$
   &
   &   in $\Sites$
   \\ \hline
        Restrictions
   &   Hermitian or
   &   Unitary
   &   Unitary (extension
   &   Hermitian
   &   Unitary
   \\ on $\W$ \& $V$
   &   unitary
   &
   &   to Hermitian possible)
   &   and unitary
   &
   \\ \hline
\end{tabular}
\caption{\caphead{Comparison of our measurement schemes with alternatives:}
This paper focuses on the weak-measurement and interference schemes
for measuring the OTOC quasiprobability $\OurKD{\rho}$
or the coarse-grained quasiprobability $\SumKD{\rho}$.
From $\OurKD{\rho}$ or $\SumKD{\rho}$,
one can infer the OTOC $F(t)$.
These schemes appear in~\cite{YungerHalpern_17_Jarzynski},
are reviewed in Sec.~\ref{section:Intro_weak_meas},
and are assessed in Sec.~\ref{section:Advantages}.
We compare our schemes with
the OTOC-measurement schemes in~\cite{Swingle_16_Measuring,Yao_16_Interferometric,Zhu_16_Measurement}.
More OTOC-measurement schemes appear in~\cite{Danshita_16_Creating,Li_16_Measuring,Garttner_16_Measuring,Tsuji_17_Exact,Campisi_16_Thermodynamics,Bohrdt_16_Scrambling}.
Each row corresponds to a desirable quantity
or to a resource potentially challenging to realize experimentally.
The regulated correlator $F_\reg(t)$ [Eq.~\eqref{eq:RegOTOC_def}]
is expected to behave similarly to $F(t)$~\cite{Maldacena_15_Bound,Yao_16_Interferometric}.
$\mathcal{D} ( \Hil )$ denotes the set of density operators defined on
the Hilbert space $\Hil$.
$\rho$ denotes the initially prepared state.
Target states $\rho_\target$ are never prepared perfectly;
$\rho$ may differ from $\rho_\target$.
Experimentalists can reconstruct $\rho$ by trivially processing
data taken to infer $\OurKD{\rho}$~\cite{YungerHalpern_17_Jarzynski}
(Sec.~\ref{section:TA_retro}).
$F^\ParenKB(t)$ denotes the $\Opsb$-fold OTOC,
which encodes $\Ops  =  2 \Opsb - 1$ time reversals.
The conventional OTOC corresponds to $\Ops = 3$.
The quasiprobability behind $F^\ParenKB(t)$ is
$\OurKD{\rho}^\ParenK$ (Sec.~\ref{section:HigherOTOCs}).
$\Sites$ denotes the system size, e.g., the number of qubits.
The Swingle \emph{et al.} and Zhu \emph{et al.} schemes
have constant signal-to-noise ratios (SNRs)
in the absence of environmental decoherence.
The Yao \emph{et al.} scheme's SNR varies inverse-exponentially with
the system's entanglement entropy, $S_{\rm vN}$.
The system occupies a thermal state $e^{ - H / T } / Z$,
so $S_{\rm vN}  \sim  \log ( 2^\Sites )  =  \Sites$.}
\label{table:Compare}
\end{center}
\end{table*}

Section~\ref{section:Intro_weak_meas} reviews
two schemes for inferring $\OurKD{\rho}$:
a weak-measurement scheme and an interference scheme.
From $\OurKD{\rho}$ measurements,
one can infer the OTOC $F(t)$.
We evaluate our schemes' pros and cons.
Alternative schemes for measuring $F(t)$ have been proposed~\cite{Swingle_16_Measuring,Danshita_16_Creating,Yao_16_Interferometric,Zhu_16_Measurement,Tsuji_17_Exact,Campisi_16_Thermodynamics,Bohrdt_16_Scrambling},
and two schemes have been realized~\cite{Li_16_Measuring,Garttner_16_Measuring}.
We compare our schemes with alternatives,
as summarized in Table~\ref{table:Compare}.
For specificity, we focus on~\cite{Swingle_16_Measuring,Yao_16_Interferometric,Zhu_16_Measurement}.

The weak-measurement scheme augments
the set of techniques and platforms
with which $F(t)$ can be measured.
Alternative schemes rely on interferometry~\cite{Swingle_16_Measuring,Yao_16_Interferometric,Bohrdt_16_Scrambling}, controlled unitaries~\cite{Swingle_16_Measuring,Zhu_16_Measurement},
ultracold-atoms tools~\cite{Danshita_16_Creating,Tsuji_17_Exact,Bohrdt_16_Scrambling},
and strong two-point measurements~\cite{Campisi_16_Thermodynamics}.
Weak measurements, we have shown, belong in the OTOC-measurement toolkit.
Such weak measurements are expected
to be realizable, in the immediate future,
with superconducting qubits~\cite{White_16_Preserving,Hacohen_16_Quantum,Rundle_16_Quantum,Takita_16_Demonstration,Kelly_15_State,Heeres_16_Implementing,Riste_15_Detecting},
trapped ions~\cite{Gardiner_97_Quantum,Choudhary_13_Implementation,Lutterbach_97_Method,Debnath_16_Nature,Monz_16_Realization,Linke_16_Experimental,Linke_17_Experimental},
cavity QED~\cite{Guerlin_07_QND,Murch_13_SingleTrajectories},
ultracold atoms~\cite{Browaeys_16_Experimental},
and perhaps NMR~\cite{Xiao_06_NMR,Dawei_14_Experimental}.
Circuits for weakly measuring qubit systems
have been designed~\cite{Groen_13_Partial,Hacohen_16_Quantum}.
Initial proof-of-principle experiments might not require
direct access to the qubits:
The five superconducting qubits available from IBM,
via the cloud, might suffice~\cite{IBM_QC}.
Random two-qubit unitaries could simulate chaotic Hamiltonian evolution.

In many weak-measurement experiments,
just one weak measurement is performed per trial~\cite{Lundeen_11_Direct,Lundeen_12_Procedure,Bamber_14_Observing,Mirhosseini_14_Compressive}.
Yet two weak measurements
have recently been performed sequentially~\cite{Piacentini_16_Measuring,Suzuki_16_Observation,Thekkadath_16_Direct}.
Experimentalists aimed to ``directly measure general quantum states''~\cite{Lundeen_12_Procedure}
and to infer about non-Hermitian observable-like operators.
The OTOC motivates a new application
of recently realized sequential weak measurements.

Our schemes furnish not only the OTOC $F(t)$,
but also more information:
\begin{enumerate}[(1)]

   \item From the weak-measurement scheme in~\cite{YungerHalpern_17_Jarzynski},
   we can infer the following:
   \begin{enumerate}[(A)]
      \item  The OTOC quasiprobability $\OurKD{\rho}$.
      The quasiprobability is more fundamental than $F(t)$,
as combining $\OurKD{\rho} ( . )$ values yields $F(t)$
[Eq.~\eqref{eq:RecoverF1}].

      \item  The OTOC $F(t)$.

      \item The form $\rho$ of the state prepared.
      Suppose that we wish to evaluate $F(t)$ on
a target state $\rho_\target$.
$\rho_\target$ might be difficult to prepare, e.g., might be thermal.
The prepared state $\rho$ approximates $\rho_\target$.
Consider performing
the weak-measurement protocol $\Protocol$
with $\rho$. One infers $\OurKD{\rho}$.
Summing $\OurKD{\rho}( . )$ values yields the form of $\rho$.
We can assess the preparation's accuracy
without performing tomography independently.
Whether this assessment meets experimentalists' requirements for precision
remains to be seen.
Details appear in Sec.~\ref{section:TA_Coeffs}.
   \end{enumerate}

   \item The weak-measurement protocol $\Protocol$
            is simplified later in this section.
            Upon implementing the simplified protocol,
            we can infer the following information:
   \begin{enumerate}[(A)]
      \item The coarse-grained OTOC quasiprobability $\SumKD{\rho}$.
      Though less fundamental than the fine-grained $\OurKD{\rho}$,
      $\SumKD{\rho}$ implies the OTOC's form
      [Eq.~\eqref{eq:RecoverF2}].
      \item  The OTOC $F(t)$.
   \end{enumerate}

   \item Upon implementing the interferometry scheme in~\cite{YungerHalpern_17_Jarzynski},
   we can infer the following information:
   \begin{enumerate}[(A)]
      \item  The OTOC quasiprobability $\OurKD{\rho}$.
      \item  The OTOC $F(t)$.
      \item  The form of the state $\rho$ prepared.
      \item  All the $\Opsb$-fold OTOCs $F^\ParenKB(t)$,
      which generalize the OTOC $F(t)$.
                $F(t)$ encodes three time reversals.
                $F^\ParenKB(t)$ encodes $\Ops = 2 \Opsb - 1 = 3, 5, \ldots$
                time reversals.
                Details appear in Sec.~\ref{section:HigherOTOCs}.
      \item The quasiprobability $\OurKD{\rho}^\ParenK$
               behind $F^\ParenKB(t)$, for all $\Ops$
               (Sec.~\ref{section:HigherOTOCs}).
   \end{enumerate}

\end{enumerate}

We have delineated the information inferable from
the weak-measurement and interference schemes
for measuring $\OurKD{\rho}$ and $F(t)$.
Let us turn to other pros and cons.

The weak-measurement scheme's ancillas
need not couple to the whole system.
One measures a system weakly by coupling an ancilla to the system,
then measuring the ancilla strongly.
Our weak-measurement protocol requires
one ancilla per weak measurement.
Let us focus, for concreteness, on
an $\SumKD{\rho}$ measurement for a general $\rho$.
The protocol involves three weak measurements and so three ancillas.
Suppose that $\W$ and $V$ manifest as one-qubit Paulis
localized at opposite ends of a spin chain.
Each ancilla need interact with only one site (Fig.~\ref{fig:Circuit}).
In contrast, the ancilla in~\cite{Zhu_16_Measurement}
couples to the entire system.
So does the ancilla in our interference scheme
for measuring $\OurKD{\rho}$.
Global couplings can be engineered in some platforms,
though other platforms pose challenges.
Like our weak-measurement scheme,~\cite{Swingle_16_Measuring} and~\cite{Yao_16_Interferometric}
require only local ancilla couplings.

In the weak-measurement protocol,
each ancilla's state must remain coherent
during only one weak measurement---during
the action of one (composite) gate in a circuit.
The first ancilla may be erased, then reused in the third weak measurement.
In contrast, each ancilla in~\cite{Swingle_16_Measuring,Yao_16_Interferometric,Zhu_16_Measurement}
remains in use throughout the protocol.
The Swingle \emph{et al.} scheme
for measuring $\Re \LParen F(t) \RParen$, too,
requires an ancilla that remains coherent throughout the protocol~\cite{Swingle_16_Measuring}.
The longer an ancilla's ``active-duty'' time,
the more likely the ancilla's state is to decohere.
Like the weak-measurement sheme,
the Swingle \emph{et al.} scheme for measuring $| F (t) |^2$
requires no ancilla~\cite{Swingle_16_Measuring}.

Also in the interference scheme for measuring $\OurKD{\rho}$~\cite{YungerHalpern_17_Jarzynski},
an ancilla remains active throughout the protocol.
That protocol, however, is short: Time need not be reversed in any trial.
Each trial features exactly one $U$ or $U^\dag$, not both.
Time can be difficult to reverse in some platforms, for two reasons.
Suppose that a Hamiltonian $H$ generates a forward evolution.
A perturbation $\varepsilon$ might lead
$- (H + \varepsilon)$ to generate the reverse evolution.
Perturbations can mar
long-time measurements of $F(t)$~\cite{Zhu_16_Measurement}.
Second, systems interact with environments.
Decoherence might not be completely reversible~\cite{Swingle_16_Measuring}.
Hence the lack of a need for time reversal,
as in our interference scheme and in~\cite{Yao_16_Interferometric,Zhu_16_Measurement},
has been regarded as an advantage.

Unlike our interference scheme,
the weak-measurement scheme requires that time be reversed.
Perturbations $\varepsilon$ threaten the weak-measurement scheme
as they threaten the Swingle \emph{et al.} scheme~\cite{Swingle_16_Measuring}.
$\varepsilon$'s might threaten the weak-measurement scheme more,
because time is inverted twice in our scheme.
Time is inverted only once in~\cite{Swingle_16_Measuring}.
However, our error might be expected to have roughly the size
of the Swingle \emph{et al.} scheme's error~\cite{Swingle_Resilience}.
Furthermore, tools for mitigating
the Swingle \emph{et al.} scheme's inversion error
are being investigated~\cite{Swingle_Resilience}.
Resilience of the Swingle \emph{et al.} scheme to decoherence
has been analyzed~\cite{Swingle_16_Measuring}.
These tools may be applied to the weak-measurement scheme~\cite{Swingle_Resilience}.
Like resilience, our schemes' signal-to-noise ratios
require further study.

As noted earlier, as the system size $\Sites$ grows,
the number of trials required to infer $\OurKD{\rho}$ grows exponentially.
So does the number of ancillas required to infer $\OurKD{\rho}$:
Measuring a degeneracy parameter $\DegenW_{w_\ell}$ or $\DegenV_{v_m}$
requires a measurement of each spin.
Yet the number of trials, and the number of ancillas,
required to measure the coarse-grained $\SumKD{\rho}$
remains constant as $\Sites$ grows.
One can infer $\SumKD{\rho}$ from weak measurements
and, alternatively, from other $F(t)$-measurement schemes
(Sec.~\ref{section:MeasSumFromOTOC}).
$\SumKD{\rho}$ is less fundamental than $\OurKD{\rho}$,
as $\SumKD{\rho}$ results from coarse-graining $\OurKD{\rho}$.
$\SumKD{\rho}$, however, exhibits
nonclassicality and OTOC time scales (Sec.~\ref{section:Numerics}).
Measuring $\SumKD{\rho}$ can balance
the desire for fundamental knowledge with practicalities.

The weak-measurement scheme for inferring $\SumKD{\rho}$
can be rendered more convenient.
Section~\ref{section:ProjTrick} describes measurements
of projectors $\Pi$.
Experimentalists might prefer measuring
Pauli operators $\sigma^\alpha$.
Measuring Paulis suffices
for inferring a multiqubit system's $\SumKD{\rho}$:
The relevant $\Pi$ projects onto
an eigenspace of a $\sigma^\alpha$.
Measuring the $\sigma^\alpha$ yields $\pm 1$.
These possible outcomes map bijectively onto
the possible $\Pi$-measurement outcomes.
See Footnote~\ref{footnote:MeasPauli} for mathematics.

Our weak-measurement and interference schemes
offer the advantage of involving general operators.
$\W$ and $V$ must be Hermitian or unitary,
not necessarily one or the other.
Suppose that $\W$ and $V$ are unitary.
Hermitian operators $\GW$ and $\GV$ generate $\W$ and $V$,
as discussed in Sec.~\ref{section:SetUp}.
$\GW$ and $\GV$ may be measured in place of $\W$ and $V$.
This flexibility expands upon the measurement opportunities of, e.g.,~\cite{Swingle_16_Measuring,Yao_16_Interferometric,Zhu_16_Measurement}, 
which require unitary operators.

Our weak-measurement and interference schemes offer leeway
in choosing not only $\W$ and $V$, but also $\rho$.
The state can assume any form  $\rho \in \mathcal{D} ( \Hil )$.
In contrast, infinite-temperature Gibbs states $\rho = \id / \Dim$
were used in~\cite{Li_16_Measuring,Garttner_16_Measuring}.
Thermality of $\rho$ is assumed in~\cite{Yao_16_Interferometric}.
Commutation of $\rho$ with $V$ is assumed in~\cite{Campisi_16_Thermodynamics}.
If $\rho$ shares a $V$ eigenbasis or the $\W(t)$ eigenbasis,
e.g., if $\rho = \id / \Dim$,
our weak-measurement protocol simplifies
from requiring three sequential weak measurements
to requiring two.

\subsection{Circuit for inferring $\SumKD{\rho}$
from weak measurements}
\label{section:Circuit}

Consider a 1D chain $\Sys$ of $\Sites$ qubits.
A circuit implements the weak-measurement scheme
reviewed in Sec.~\ref{section:Intro_weak_meas}.
We exhibit a circuit for measuring $\SumKD{\rho}$.
One subcircuit implements each weak measurement.
These subcircuits result from augmenting
Fig.~1 of~\cite{Dressel_14_Implementing}.

Dressel \emph{et al.} use the \emph{partial-projection formalism},
which we review first.
We introduce notation, then review
the weak-measurement subcircuit of~\cite{Dressel_14_Implementing}.
Copies of the subcircuit are embedded into
our $\SumKD{\rho}$-measurement circuit.

\subsubsection{Partial-projection operators}
\label{section:PartialProj_main}

Partial-projection operators update a state after a measurement
that may provide incomplete information.
Suppose that $\Sys$ begins in a state $\ket{ \psi }$.
Consider performing a measurement that could output $+$ or $-$.
Let $\Pi_+$ and $\Pi_-$ denote
the projectors onto the $+$ and $-$ eigenspaces.
Parameters $p, q  \in  [0, 1]$ quantify the correlation
between the outcome and the premeasurement state.
If $\ket{ \psi }$ is a $+$ eigenstate, the measurement has
a probability $p$ of outputting $+$.
If $\ket{ \psi }$ is a $-$ eigenstate, the measurement has
a probability $q$ of outputting $-$.

Suppose that outcome $+$ obtains.
We update $\ket{ \psi }$ using the \emph{partial-projection operator}
$D_+  :=  \sqrt{p}  \;  \Pi_+
+  \sqrt{1 - q} \;  \Pi_-$:
$\ket{\psi}  \mapsto
   \frac{ D_{+} \ket{ \psi } }{ || D_+ \ket{ \psi } ||^2 } \, .$
If the measurement yields $-$,
we update $\ket{ \psi }$ with
$D_-  :=  \sqrt{1 - p}  \;  \Pi_+
+  \sqrt{q}  \;  \Pi_-$.

The measurement is strong if
$(p, q) = (0, 1)$ or $(1, 0)$.
$D_+$ and $D_-$ reduce to projectors.
The measurement collapses $\ket{ \psi }$ onto an eigenspace.
The measurement is weak if $p$ and $q$ lie close to $\frac{1}{2}$:
$D_\pm$ lies close to the normalized identity, $\frac{\id}{ \Dim }$.
Such an operator barely changes the state.
The measurement provides hardly any information.

We modeled measurements with Kraus operators $M_x$
in Sec.~\ref{section:Intro_weak_meas}.
The polar decomposition of $M_x$~\cite{Preskill_15_Ch3}
is a partial-projection operator.
Consider measuring a qubit's $\sigma^z$.
Recall that $X$ denotes a detector observable.
Suppose that, if an $X$ measurement yields $x$,
a subsequent measurement of the spin's $\sigma^z$
most likely yields $+$.
The Kraus operator $M_x  =  \sqrt{ p(x) }  \:  \id  +  g(x)  \,  \Pi_+$
updates the system's state.
$M_x$ is related to $D_+$ by
   $D_+  =  U_x  \sqrt{ M_x^\dag  M_x }$
for some unitary $U_x$.
The form of $U_x$ depends on the system-detector coupling
and on the detector-measurement outcome.

The imbalance $| p - q |$ can be tuned experimentally.
Our scheme has no need for a nonzero imbalance.
We assume that $p$ equals $q$.

\subsubsection{Notation}
\label{section:Circuit_notation}

Let $\bm{\sigma}  :=  \sigma^x  \,  \hat{ \mathbf{ x } }
+  \sigma^y  \,  \hat{ \mathbf{ y } }
+  \sigma^z \,  \hat{ \mathbf{ z } }$
denote a vector of one-qubit Pauli operators.
The $\sigma^z$ basis serves as
the computational basis in~\cite{Dressel_14_Implementing}.
We will exchange the $\sigma^z$ basis with
the $\W$ eigenbasis, or with the $V$ eigenbasis,
in each weak-measurement subcircuit.

In our spin-chain example, $\W$ and $V$ denote one-qubit Pauli operators
localized on opposite ends of the chain $\Sys$:
$\W  =  \sigma^{ \W }  \otimes  \id^{ \otimes (\Sites - 1) }$,
and $V  =  \id^{ \otimes (\Sites - 1) }  \otimes  \sigma^{ V }$.
Unit vectors $\hat{\W}, \hat{V}  \in  \mathbb{R}^3$
are chosen such that
$\sigma^{ n }  :=  \bm{\sigma}  \cdot  \hat{ \bm{n} }$,
for $n  =  \W,  V$.

The one-qubit Paulis eigendecompose as
$\sigma^{ \W }  =  \ketbra{ + \W }{ + \W }
-  \ketbra{ - \W }{ - \W }$ and
$\sigma^{ V }  =  \ketbra{ +V }{ +V }  -  \ketbra{ -V }{ -V }$.
The whole-system operators eigendecompose as
$\W  =  \ProjW{+}  -  \ProjW{-}$  and
$V  =  \ProjV{+}  -  \ProjV{-}$.
A rotation operator $R_{ n }$ maps
the $\sigma^z$ eigenstates to the $\sigma^{ n }$ eigenstates:
$R_{ n }  \ket{ +z }  =  \ket{ + n }$, and
$R_{ n }  \ket{ -z }  =  \ket{ - n }$.

We model weak $\W$ measurements
with the partial-projection operators
\begin{align}
   \label{eq:PartialProjW}
   & D_+^{ \W }  :=  \sqrt{p_{\W}}  \;  \ProjW{+}
   +  \sqrt{1 - p_{\W} }  \;  \ProjW{-}
   \;  \:  \text{and}  \\  &
   D_-^{ \W }  :=  \sqrt{ 1 - p_{\W} }  \;  \ProjW{+}
   +  \sqrt{p_{\W}}  \;  \ProjW{-}  \, .
\end{align}
The $V$ partial-projection operators
are defined analogously:
\begin{align}
   \label{eq:PartialProjV}
   & D_+^{ V }  :=  \sqrt{p_V}  \;  \ProjV{+}
   +  \sqrt{1 - p_V }  \;  \ProjV{-}
   \;  \:  \text{and}  \\  &
   D_-^{ V }  :=  \sqrt{ 1 - p_V }  \;  \ProjV{+}
   +  \sqrt{p_V}  \;  \ProjV{-}  \, .
\end{align}

\subsubsection{Weak-measurement subcircuit}
\label{section:Weak_subcircuit}

%
%
\begin{figure}[h]
\centering
\begin{subfigure}{0.5\textwidth}
\centering
\includegraphics[width=.95\textwidth]{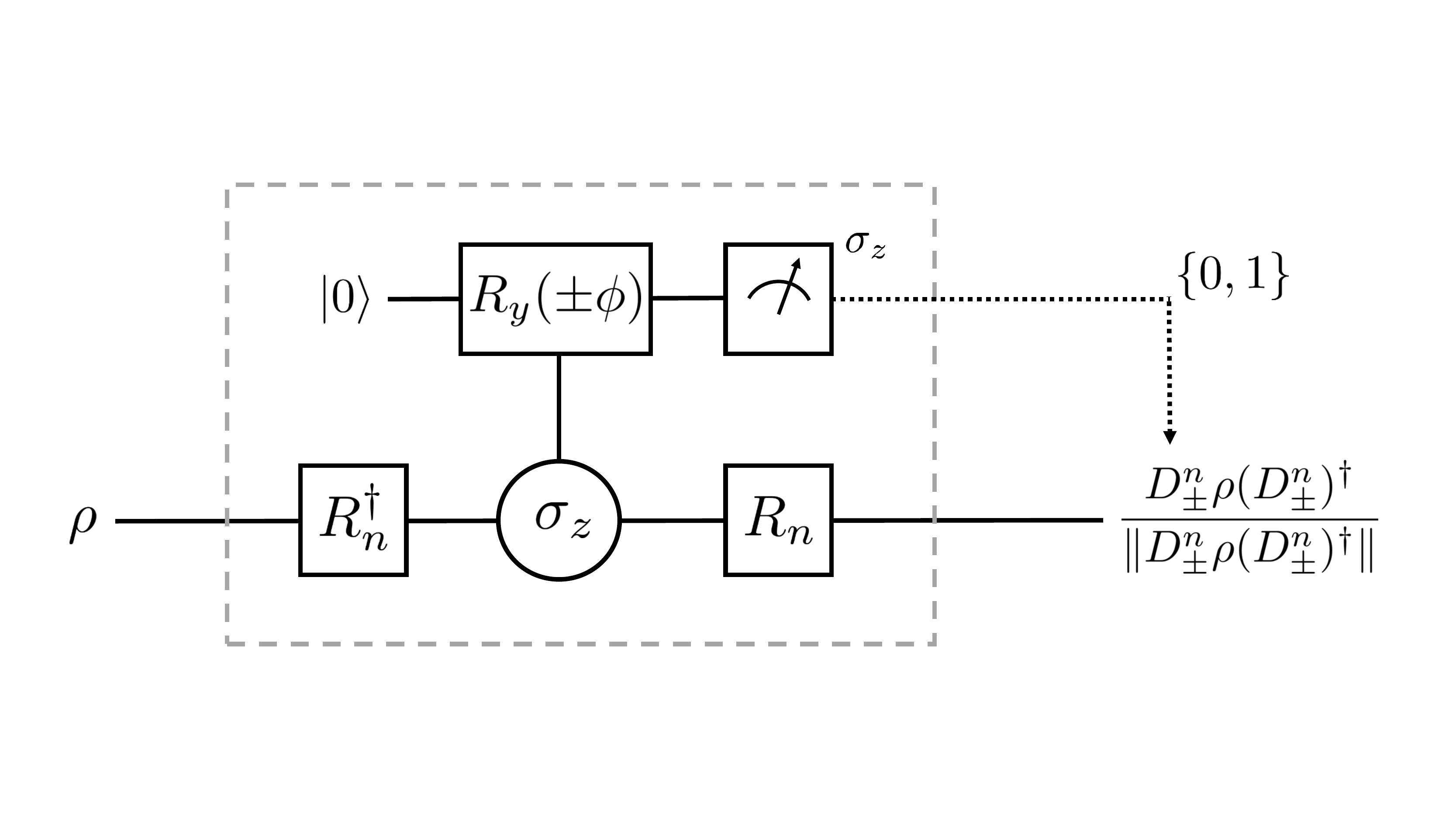}
\caption{}
\label{fig:D_subcircuit}
\end{subfigure}
\begin{subfigure}{.5\textwidth}
\centering
\includegraphics[width=.95\textwidth]{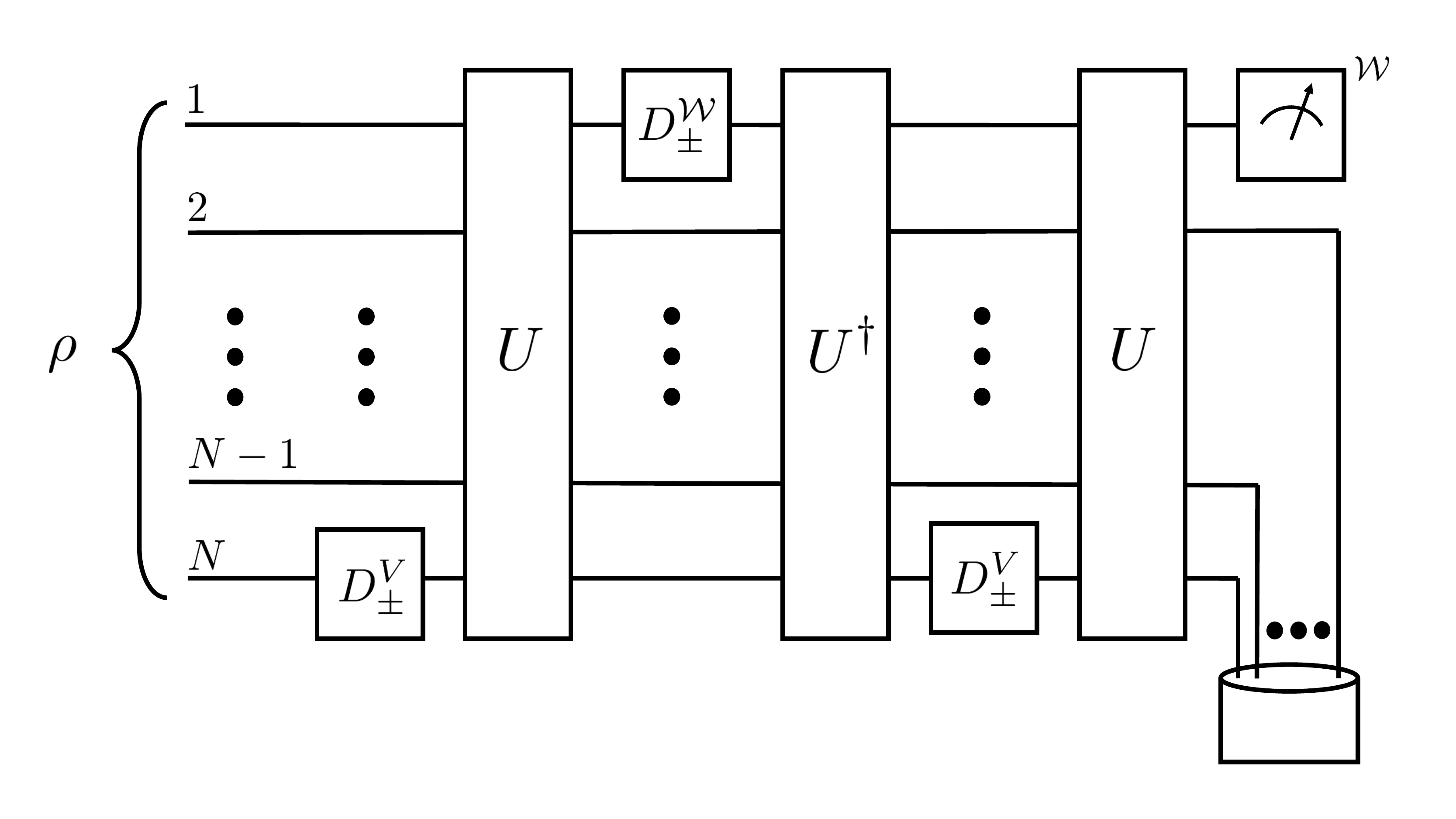}
\caption{}
\label{fig:Full_circuit}
\end{subfigure}
\caption{\caphead{Quantum circuit for inferring
the coarse-grained OTOC quasiprobability $\SumKD{\rho}$
from weak measurements:}
We consider a system of $\Sites$ qubits
prepared in a state $\rho$.
The local operators $\W = \sigma^{ \W }  \otimes
\id^{ \otimes (\Sites - 1) }$ and
$V  = \id^{ \otimes (\Sites - 1) }  \otimes
\sigma^{ V }$
manifest as one-qubit Paulis.
Weak measurements can be used to infer
the coarse-grained quasiprobability $\SumKD{\rho}$.
Combining values of $\SumKD{\rho}$ yields the OTOC $F(t)$.
Figure~\ref{fig:D_subcircuit} depicts a subcircuit
used to implement a weak measurement
of $n = \W$ or $V$.
An ancilla is prepared in a fiducial state $\ket{0}$.
A unitary $R_{ n }^\dag$ rotates
the qubit's $\sigma^{ n }$ eigenbasis
into its $\sigma^z$ eigenbasis.
$R_y ( \pm \phi )$ rotates the ancilla's state counterclockwise
about the $y$-axis through a small angle $\pm \phi$,
controlled by the system's $\sigma^z$.
The angle's smallness guarantees the measurement's weakness.
$R_{ n }$ rotates the system's $\sigma^z$ eigenbasis
back into the $\sigma^{ n }$ eigenbasis.
The ancilla's $\sigma^z$ is measured strongly.
The outcome, $+$ or $-$, dictates
which partial-projection operator $D_{\pm}^n$
updates the state.
Figure~\ref{fig:Full_circuit} shows the circuit
used to measure $\SumKD{\rho}$.
Three weak measurements, interspersed with
three time evolutions ($U$, $U^\dag$, and $U$),
precede a strong measurement.
Suppose that the initial state, $\rho$, commutes with $\W$ or $V$,
e.g., $\rho = \id / \Dim$.
Figure~\ref{fig:Full_circuit} requires only two weak measurements.
}
\label{fig:Circuit}
\end{figure}

Figure~\ref{fig:D_subcircuit} depicts a subcircuit
for measuring $n = \W$ or $V$ weakly.
To simplify notation, we relabel $p_{n}$ as $p$.
Most of the subcircuit appears in Fig.~1 of~\cite{Dressel_14_Implementing}.
We set the imbalance parameter $\epsilon$ to 0.
We sandwich Fig.~1 of~\cite{Dressel_14_Implementing}
between two one-qubit unitaries.
The sandwiching interchanges the computational basis
with the $n$ eigenbasis.

The subcircuit implements the following algorithm:
\begin{enumerate}[(1)]

   \item   Rotate the $n$ eigenbasis into the $\sigma^z$ eigenbasis,
   using $R_{ n }^\dag$.

   \item   Prepare an ancilla in a fiducial state $\ket{0}  \equiv   \ket{ +z }$.

   \item   Entangle $\Sys$ with the ancilla via a $Z$-controlled-$Y$:
   If $\Sys$ is in state $\ket{0}$, rotate the ancilla's state
   counterclockwise (CCW) through a small angle
   $\phi  \ll  \frac{\pi}{2}$ about the $y$-axis.
   Let $R_y ( \phi )$ denote the one-qubit unitary
   that implements this rotation.
   If $\Sys$ is in state $\ket{1}$, rotate the ancilla's state CCW
   through an angle $- \phi$, with $R_y ( - \phi )$.

   \item   Measure the ancilla's $\sigma^z$.
   If the measurement yields outcome $+$,
   $D_+$ updates the system's state;
   and if $-$, then $D_-$.

   \item   Rotate the $\sigma^z$ eigenbasis
   into the $n$ eigenbasis, using $R_{ n }$.

\end{enumerate} \noindent
The measurement is weak because $\phi$ is small.
Rotating through a small angle precisely
can pose challenges~\cite{White_16_Preserving}.

\subsubsection{Full circuit for weak-measurement scheme}
\label{section:Full_circuit}

Figure~\ref{fig:Full_circuit} shows the circuit
for measuring $\SumKD{\rho}$.
The full circuit contains three weak-measurement subcircuits.
Each ancilla serves in only one subcircuit.
No ancilla need remain coherent throughout the protocol,
as discussed in Sec.~\ref{section:Advantages}.
The ancilla used in the first $V$ measurement
can be recycled for the final $V$ measurement.

The circuit simplifies in a special case.
Suppose that $\rho$ shares an eigenbasis with $V$ or with $\W(t)$,
e.g., $\rho = \id / \Dim$.
Only two weak measurements are needed,
as discussed in Sec.~\ref{section:Intro_weak_meas}.

We can augment the circuit to measure $\OurKD{\rho}$,
rather than $\SumKD{\rho}$:
During each weak measurement,
every qubit will be measured.
The qubits can be measured individually:
The $\Sites$-qubit measurement can be
a product of local measurements.
Consider, for concreteness, the first weak measurement.
Measuring just qubit $\Sites$ would yield
an eigenvalue $v_1$ of $V$.
We would infer whether qubit $\Sites$ pointed upward or downward
along the $\hat{V}$ axis.
Measuring all the qubits would yield
a degeneracy parameter $\DegenV_{v_1}$.
We could define $\DegenV_{v_\ell}$ as encoding
the $\hat{V}$-components of
the other $\Sites - 1$ qubits' angular momenta.

%
%
%
\subsection{How to infer $\SumKD{\rho}$
from other OTOC-measurement schemes}
\label{section:MeasSumFromOTOC}

$F(t)$ can be inferred, we have seen, from
the quasiprobability $\OurKD{\rho}$
and from the coarse-grained $\SumKD{\rho}$.
$\SumKD{\rho}$ can be inferred from $F(t)$-measurement schemes,
we show, if the eigenvalues of $\W$ and $V$ equal $\pm 1$.
We assume, throughout this section, that they do.
The eigenvalues equal $\pm 1$ if $\W$ and $V$ are Pauli operators.

The projectors~\eqref{eq:ProjW} and~\eqref{eq:ProjV}
can be expressed as
\begin{align}
   \label{eq:ProjW2}
   \ProjW{ w_\ell }  =  \frac{1}{2} ( \id  +  w_\ell \W )
   \quad \text{and} \quad
   \ProjV{ v_\ell }  =  \frac{1}{2}  ( \id  +  v_\ell V ) \, .
\end{align}
Consider substituting from Eqs.~\eqref{eq:ProjW2}
into Eq.~\eqref{eq:SumKD_simple2}.
Multiplying out yields sixteen terms.
If $\expval{ . }  :=  \Tr (  .  \, . )$,
\begin{align}
   \label{eq:ProjTrick}
   & \SumKD{\rho} ( v_1, w_2, v_2, w_3 )
   =  \frac{1}{16}  \Big[  1  +  ( w_2 + w_3 )  \expval{ \W (t) }
   \nonumber \\ &
   +  ( v_1  +  v_2 )  \expval{ V }  +  w_2 w_3  \expval{ \W^2 (t) }
   +  v_1 v_2  \expval{ V^2 }
   \nonumber \\ &
   +  ( w_2 v_1  +  w_3 v_1  +  w_3 v_2 )  \expval{ \W (t) V }
   +  w_2  v_2  \expval{ V \W(t) }
   \nonumber \\ &
   +  w_2 w_3 v_1  \expval{ \W^2 (t) V }
   +  w_3  v_1  v_2  \expval{ \W(t)  V^2  }
   \nonumber \\ &
   +  w_2  w_3  v_2  \expval{ \W(t)  V  \W(t) }
   +  w_2  v_1  v_2  \expval{ V \W(t)  V }
   \nonumber \\ &
   +  w_2  w_3  v_1  v_2  \,  F(t)  \Big]  \, .
\end{align}
If $\W(t)$ and $V$ are unitary, they square to $\id$.
Equation~\eqref{eq:ProjTrick} simplifies to
\begin{align}
   \label{eq:ProjTrick2}
   & \SumKD{\rho} ( v_1, w_2, v_2, w_3 )
   =  \frac{1}{16}  \Big\{  ( 1  +  w_2 w_3  +  v_1 v_2 )
   \nonumber  \\ &
   +  [ w_2  +  w_3 ( 1  +  v_1 v_2 ) ]  \expval{ \W(t) }
   +  [  v_1 ( 1 + w_2 w_3 )  +  v_2 ]  \expval{ V }
   \nonumber \\ &
   +  ( w_2  v_1  +  w_3  v_1  +  w_3  v_2 )  \expval{ \W(t) V }
   +  w_2 v_2  \expval{ V \W(t) }
    \nonumber \\ &
   +  w_2  w_3 v_2  \expval{ \W (t) V \W(t) }
   +  w_2  v_1 v_2  \expval{ V \W(t) V }
   \nonumber \\ &+  w_2 w_3 v_1 v_2  \, F (t)  \Big\}  \, .
\end{align}

The first term is constant.
The next two terms are single-observable expectation values.
The next two terms are two-point correlation functions.
$\expval{ V \W(t)  V }$ and $\expval{ \W(t)  V  \W(t) }$
are time-ordered correlation functions.
$F(t)$ is the OTOC.
$F(t)$ is the most difficult to measure.
If one can measure it, one likely has the tools to infer $\SumKD{\rho}$.
One can measure every term, for example,
using the set-up in~\cite{Swingle_16_Measuring}.

\section{Numerical simulations}
\label{section:Numerics}

We now study the OTOC quasiprobability's physical content in two simple models. In this section, we study a geometrically local 1D model, an Ising chain with transverse and longitudinal fields. In Sec.~\ref{section:Brownian}, we study a geometrically nonlocal model known as the \emph{Brownian-circuit model.} This model effectively has a time-dependent Hamiltonian.

We compare the physics of $\SumKD{\rho}$ with that of the OTOC.
The time scales inherent in $\SumKD{\rho}$,
as compared to the OTOC's time scales, particularly interest us.
We study also nonclassical behaviors---negative and nonreal values---of
$\SumKD{\rho}$.
Finally, we find a parallel with classical chaos:
The onset of scrambling breaks a symmetry.
This breaking manifests in bifurcations of $\SumKD{\rho}$,
reminiscent of pitchfork diagrams.

The Ising chain is defined on a Hilbert space of $\Sites$ spin-$\frac{1}{2}$ degrees of freedom.
The total Hilbert space has dimensionality $\Dim = 2^\Sites$.
The single-site Pauli matrices are labeled $\{\sigma^x_i,\sigma^y_i,\sigma^z_i\}$, for $i=1,...,\Sites$. The Hamiltonian is
\begin{align}\label{eq:IsingH}
  H = - J \sum_{i=1}^{\Sites-1} \sigma_i^z \sigma_{i+1}^z - h \sum_{i=1}^{\Sites} \sigma_i^z - g \sum_{i=1}^{\Sites} \sigma^x_i  \, .
\end{align}
The chain has open boundary conditions. Energies are measured in units of $J$.
Times are measured in units of $1/J$. The interaction strength is thus set to one, $J=1$, henceforth. We numerically study this model for $\Sites=10$ by exactly diagonalizing $H$. This system size suffices for probing
the quasiprobability's time scales.
However, $\Sites = 10$ does not necessarily illustrate the thermodynamic limit.

When $h = 0$, this model is integrable and can be solved with noninteracting-fermion variables. When $h\neq 0$, the model appears to be reasonably chaotic. These statements' meanings are clarified in the data below.
As expected, the quasiprobability's qualitative behavior
is sensitive primarily to
whether $H$ is integrable,
as well as to the initial state's form.
We study two sets of parameters,
\begin{align}
  \text{Integrable:} \;  & h=0,  \:  g=1.05
  \quad \text{and} \nonumber \\
  \text{Nonintegrable:}  \; & h=.5,   \:  g=1.05  \, .
\end{align}
We study several classes of initial states $\rho$,
including thermal states, random pure states, and product states.

For $\W$ and $V$, we choose single-Pauli operators
that act nontrivially on just the chain's ends.
We illustrate with $\mathcal{W} = \sigma_1^{x}$ or $\mathcal{W}
= \sigma_1^z$
and $V= \sigma_\Sites^x$ or $\sigma_\Sites^z$.
These operators are unitary and Hermitian.
They square to the identity, enabling us to use Eq.~\eqref{eq:ProjTrick2}.
We calculate the coarse-grained quasiprobability directly:
\begin{align}
   \label{eq:SumKD_Num}
   \SumKD{\rho}( v_1 , w_2 , v_2 , w_3 )
   = \text{Tr}\left( \rho \Pi^{\mathcal{W}(t)}_{w_3} \Pi^V_{v_2} \Pi^{\mathcal{W}(t)}_{w_2}    \Pi^{V}_{v_1} \right)  \, .
\end{align}
For a Pauli operator $\Oper$,
$\Pi^\Oper_a = \frac{1}{2} \: ( 1 + a \Oper )$ projects onto
the $a \in \{1,-1\}$ eigenspace.
We also compare the quasiprobability with the OTOC,
Eq.~\eqref{eq:RecoverF2}.

$F(t)$ deviates from one at roughly
the time needed for information to propagate
from one end of the chain to the other.
This onset time, which up to a constant shift is also approximately the scrambling time, lies approximately between $t=4$ and $t=6$,
according to our the data.
The system's length and the butterfly velocity $v_{\rm B}$
set the scrambling time (Sec.~\ref{section:OTOC_review}).
Every term in the Hamiltonian~\eqref{eq:IsingH} is order-one.
Hence $v_{\rm B}$ is expected to be order-one, too.
In light of our spin chain's length,
the data below are all consistent with a $v_{\rm B}$ of approximately two.

\subsection{Thermal states}

We consider first thermal states $\rho \propto e^{-H/T}$.
Data for the infinite-temperature ($T = \infty$) state,
with $\mathcal{W} = \sigma_1^z$, $V = \sigma_\Sites^z$,
and nonintegrable parameters, appear in
Figures \ref{fig:TInf_zz_F}, \ref{fig:TInf_zz_AR}, and \ref{fig:TInf_zz_AI}.
The legend is labeled such that $abcd$ corresponds to
$w_3 = (-1)^a$, $v_2 = (-1)^b$, $w_2 = (-1)^c$, and $v_1 = (-1)^d$. This labelling corresponds to the order in which the operators appear in Eq.~\eqref{eq:SumKD_Num}.

Three behaviors merit comment.
Generically, the coarse-grained quasiprobability is a complex number:
$\SumKD{\rho}( . )  \in  \mathbb{C}$.
However, $\SumKD{( \id / \Dim) }$ is real.
The imaginary component $\Im \left(  \SumKD{( \id / \Dim) }  \right)$
might appear nonzero in Fig.~\ref{fig:TInf_zz_AI}.
Yet $\Im \left(  \SumKD{( \id / \Dim) }  \right)  \leq  10^{ -16 }$.
This value equals zero, to within machine precision.
The second feature to notice is that
the time required for $\SumKD{ ( \id / \Dim) }$ to deviate from its initial value
equals approximately the time required for the OTOC to deviate from its initial value. Third, although $\SumKD{ ( \id / \Dim) }$ is real, it is negative and hence nonclassical for some values of its arguments.

What about lower temperatures?
Data for the $T=1$ thermal state are shown in Figures \ref{fig:T1_zz_F}, \ref{fig:T1_zz_AR}, and \ref{fig:T1_zz_AI}.
The coarse-grained quasiprobability is no longer real.
Here, too, the time required for $\SumKD{\rho}$ to deviate significantly from its initial value is comparable with the time scale of changes in $F(t)$.
This comparability characterizes
the real and imaginary parts of $\SumKD{\rho}$.
Both parts oscillate at long times. In the small systems considered here, such oscillations can arise from finite-size effects,
including the energy spectrum's discreteness.
With nonintegrable parameters, this model has an energy gap $\Delta_{\Sites=10} = 2.92$ above the ground state.
The temperature $T=1$ is smaller than the gap.
Hence lowering $T$ from $\infty$ to 1
brings the thermal state close to the ground state.

What about long-time behavior?
At infinite temperature, $\SumKD{ (\id / \Dim ) }$ approaches a limiting form
after the scrambling time
but before any recurrence time.
Furthermore, $\SumKD{ (\id / \Dim ) }$ can approach one of
only a few possible limiting values,
depending on the function's arguments.
This behavior follows from the terms in Eq.~\eqref{eq:ProjTrick2}.
At infinite temperature, $\langle \mathcal{W} \rangle = \langle V \rangle = 0$.
Also the 3-point functions vanish, due to the trace's cyclicity.
We expect the nontrivial 2- and 4-point functions
to be small at late times.
(Such smallness is visible in the 4-point function in Fig.~\ref{fig:TInf_zz_F}.) Hence Eq.~\eqref{eq:ProjTrick2} reduces as
\begin{align}\label{eq:infiniteTlatetime}
  \SumKD{\rho}( v_1 , w_2 , v_2 , w_3 )
  \underbrace{\longrightarrow}_{t\rightarrow \infty} \frac{1 + w_2 w_3 + v_1 v_2}{16}  \, .
\end{align}

According to Eq.~\eqref{eq:infiniteTlatetime},
the late-time values of $\SumKD{ (\id / \Dim ) }$
should cluster around $3/16$, $1/16$, and $-1/16$.
This expectation is roughly consistent with Fig.~\ref{fig:TInf_zz_AR},
modulo the upper lines' bifurcation.

A bifurcation of $\SumKD{\rho}$ signals
the breaking of a symmetry at the onset of scrambling. Similarly, pitchfork plots signal the breaking of a symmetry
in classical chaos~\cite{Strogatz_00_Non}.
The symmetry's mathematical form follows from Eq.~\eqref{eq:ProjTrick2}.
At early times, $\W(t)$ commutes with $V$, and $F(t) \approx 1$.
Suppose, for simplicity, that $\rho = \id / \Dim$.
The expectation values $\expval{ \W(t) }$ and $\expval{V}$ vanish,
because every Pauli has a zero trace.
Equation~\eqref{eq:ProjTrick2} becomes
\begin{align}
   \label{eq:ProjTrick_early}
   & \SumKD{\rho} ( v_1 , w_2 , v_2 , w_3 )
   =  \frac{1}{16}  \Big[ (1 + w_2 w_3  +  v_1 v_2  +  w_2 w_3 v_1 v_2 )
   \nonumber \\ & \qquad \qquad \qquad
   +  ( w_2  +  w_3 )  ( v_1  +  v_2 )  \expval{ \W(t)  V }  \Big]  \, .
\end{align}

Suppose that $w_2 = - w_3$ and/or $v_1 = - v_2$,
as in the lower lines in Fig.~\ref{fig:TInf_zz_AR}.
$\SumKD{\rho} (.)$ reduces to the constant
\begin{align}
   \label{eq:Early_square}
   & \frac{1}{16}  ( 1 + w_2 w_3  +  v_1 v_2  +  w_2 w_3 v_1 v_2 )
   \\ \nonumber &
   =  \frac{1}{32}  \Big[  ( 1 + w_2 w_3  +  v_1 v_2  )^2
   -  ( w_2 w_3 )^2  -  ( v_1 v_2 )^2  +  1  \Big] \, .
\end{align}
The right-hand side depends on the eigenvalues $w_\ell$ and $v_m$
only through squares.
$\SumKD{\rho} ( . )$ remains invariant
under the interchange of $w_2$ with $w_3$,
under the interchange of $v_1$ with $v_2$,
under the simultaneous negations of $w_2$ and $w_3$,
and under the simultaneous negations of $v_1$ and $v_2$.
These symmetries have operational significances:
$\OurKD{\rho}$ remains constant under permutations and negations
of measurement outcomes in
the weak-measurement scheme (Sec.~\ref{section:Intro_weak_meas}).
Symmetries break as the system starts scrambling:
$F(t)$ shrinks, shrinking the final term in Eq.~\eqref{eq:Early_square}.
$\SumKD{\rho}$ starts depending not only on
squares of $w_\ell$-and-$v_m$ functions,
but also on the eigenvalues individually.

Whereas the shrinking of $F(t)$ bifurcates
the lower lines in Fig.~\ref{fig:TInf_zz_AR},
the shrinking does not bifurcate the upper lines.
The reason is that each upper line corresponds to $w_2 w_3  =  v_1 v_2  =  1$.
[At early times, $| F(t) |$ is small enough that
any $F(t)$-dependent correction would fall within the lines' widths.]
Hence the final term in Eq.~\eqref{eq:ProjTrick_early} is proportional to
$\pm \expval{ \mathcal{W}(t) V }$.
This prediction is consistent with the observed splitting.
The $\expval{ \mathcal{W}(t) V }$ term does not split the lower lines:
Each lower line satisfies $w_2 = - w_3$ and/or $v_1 = - v_2$.
Hence the $\expval{ \mathcal{W}(t) V }$ term vanishes.
We leave as an open question
whether these pitchforks can be understood in terms of equilibria,
like classical-chaos pitchforks~\cite{Strogatz_00_Non}.

In contrast with the $T=\infty$ data, the $T=1$ data
oscillate markedly at late times
(after the quasiprobability's initial sharp change).
We expect these oscillations to decay to zero at late times,
if the system is chaotic, in the thermodynamic limit.
Unlike at infinite temperature, $\mathcal{W}$ and $V$
can have nonzero expectation values.
But, if all nontrivial connected correlation functions have decayed,
Eq.~\eqref{eq:ProjTrick2} still implies a simple dependence on
the $w_\ell$ and $v_m$ parameters at late times.

Finally, Figures \ref{fig:TInf_zz_int_F} and \ref{fig:TInf_zz_int_AR} show
the coarse-grained quasiprobability at infinite temperature,
$\SumKD{ (\id / \Dim ) }$, with integrable parameters.
The imaginary part remains zero, so we do not show it.
The difference from the behavior in
Figures \ref{fig:TInf_zz_F} and \ref{fig:TInf_zz_AR}
(which shows $T = \infty$, nonintegrable-$H$ data) is obvious.
Most dramatic is the large revival that occurs
at what would, in the nonintegrable model, be a late time.
Although this is not shown, the quasiprobability depends significantly
on the choice of operator.
This dependence is expected, since different Pauli operators have different degrees of complexity in terms of the noninteracting-fermion variables.


\begin{figure}
\begin{center}
\includegraphics[width=.49\textwidth]{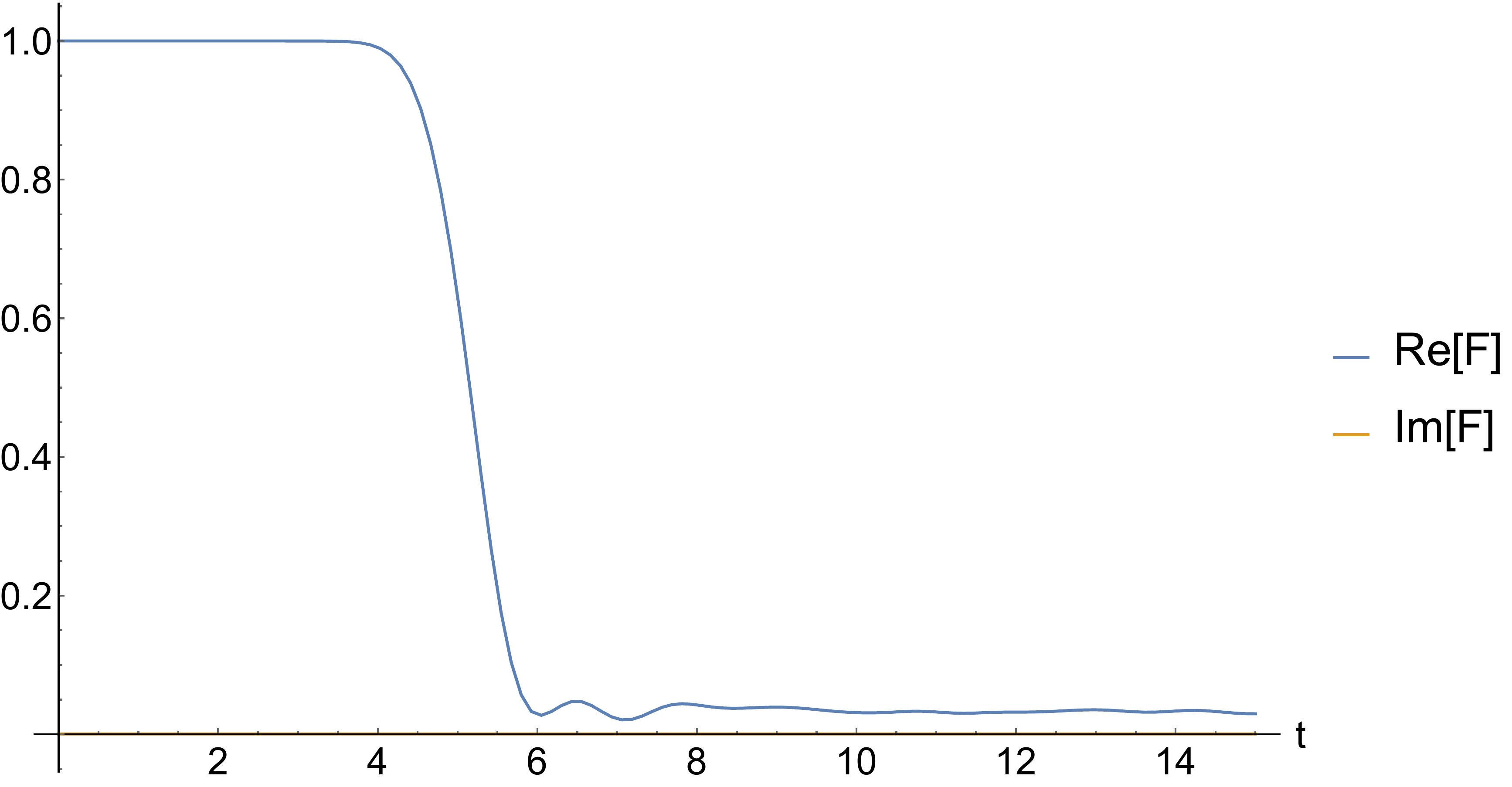}
\end{center}
\caption{Real (upper curve) and imaginary (lower curve) parts of $F(t)$ as a function of time. $T=\infty$ thermal state. Nonintegrable parameters, $\Sites=10$, $\mathcal{W}=\sigma_1^z$, $V=\sigma_\Sites^z$. }
\label{fig:TInf_zz_F}
\end{figure}

\begin{figure}
\begin{center}
\includegraphics[width=.49\textwidth]{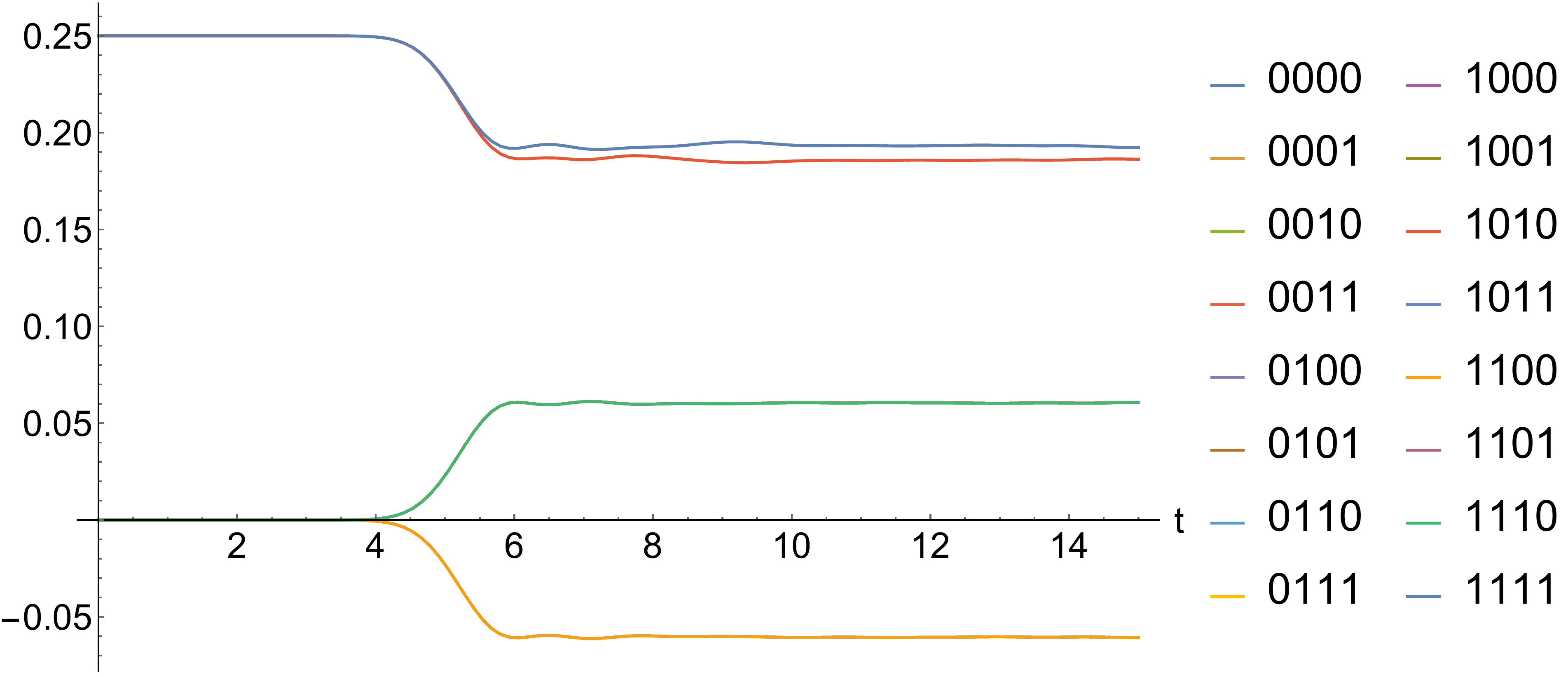}
\end{center}
\caption{Real part of $\SumKD{\rho}$ as a function of time. $T=\infty$ thermal state. Nonintegrable parameters, $\Sites=10$, $\mathcal{W}=\sigma_1^z$, $V=\sigma_\Sites^z$. There are many degeneracies. The upper curves include $0000$ and $1010$, while the top of the lower pitchfork includes $1110$ and the bottom of the lower pitchfork includes $0001$.}
\label{fig:TInf_zz_AR}
\end{figure}

\begin{figure}
\begin{center}
\includegraphics[width=.49\textwidth]{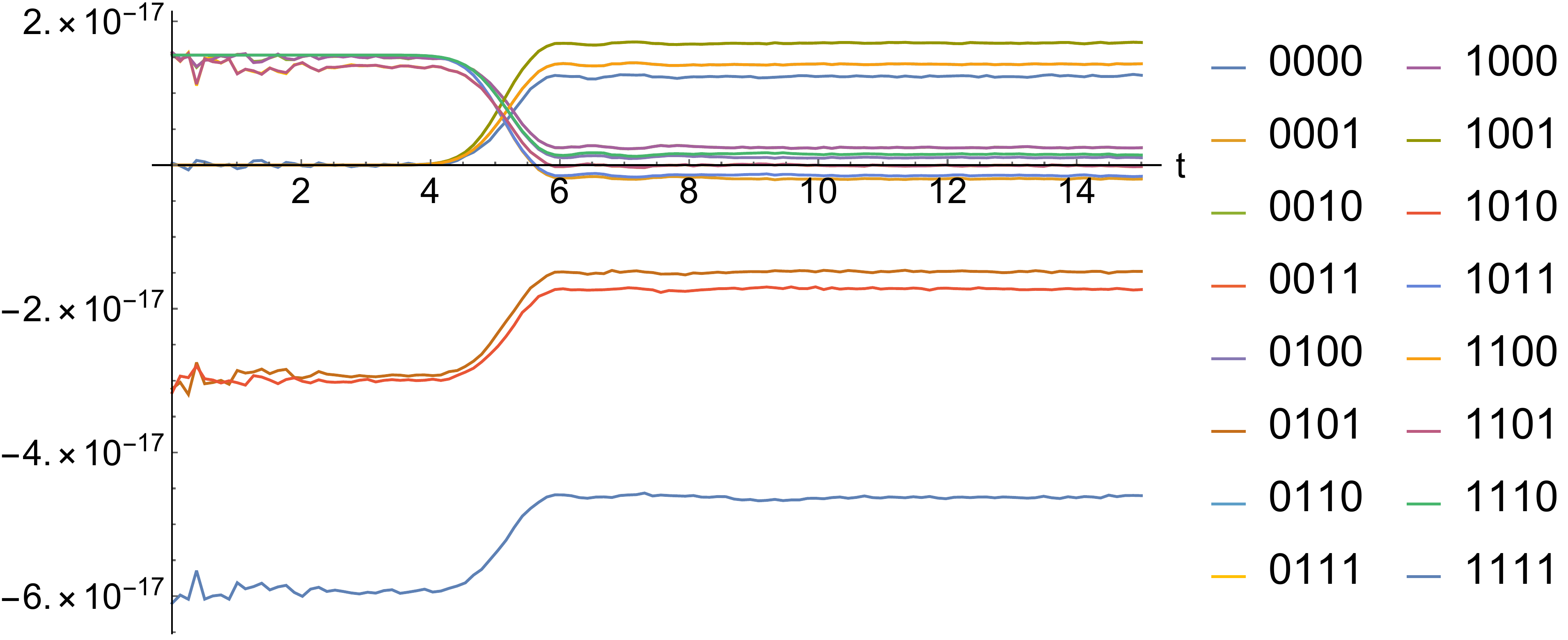}
\end{center}
\caption{Imaginary part of $\SumKD{\rho}$ as a function of time. $T=\infty$ thermal state. Nonintegrable parameters, $\Sites=10$, $\mathcal{W}=\sigma_1^z$, $V=\sigma_\Sites^z$.
To within machine precision, $\Im \left(  \OurKD{\rho}  \right)$ vanishes
for all values of the arguments.}
\label{fig:TInf_zz_AI}
\end{figure}


\begin{figure}
\begin{center}
\includegraphics[width=.49\textwidth]{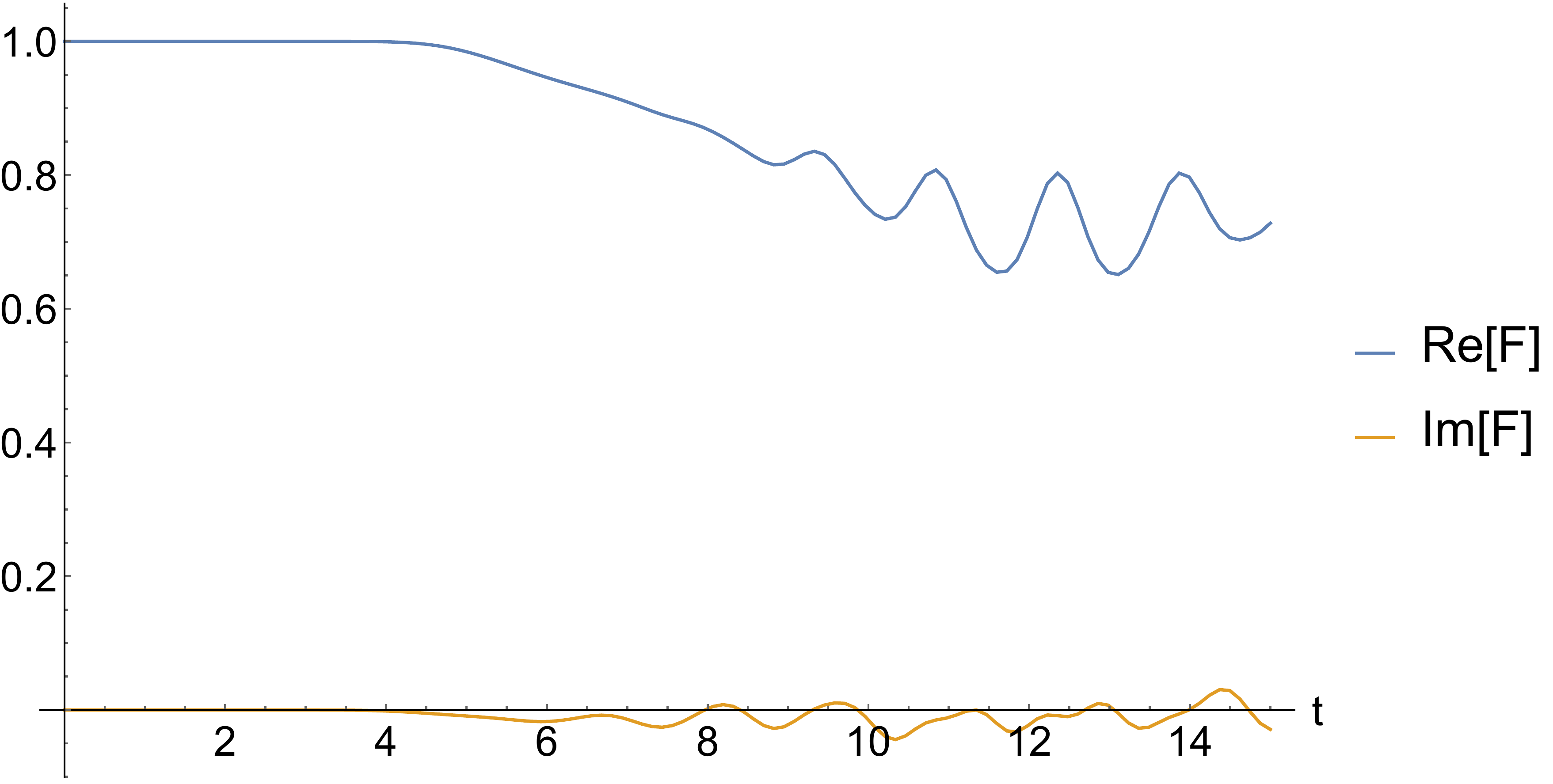}
\end{center}
\caption{Real (upper curve) and imaginary (lower curve) parts of $F(t)$ as a function of time. $T=1$ thermal state. Nonintegrable parameters, $\Sites=10$, $\mathcal{W}=\sigma_1^z$, $V=\sigma_\Sites^z$. }
\label{fig:T1_zz_F}
\end{figure}

\begin{figure}
\begin{center}
\includegraphics[width=.49\textwidth]{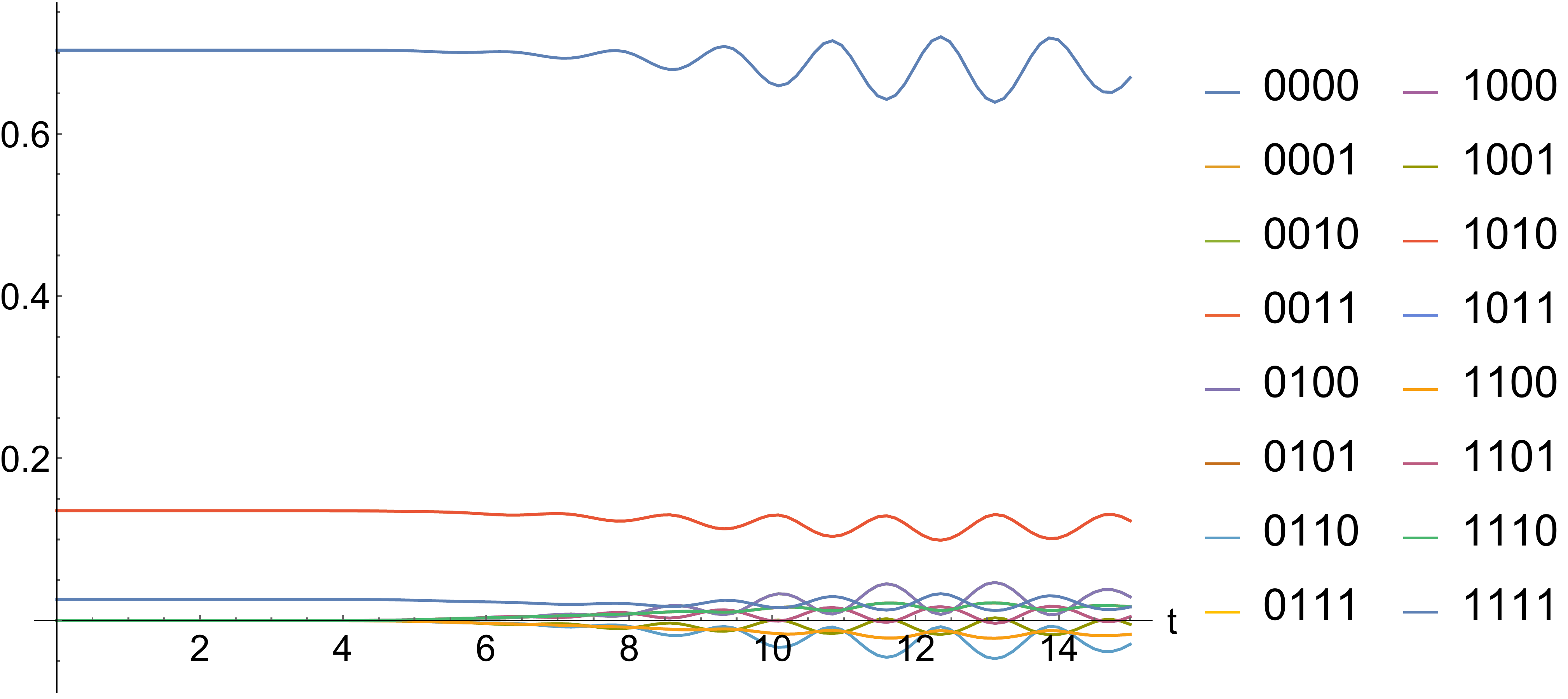}
\end{center}
\caption{Real part of $\SumKD{\rho}$ as a function of time. $T=1$ thermal state. Nonintegrable parameters, $\Sites=10$, $\mathcal{W}=\sigma_1^z$, $V=\sigma_\Sites^z$. The upper curve includes $0000$, the middle curve includes $0011$, while the lower cluster of curves includes $0100$. } \label{fig:T1_zz_AR}
\end{figure}

\begin{figure}
\begin{center}
\includegraphics[width=.49\textwidth]{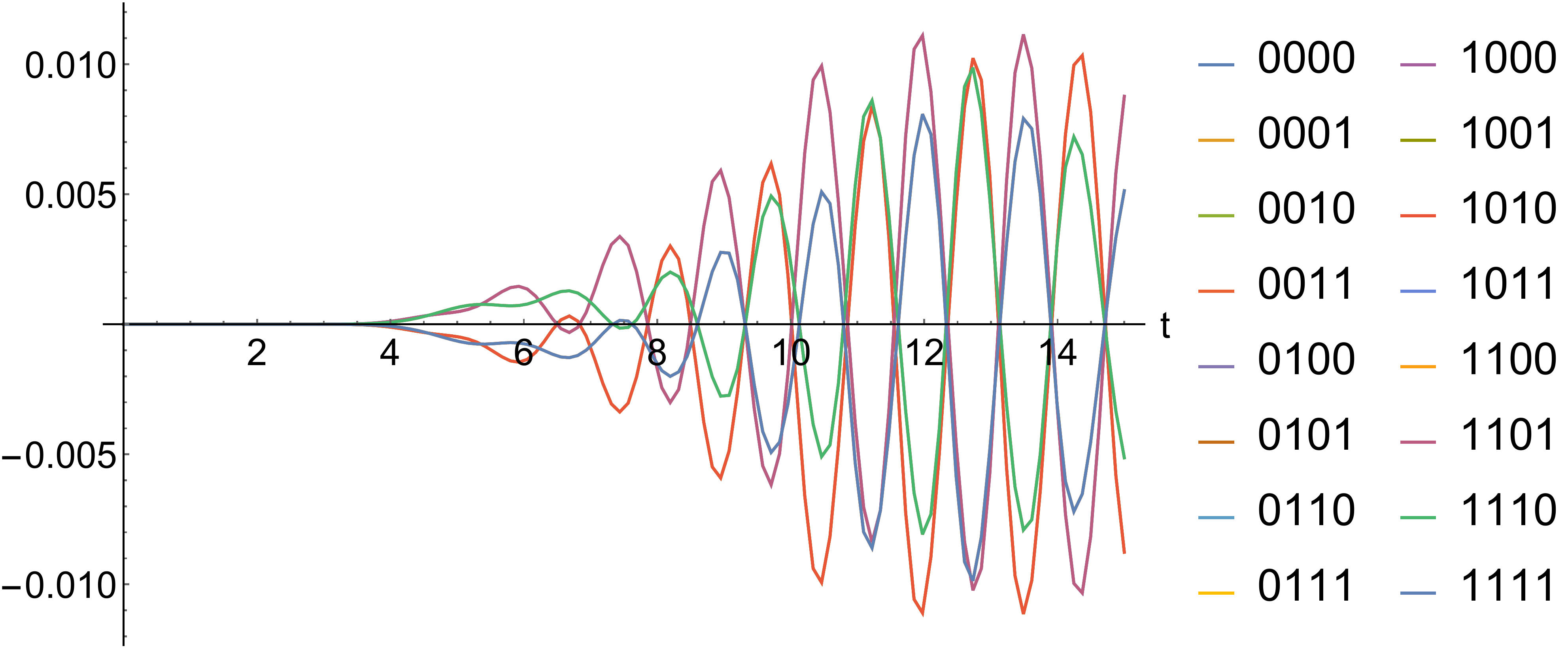}
\end{center}
\caption{Imaginary part of $\SumKD{\rho}$ as a function of time. $T=1$ thermal state. Nonintegrable parameters, $\Sites=10$, $\mathcal{W}=\sigma_1^z$, $V=\sigma_\Sites^z$. The various curves display similar looking oscillations as a function of time; some of the curves appear to be more-or-less related by a factor of minus one, for example $1111$ and $1110$.}
\label{fig:T1_zz_AI}
\end{figure}


\begin{figure}
\begin{center}
\includegraphics[width=.49\textwidth]{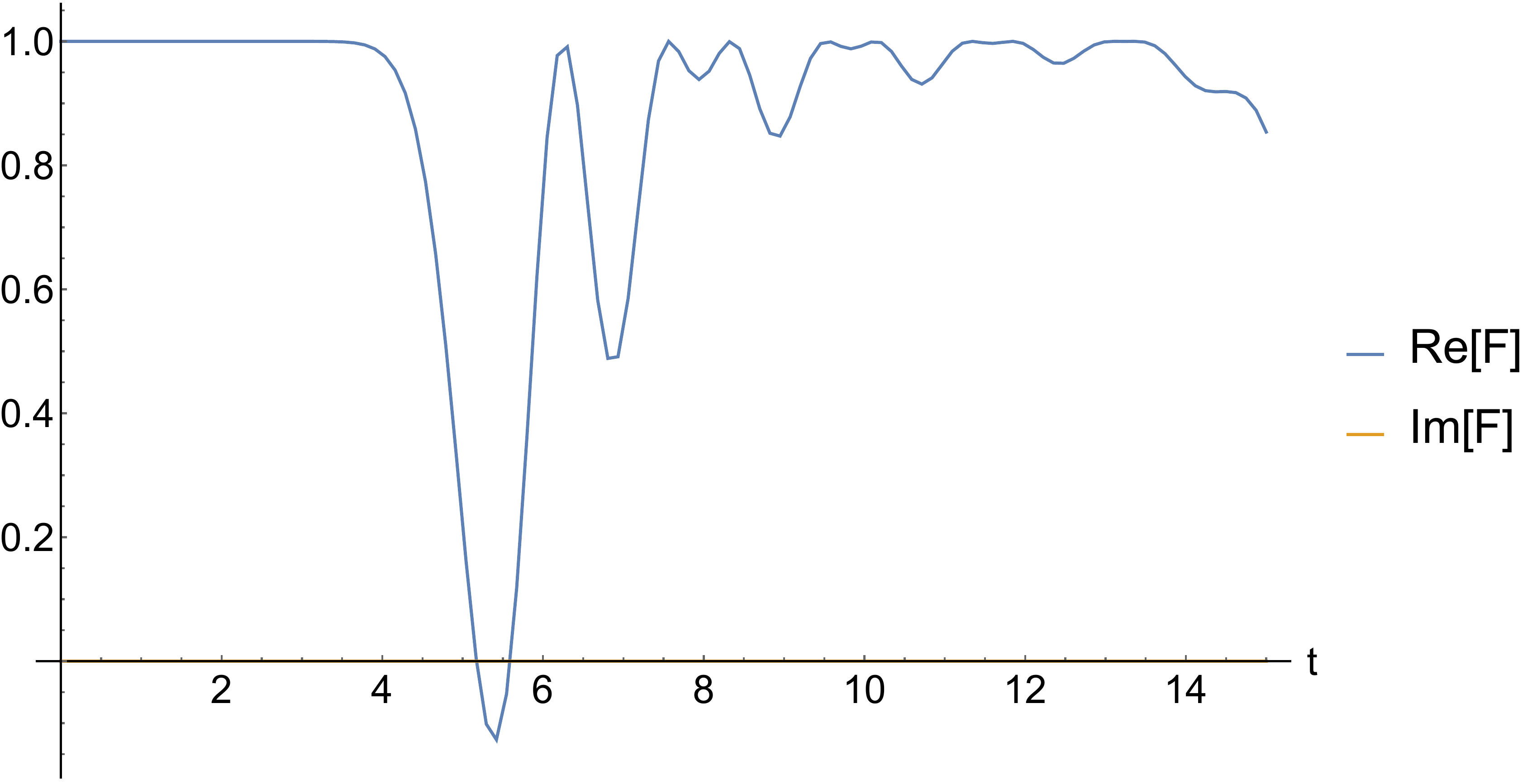}
\end{center}
\caption{Real (upper curve) and imaginary (lower curve) parts of $F(t)$ as a function of time. $T=\infty$ thermal state. Integrable parameters, $\Sites=10$, $\mathcal{W}=\sigma_1^z$, $V=\sigma_\Sites^z$. }
\label{fig:TInf_zz_int_F}
\end{figure}

\begin{figure}
\begin{center}
\includegraphics[width=.49\textwidth]{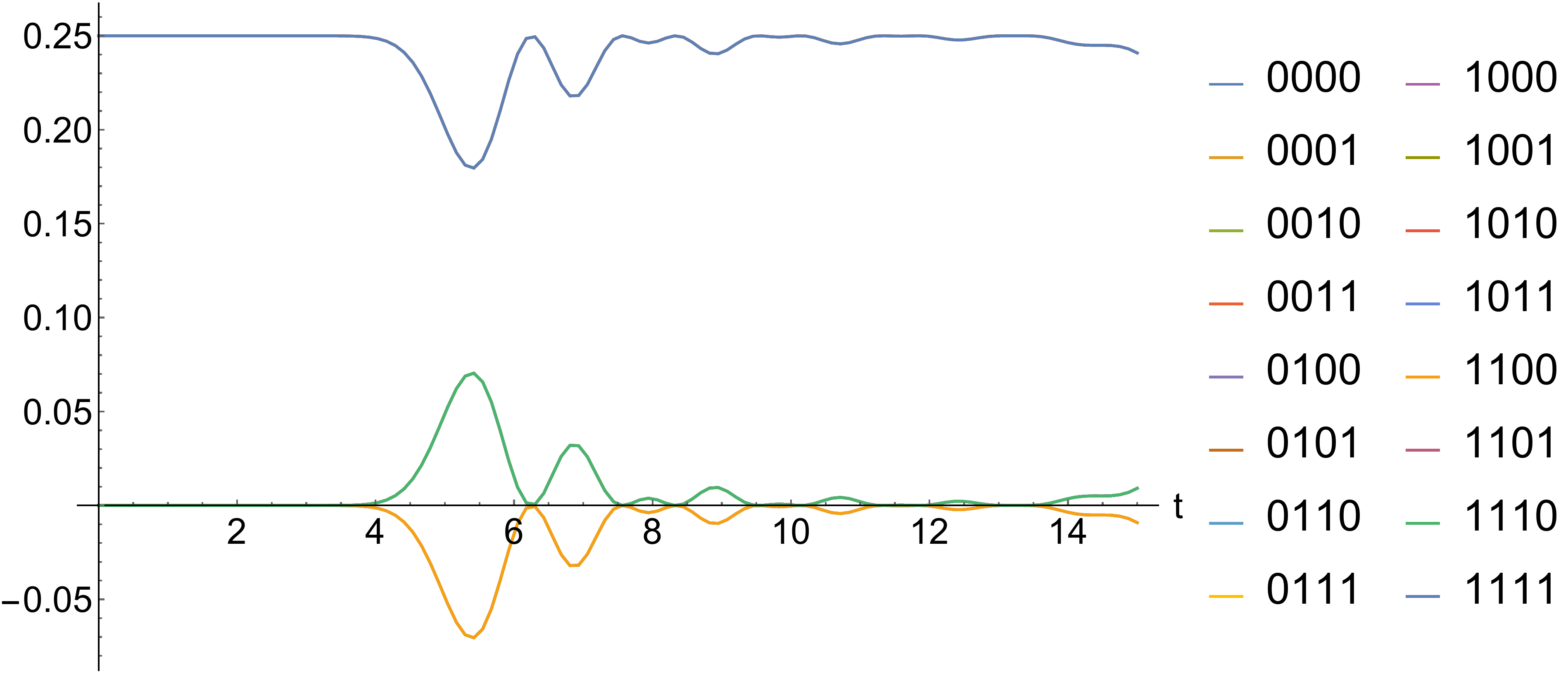}
\end{center}
\caption{Real part of $\SumKD{\rho}$ as a function of time. $T=\infty$ thermal state. Integrable parameters, $\Sites=10$, $\mathcal{W}=\sigma_1^z$, $V=\sigma_\Sites^z$. Upper curve includes $0000$ while for the lower curves, the upper includes $1110$ and the lower includes $0001$. }
\label{fig:TInf_zz_int_AR}
\end{figure}


\subsection{Random states}

We now consider random pure states
$\rho \propto \ketbra{ \psi }{ \psi }$ and nonintegrable parameters.
Figures \ref{fig:rand_zz_F}, \ref{fig:rand_zz_AR}, and \ref{fig:rand_zz_AI}
show $F(t)$ and $\SumKD{\rho}$
for the operator choice $\mathcal{W} = \sigma_1^z$ and $V = \sigma_\Sites^z$
in a randomly chosen pure state.
The pure state is drawn according to the Haar measure.
Each figure shows a single shot
(contains data from just one pure state).
Broadly speaking, the features are similar to those exhibited by
the infinite-temperature $\rho = \id / \Dim$, with additional fluctuations.

The upper branch of lines in Fig.~\ref{fig:rand_zz_AR}
exhibits dynamics before the OTOC does.
However, lines' average positions move significantly
(the lower lines bifurcate, and the upper lines shift downward)
only after the OTOC begins to evolve.
The early motion must be associated with
the early dynamics of the 2- and 3-point functions in Eq.~\eqref{eq:ProjTrick2}.
The late-time values are roughly consistent with
those for $\rho = \id / \Dim$ but fluctuate more pronouncedly.

The agreement between random pure states
and the $T = \infty$ thermal state is expected,
due to closed-system thermalization~\cite{D'Alessio_16_From,Gogolin_16_Equilibration}.
Consider assigning a temperature to a pure state
by matching its energy density with
the energy density of the thermal state $e^{ - H / T } / Z$,
cast as a function of temperature.
With high probability, any given random pure state
corresponds to an infinite temperature.
The reason is the thermodynamic entropy's
monotonic increase with temperature.
Since the thermodynamic entropy gives the density of states,
more states correspond to higher temperatures.
Most states correspond to infinite temperature.

For the random states and system sizes $\Sites$ considered,
if $H$ is nonintegrable,
the agreement with thermal results is not complete.
However, the physics appears qualitatively similar.

\begin{figure}
\begin{center}
\includegraphics[width=.49\textwidth]{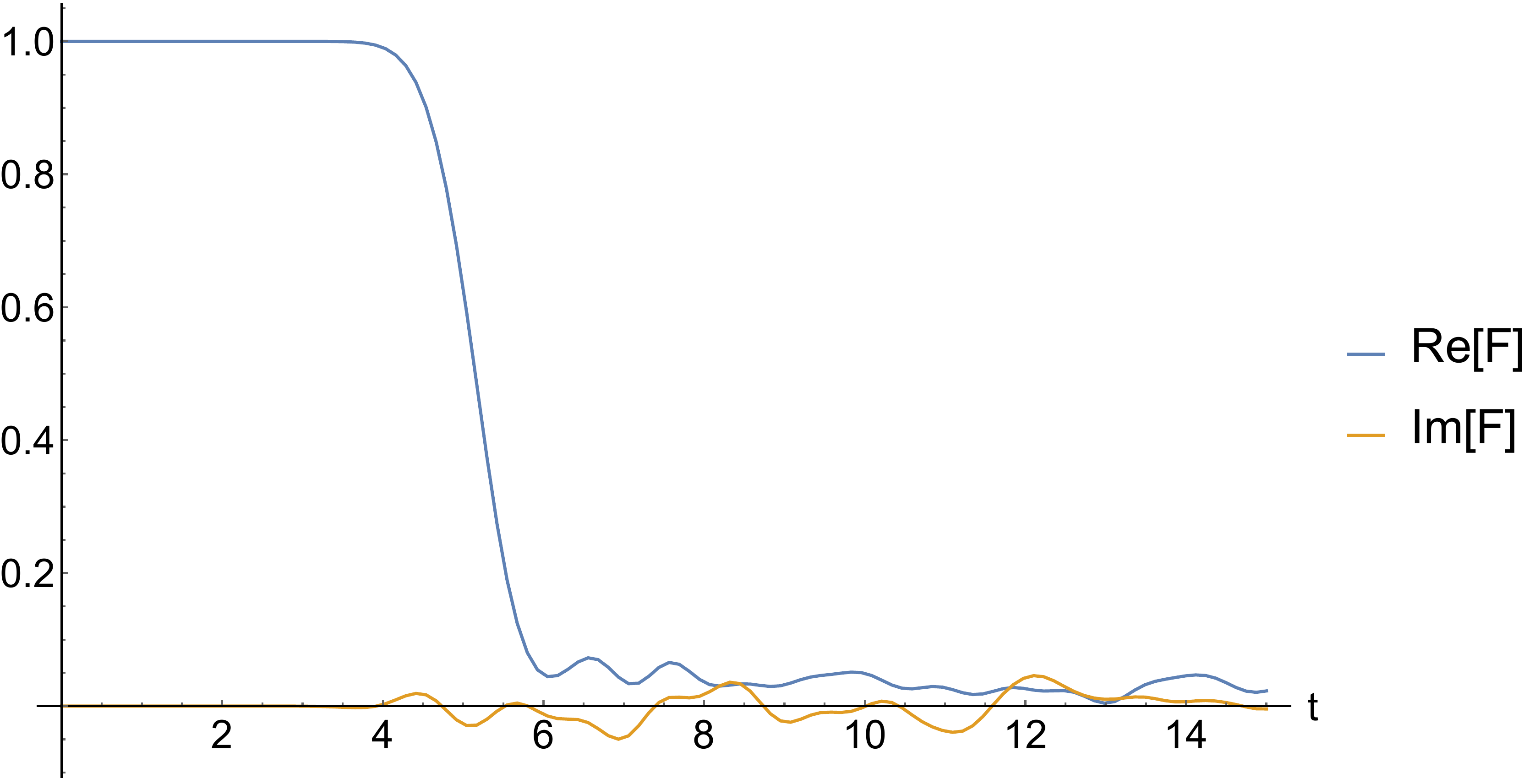}
\end{center}
\caption{Real (upper curve) and imaginary (lower curve) parts of $F(t)$ as a function of time. Random pure state. Nonintegrable parameters, $\Sites=10$, $\mathcal{W}=\sigma_1^z$, $V=\sigma_\Sites^z$. }
\label{fig:rand_zz_F}
\end{figure}

\begin{figure}
\begin{center}
\includegraphics[width=.49\textwidth]{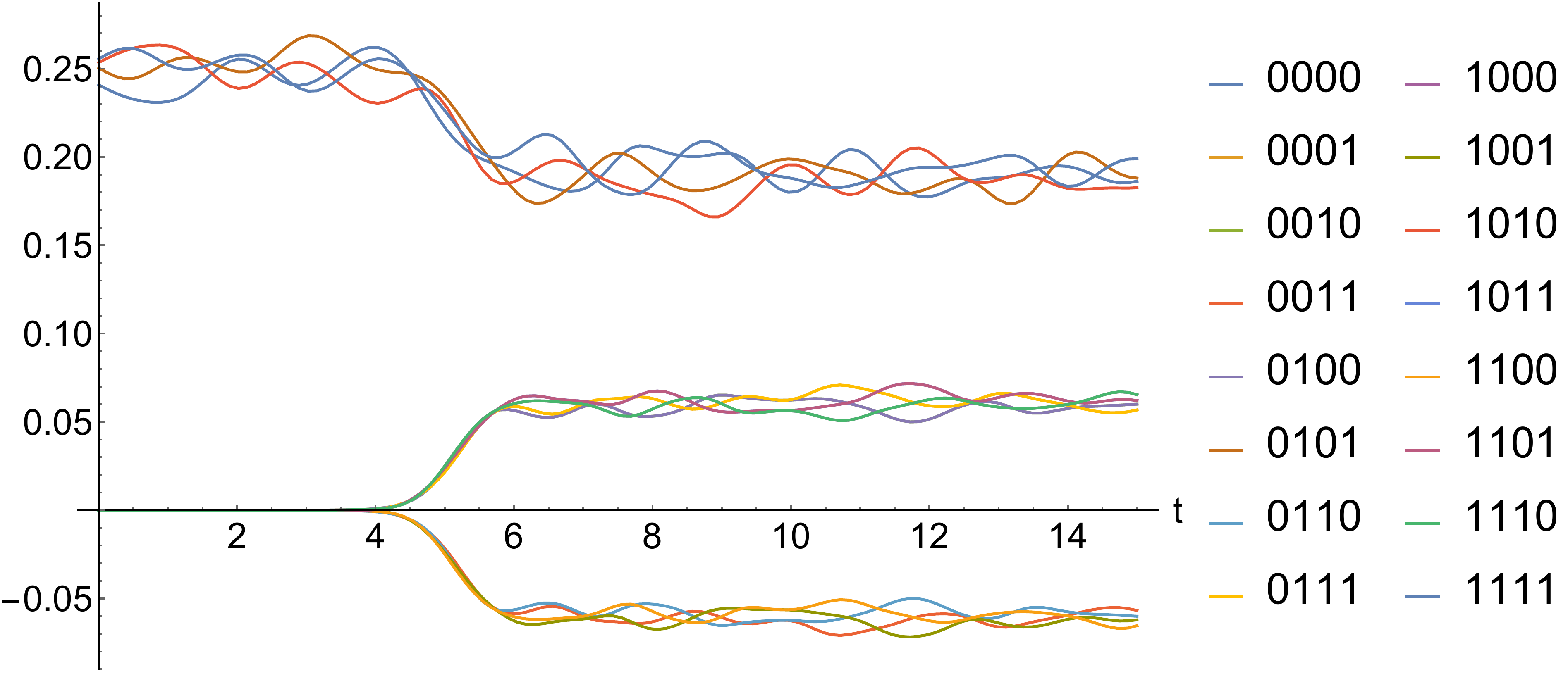}
\end{center}
\caption{Real part of $\SumKD{\rho}$ as a function of time. Random pure state. Nonintegrable parameters, $\Sites=10$, $\mathcal{W}=\sigma_1^z$, $V=\sigma_\Sites^z$. The upper cluster includes $0000$ and $0011$. The upper prong of the lower pitchfork includes $1110$ while the lower prong includes $1001$. }
\label{fig:rand_zz_AR}
\end{figure}

\begin{figure}
\begin{center}
\includegraphics[width=.49\textwidth]{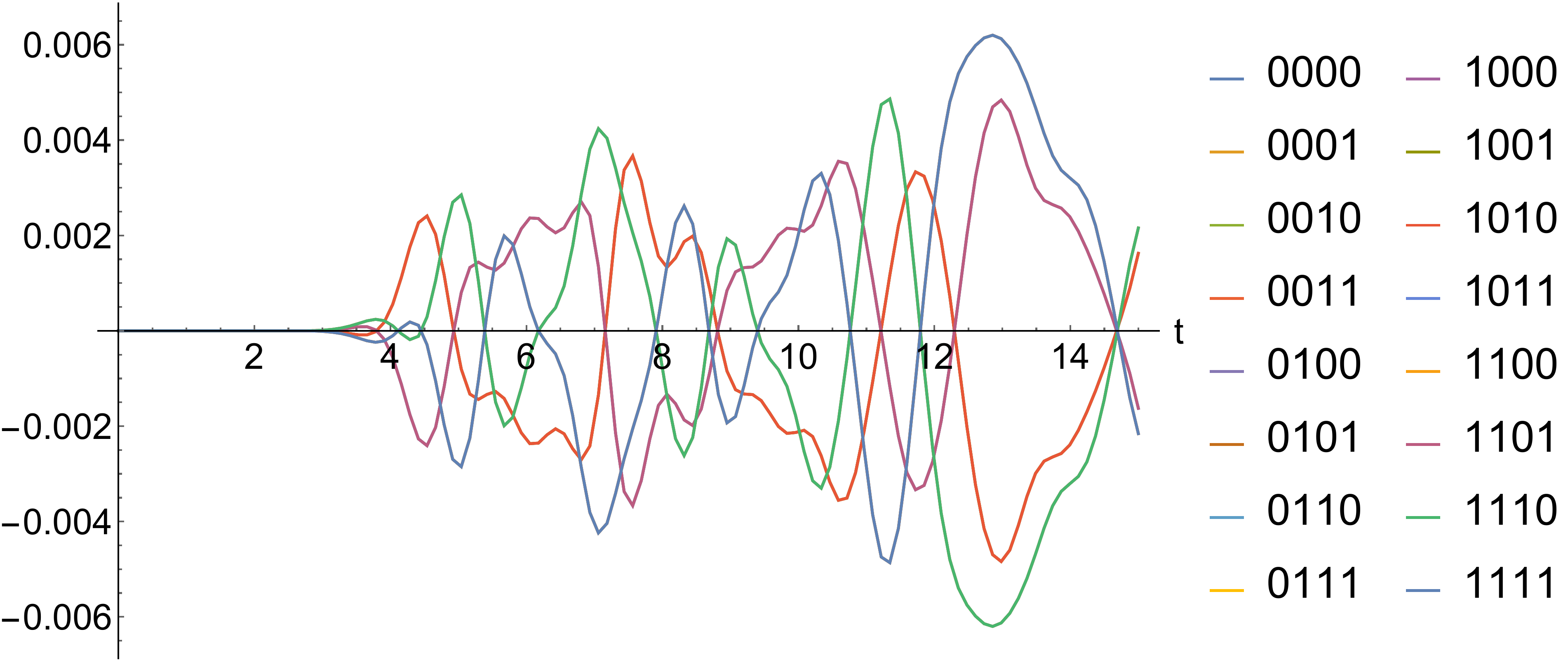}
\end{center}
\caption{Imaginary part of $\SumKD{\rho}$ as a function of time. Random pure state. Nonintegrable parameters, $\Sites=10$, $\mathcal{W}=\sigma_1^z$, $V=\sigma_\Sites^z$. A familiar pattern of oscillations is visible, with some curves being more-or-less related by a factor of minus one, for example, $0011$ and $1000$.}
\label{fig:rand_zz_AI}
\end{figure}

\subsection{Product states}

Finally, we consider the product $\ket{ +x }^{ \otimes \Sites}$
of $\Sites$ copies of the $+1$ $\sigma^x$ eigenstate
(Figures~\ref{fig:xup_zz_F}--\ref{fig:xup_zz_AI}).
We continue to use $\mathcal{W} = \sigma_1^z$ and $V = \sigma_\Sites^z$.
For the Hamiltonian parameters chosen,
this state lies far from the ground state.
The state therefore should correspond to
a large effective temperature.
Figures \ref{fig:xup_zz_F}, \ref{fig:xup_zz_AR}, and \ref{fig:xup_zz_AI} show
$F (t)$ and $\SumKD{\rho}$ for nonintegrable parameters.

The real part of $F(t)$ decays significantly from its initial value of one.
The imaginary part of $F(t)$ is nonzero but remains small.
These features resemble the infinite-temperature features.
However, the late-time $F(t)$ values are substantially larger than
in the $T = \infty$ case and oscillate significantly.

Correspondingly, the real and imaginary components of $\SumKD{\rho}$
oscillate significantly.
$\Re \left( \SumKD{\rho} \right)$ exhibits dynamics before scrambling begins,
as when $\rho$ is a random pure state.
The real and imaginary parts of $\SumKD{\rho}$ differ more from
their $T = \infty$ counterparts than
$F(t)$ differs from its counterpart.
Some of this differing is apparently washed out
by the averaging needed to construct $F(t)$
[Eq.~\eqref{eq:RecoverF2}].

We expected pure product states to behave
roughly like random pure states.
The data support this expectation very roughly, at best.
Whether finite-size effects cause this deviation,
we leave as a question for further study.

\begin{figure}
\begin{center}
\includegraphics[width=.49\textwidth]{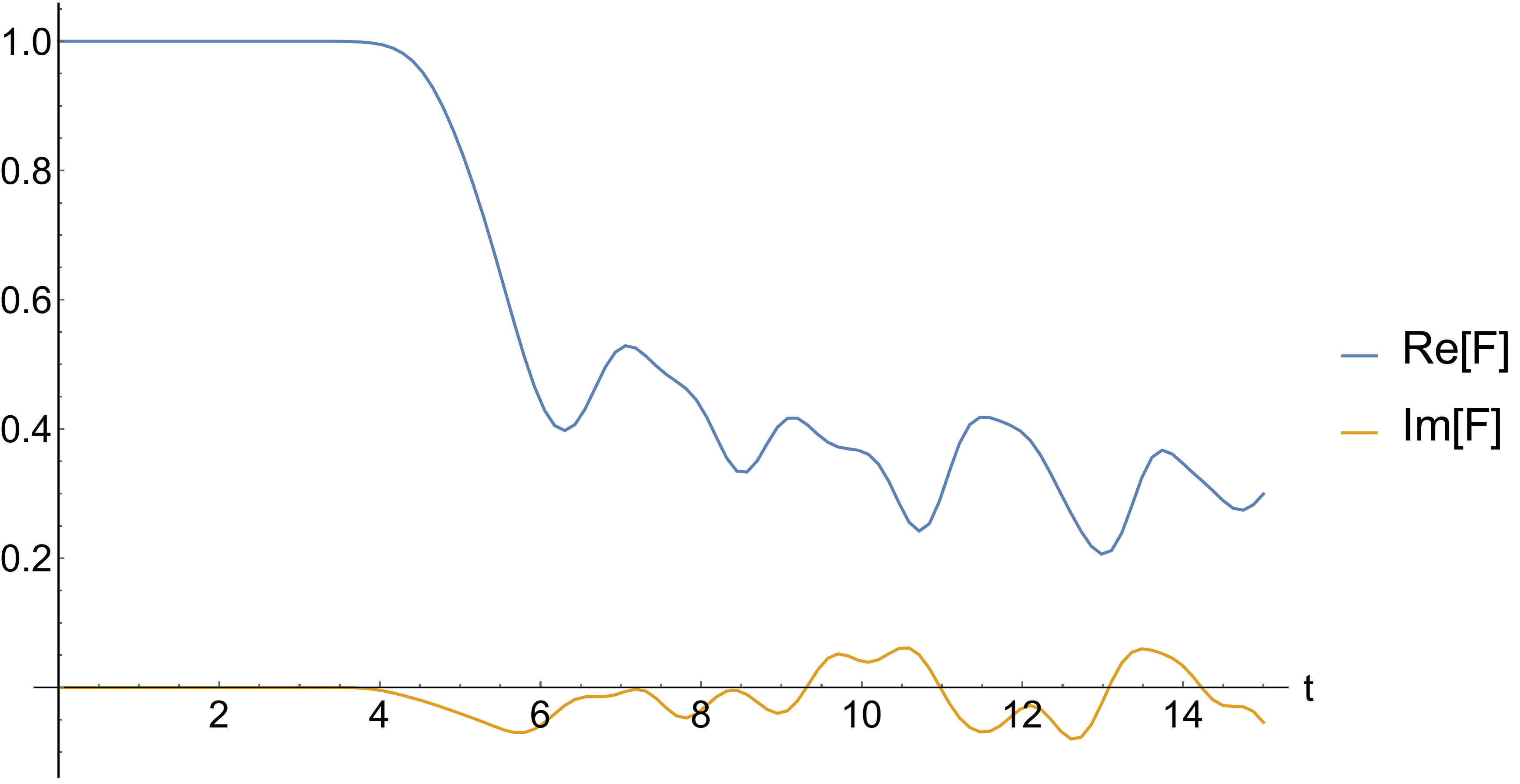}
\end{center}
\caption{Real (upper curve) and imaginary (lower curve) parts of $F(t)$ as a function of time.
Product $\ket{ +x }^{ \otimes \Sites}$ of $\Sites$ copies of
the $+1$ $\sigma^x$ eigenstate.
Nonintegrable parameters, $\Sites=10$, $\mathcal{W}=\sigma_1^z$, $V=\sigma_\Sites^z$. }
\label{fig:xup_zz_F}
\end{figure}

\begin{figure}
\begin{center}
\includegraphics[width=.49\textwidth]{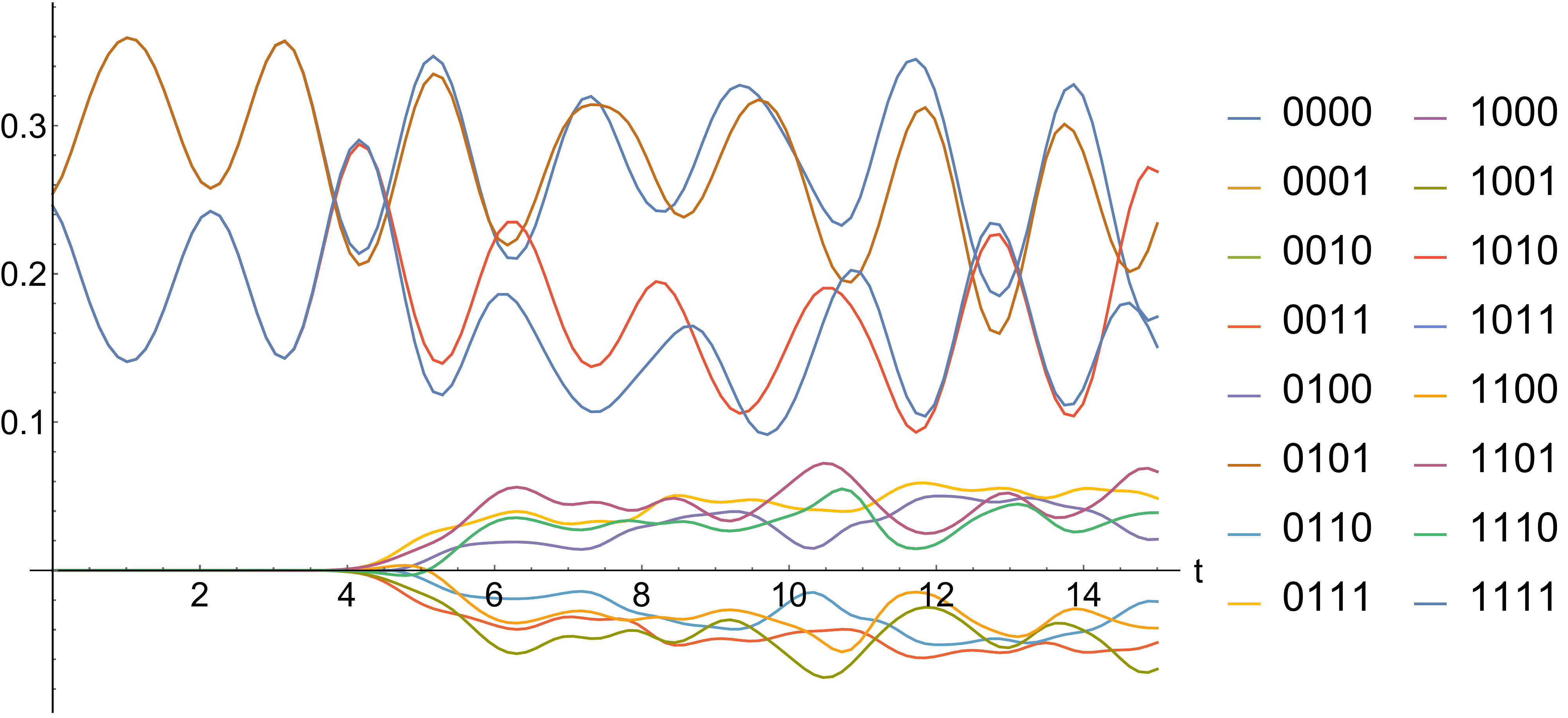}
\end{center}
\caption{Real part of $\SumKD{\rho}$ as a function of time.
Product $\ket{ +x }^{ \otimes \Sites}$ of $\Sites$ copies of
the $+1$ $\sigma^x$ eigenstate.
Nonintegrable parameters, $\Sites=10$, $\mathcal{W}=\sigma_1^z$, $V=\sigma_\Sites^z$. The top of the upper cluster includes $0101$ while the bottom of the upper cluster includes $0011$. The top of the lower pitchfork includes $1110$ while the bottom of the lower pitchfork includes $1001$.}
\label{fig:xup_zz_AR}
\end{figure}

\begin{figure}
\begin{center}
\includegraphics[width=.49\textwidth]{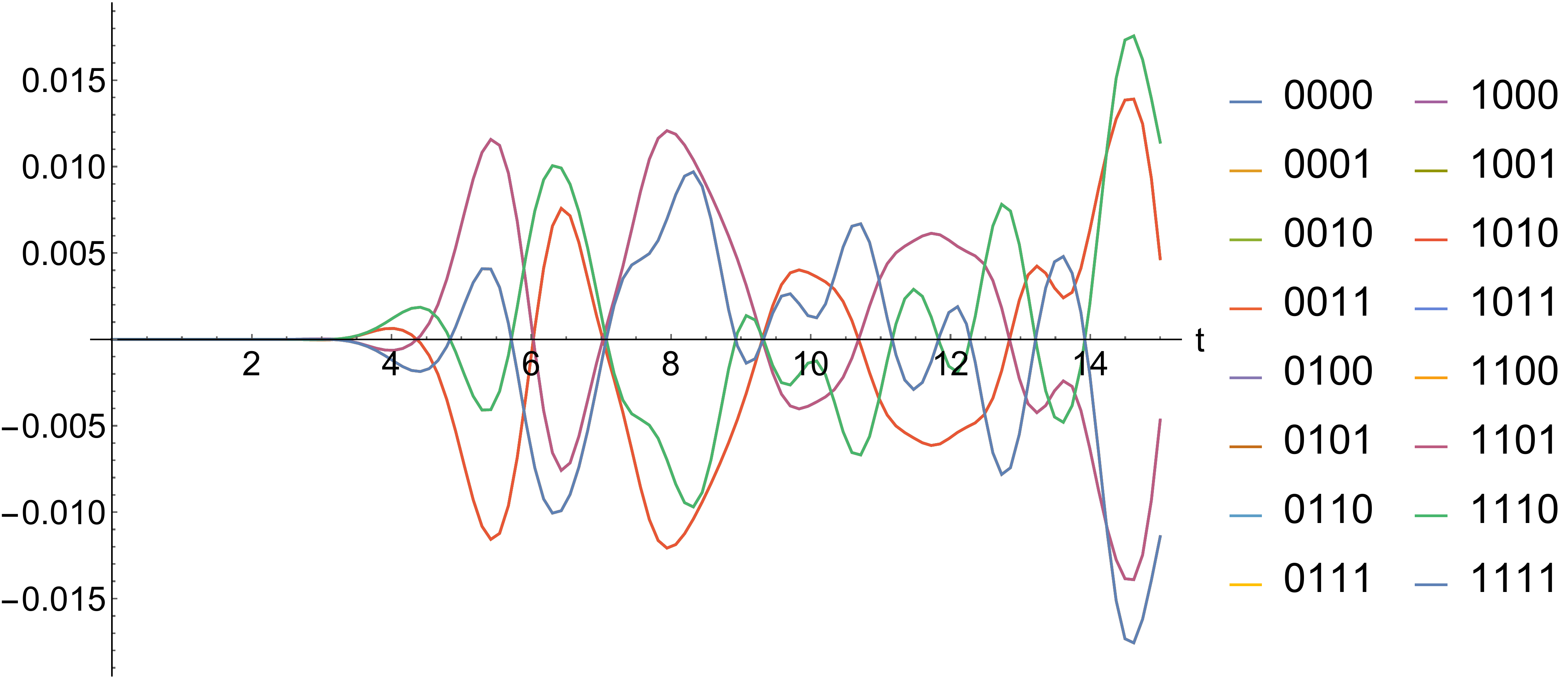}
\end{center}
\caption{Imaginary part of $\SumKD{\rho}$ as a function of time.
Product $\ket{ +x }^{ \otimes \Sites}$ of $\Sites$ copies of
the $+1$ $\sigma^x$ eigenstate.
Nonintegrable parameters, $\Sites=10$, $\mathcal{W}=\sigma_1^z$, $V=\sigma_\Sites^z$. Similar physics to other depictions of the imaginary part.}
\label{fig:xup_zz_AI}
\end{figure}

\subsection{Summary}

The main messages from this study are the following.
\begin{enumerate}[(1)]

\item The coarse-grained quasiprobability $\SumKD{\rho}$ is generically complex.
Exceptions include the $T = \infty$ thermal state $\id / \Dim$
and states $\rho$ that share
an eigenbasis with $V$ or with $\W(t)$
[e.g., as in Eq.~\eqref{eq:WRho}].
Recall that the KD distribution's nonreality
signals nonclassical physics (Sec.~\ref{section:Intro_to_KD}).

\item The derived quantity $P(W,W')$ is generically complex,
our results imply.\footnote{
The relevant plots are not shown,
so that this section maintains
a coherent focus on $\SumKD{\rho}$.
This result merits inclusion, however,
as $P(W, W')$ plays important roles in
(i)~\cite{YungerHalpern_17_Jarzynski}
and (ii) connections between the OTOC
and quantum thermodynamics (Sec.~\ref{section:Outlook}).
}
Nonclassicality thus survives even
the partial marginalization that defines $P$ [Eq.~\eqref{eq:PWWPrime}].
In general, marginalization can cause interference to dampen nonclassicality.
(We observe such dampening in Property~\ref{prop:MargOurKD}
of Sec.~\ref{section:TA_Props}
and in Property~\ref{prop:MargP} of Appendix~\ref{section:P_Properties}.)

\item Random pure states' quasiprobabilities resemble
the $T = \infty$ thermal state's quasiprobability
but fluctuate more.

\item Certain product states' quasiprobabilities
display anomalously large fluctuations.
We expected these states to resemble random states more.

\item The $\SumKD{\rho}$'s generated by integrable Hamiltonians
differ markedly from
the $\SumKD{\rho}$'s generated by nonintegrable Hamiltonians.
Both types of $\SumKD{\rho}$'s achieve nonclassical values, however.
We did not clearly observe a third class of behavior.

\item The time scale after which $\SumKD{\rho}$ changes significantly
is similar to the OTOC time scale.
$\SumKD{\rho}$ can display nontrivial early-time dynamics
not visible in $F(t)$.
This dynamics can arise, for example, because of
the 2-point function contained in the expansion of $\SumKD{\rho}$
[see Eq.~\eqref{eq:ProjTrick2}].

\item $\SumKD{\rho}$ reveals that scrambling breaks a symmetry.
Operationally, the symmetry consists of invariances of $\SumKD{\rho}$
under permutations and negations of measurement outcomes
in the weak-measurement scheme (Sec.~\ref{section:Intro_weak_meas}).
The symmetry breaking manifests in bifurcations of $\SumKD{\rho}$.
These bifurcations evoke classical-chaos pitchfork diagrams,
which also arise when a symmetry breaks.
One equilibrium point splits into three
in the classical case~\cite{Strogatz_00_Non}.
Perhaps the quasiprobability's pitchforks can be recast
in terms of equilibria.

\end{enumerate}

\section{Calculation of $\SumKD{\rho}$ averaged over Brownian circuits}
\label{section:Brownian}

We study a geometrically nonlocal model---the \emph{Brownian-circuit model}---governed by a time-dependent Hamiltonian~\cite{Lashkari_13_Towards}.
We access physics qualitatively different from
the physics displayed in the numerics of Sec.~\ref{section:Numerics}.
We also derive results for large systems
and compare with the finite-size numerics.
Since the two models' locality properties differ,
we do not expect agreement at early times.
The late-time scrambled states, however,
may be expected to share similarities.
We summarize our main findings at the end of the section.

We consider a system of $\Sites$ qubits
governed by the random time-dependent Hamiltonian
\begin{align}
H(t) \propto \sum_{i < j} \sum_{\alpha_i ,\alpha_j}
J^{\alpha_i,\alpha_j}_{i,j}(t)  \,
\sigma_i^{\alpha_i} \sigma_j^{\alpha_j}  \, .
\end{align}
The couplings $J$ are time-dependent random variables.
We denote the site-$i$ identity operator and Pauli operators by
$\sigma_i^\alpha$, for $\alpha=0,1,2,3$.
According to the model's precise formulation,
the time-evolution operator $U(t)$ is a random variable that obeys
\begin{align}\label{eq:browncirc}
   U(t+dt) & - U(t)  = - \frac{\Sites}{2} U(t) dt
   - i  \,  dB(t)  \, .
\end{align}
The final term's $dB(t)$ has the form
\begin{align}
   \label{eq:dB}
   dB(t) = \sqrt{\frac{1}{8(\Sites-1)}} \sum_{i < j} \sum_{\alpha_i, \alpha_j}
   \sigma_i^{\alpha_i} \sigma_j^{\alpha_j} dB^{\alpha_i,\alpha_j}_{i,j}(t)  \, .
\end{align}
We will sometimes call Eq.~\eqref{eq:dB} ``$dB$.''
$dB$ is a Gaussian random variable
with zero mean and with variance
\begin{align} \label{eq:dBdB}
   \mathbf{E}_B
   \left\{ dB^{\alpha,\beta}_{i,j}  \,
   dB^{\alpha',\beta'}_{i',j'} \right\} = \delta_{\alpha,\alpha'}\delta_{\beta,\beta'}
   \delta_{i,i'} \delta_{j,j'}  \,  dt.
\end{align}
The expectation value $\mathbf{E}_B$ is an average
over realizations of the noise $B$.
We demand that $dt  \,  dt =0$ and $dB  \,  dt = 0$,
in accordance with the standard Ito calculus.
$dB(t)$ is independent of $U(t)$,
i.e., of all previous $dB$'s.

We wish to compute the average,
over the ensemble defined by Eq.~\eqref{eq:browncirc},
of the coarse-grained quasiprobability:
\begin{align}
\mathfrak{A}( v_1 , w_2 , v_2 , w_3 )
= \mathbf{E}_B\left\{ \SumKD{\rho} ( v_1 , w_2 , v_2 , w_3 ) \right\}  \, .
\end{align}

\subsection{Infinite-temperature thermal state $\id / 2^\Sites$}

We focus here on the infinite-temperature thermal state,
$\rho = \id/2^{\Sites}$, for two reasons.
First, a system with a time-dependent Hamiltonian generically heats to infinite temperature with respect to any Hamiltonian in the ensemble.
Second, the $T = \infty$ state is convenient for calculations.
A discussion of other states follows.

The ensemble remains invariant under single-site rotations,
and all qubits are equivalent.
Therefore, all possible choices of
single-site Pauli operators for $\mathcal{W}$ and $V$ are equivalent.
Hence we choose
$\mathcal{W} = \sigma_1^z$ and $V = \sigma_2^z$
without loss of generality.

Let us return to Eq.~\eqref{eq:ProjTrick}.
Equation~\eqref{eq:ProjTrick} results from
substituting in for the projectors in $\SumKD{\rho}$.
The sum contains $16$ terms.
To each term, each projector contributes
the identity $\id$ or a nontrivial Pauli ($\mathcal{W}$ or $V$).
The terms are
\begin{enumerate}[(1)]
  \item $\id \id \id \id $: $ \text{Tr}\left\{\frac{ \id }{2^\Sites} \right\} =  1 $,

  \item $\mathcal{W}\id \id \id $, $\id V\id \id $,
  $\id \id \mathcal{W} \id $, $\id \id \id V$: $ 0 $,

  \item $\mathcal{W}V\id \id $, $\mathcal{W} \id \id V$,
  $ \id V\mathcal{W} \id $,  $ \id  \id \mathcal{W}V$: \\
  $ \text{Tr}\left\{\frac{\sigma_1^z(t) \sigma_2^z}{2^\Sites} \right\} =: G(t) $,

  \item $\mathcal{W} \id \mathcal{W} \id $, $ \id V \id V$:
  $ \text{Tr}\left\{\frac{ \id }{2^\Sites} \right\} =  1 $,

  \item $\mathcal{W}V\mathcal{W} \id $,
  $\mathcal{W}V \id V$, $\mathcal{W} \id \mathcal{W}V$,
  $ \id V\mathcal{W}V$: $ 0 $,  \quad and

  \item $\mathcal{W}V\mathcal{W}V$:
  $ \text{Tr}\left\{\frac{\sigma_1^z(t)
  \sigma_2^z \sigma_1^z(t) \sigma_2^z}{2^\Sites} \right\} = F(t) $.

\end{enumerate}
These computations rely on $\rho =  \id /2^\Sites$.
Each term that contains an odd number of Pauli operators vanishes,
due to the trace's cyclicality and to the Paulis' tracelessness.
We have introduced a 2-point function $G(t)$.
An overall factor of $1/16$ comes from the projectors' normalization.

Combining all the ingredients,
we can express $\SumKD{\rho}$ in terms of $G$ and $F$. The result is
\begin{align}
& 16  \,  \SumKD{\rho} ( v_1 , w_2 , v_2 , w_3 )
= (1 + w_2 w_3 + v_1 v_2)  \\  \nonumber
& \qquad + (w_2+w_3)(v_1+v_2)  \,  G
+ w_2 w_3 v_1 v_2  \, F.
\end{align}
This result depends on $\rho = \id / 2^\Sites$,
not on the form of the dynamics.
But to compute $\mathfrak{A}$, we must compute
\begin{align}
\mathfrak{G} = \mathbf{E}_B\left\{G \right\}
\end{align}
and
\begin{align}
\mathfrak{F} = \mathbf{E}_B\left\{F \right\}.
\end{align}

The computation of $\mathfrak{F}$ appears in the literature~\cite{Shenker_Stanford_15_Stringy}.
$\mathfrak{F}$ initially equals unity.
It decays to zero around 
$t_* = \frac{1}{3} \log \Sites$, the scrambling time.
The precise functional form of $\mathfrak{F}$ is not crucial.
The basic physics is captured in a phenomenological form
inspired by AdS/CFT computations \cite{Shenker_Stanford_15_Stringy},
\begin{align}
\mathfrak{F} \sim \left(\frac{1+c_1}{1+c_1 e^{3 t}}\right)^{c_2},
\end{align}
wherein $c_1 \sim 1/\Sites$ and $c_2 \sim 1$.

To convey a sense of the physics, we review the simpler calculation of $\mathfrak{G}$. The two-point function evolves according to
\begin{align}
G(t+dt) &= \frac{1}{2^\Sites}  \:  \text{Tr}\bigg\{
\left[U(t) - \frac{\Sites}{2 } U(t) dt - i  \,  dB  \,  U(t) \right]
\sigma_1^z
\nonumber \\ & \times
\left[U(t)^\dagger - \frac{\Sites}{2} U(t)^\dagger dt
+ i  \,  U(t)^\dagger dB^\dagger \right]  \sigma_2^z \bigg\} \, .
\end{align}
Using the usual rules of Ito stochastic calculus,
particularly Eq.~\eqref{eq:dBdB} and $dt  \, dt = dB  \, dt =0$,
we obtain
\begin{align}
   & \mathfrak{G}(t+dt)  - \mathfrak{G}(t)
   = - \Sites  \,  dt  \,  \mathfrak{G}(t)
   + dt  \,   \frac{1}{8(\Sites-1)}
   \nonumber \\ & \times
   \sum_{i < j}  \sum_{\alpha_i,\alpha_j} \frac{1}{2^\Sites}
   \mathbf{E}_B\left\{ \text{Tr}\left\{
   \sigma_1^z(t) \sigma_i^{\alpha_i} \sigma_j^{\alpha_j} \sigma_2^z     \sigma_i^{\alpha_i} \sigma_j^{\alpha_j}   \right\} \right\}.
\end{align}
We have applied the trace's cyclicality in the second term.

The second term's value depends on whether $i$ and/or $j$ equals $2$.
If $i$ and/or $j$ equals $2$,
the second term vanishes because
$\sum_{\alpha=0}^3 \sigma^\alpha \sigma^z \sigma^\alpha = 0$.
If neither $i$ nor $j$ is $2$,
$\sigma_i^{\alpha_i} \sigma_j^{\alpha_j}$ commutes with $\sigma_2^z$.
The second term becomes proportional to $G$.
In $(\Sites-1)(\Sites-2)/2$ terms, $i, j \neq 2$.
An additional factor of $4^2 = 16$
comes from the two sums over Pauli matrices. Hence
\begin{align}
\mathfrak{G}(t+dt) - \mathfrak{G}(t) = - 2 dt  \, \mathfrak{G}  \, ,
\end{align}
or
\begin{align}
\frac{d \mathfrak{G}}{dt} = - 2 \mathfrak{G} .
\end{align}

This differential equation implies that $\mathfrak{G}$
exponentially decays from its initial value.
The initial value is zero: $\mathfrak{G}(0) = G(0) = 0$.
Hence $\mathfrak{G}(t)$ is identically zero.

Although it does not arise when we consider $\mathfrak{A}$,
the ensemble-average autocorrelation function
$\mathbf{E}_B\left\{\langle \sigma_1^z(t) \sigma_1^z\rangle \right\}$
obeys a differential equation similar to
the equation obeyed by $\mathfrak{G}$.
In particular, the equation decays exponentially with
an order-one rate.

By the expectation value's linearity and the vanishing of $\mathfrak{G}$,
\begin{align} \label{avgA}
\mathfrak{A} = \frac{(1+w_2 w_3 + v_1 v_2)
+ w_2 w_3 v_1 v_2  \,  \mathfrak{F} }{16}.
\end{align}
This simple equation states that
the ensemble-averaged quasiprobability depends only on
the ensemble-averaged OTOC $F(t)$,
at infinite temperature.
The time scale of $\mathfrak{F}$'s decay is $t_* = \frac{1}{3} \log \Sites$.
Hence this is the time scale of changes in $\mathfrak{A}$.

Equation~\eqref{avgA} shows (as intuition suggests)
that $\mathfrak{A}$ depends only on
the combinations $w_2 w_3$ and $v_1 v_2$.
At $t=0$, $\mathfrak{F}(0)=1$. Hence $\mathfrak{A}$ is
\begin{align}
\mathfrak{A}_{t=0} = \frac{1 + w_2 w_3 + v_1 v_2 +  w_2 w_3 v_1 v_2}{16}.
\end{align}
The cases are
\begin{enumerate}[(1)]
\item $w_2 w_3 =1, v_1 v_2 =1$: $\mathfrak{A} = 1/4$,
\item $w_2 w_3=1, v_1 v_2=-1$: $\mathfrak{A} = 0$,
\item $w_2 w_3=-1, v_1 v_2=1$: $\mathfrak{A} = 0$,
\quad \text{and}
\item $w_2 w_3=-1, v_1 v_2=-1$: $\mathfrak{A} = 0$.
\end{enumerate}
These values are consistent with Fig.~\ref{fig:TInf_zz_AR} at $t=0$.
These values' degeneracies are consistent with
the symmetries discussed in Sec.~\ref{section:Numerics}
and in Sec.~\ref{section:TA_Props} (Property~\ref{property:Syms}).

At long times, $\mathfrak{F}(\infty) = 0$, so $\mathfrak{A}$ is
\begin{align}
\mathfrak{A}_{t=\infty} = \frac{1 + w_2 w_3 + v_1 v_2}{16}.
\end{align}
The cases are
\begin{enumerate}[(1)]
\item $w_2 w_3=1, v_1 v_2=1$: $\mathfrak{A} = 3/16$,
\item $w_2 w_3=1, v_1 v_2=-1$: $\mathfrak{A} = 1/16$,
\item $w_2 w_3=-1, v_1 v_2=1$: $\mathfrak{A} = 1/16$,
\quad \text{and}
\item $w_2 w_3=-1, v_1 v_2=-1$: $\mathfrak{A} = -1/16$.
\end{enumerate}
Modulo the splitting of the upper two lines,
this result is broadly consistent with
the long-time behavior in Fig.~\ref{fig:TInf_zz_AR}.
As the models in Sec.~\ref{section:Numerics} and this section differ,
the long-time behaviors need not agree perfectly.
However, the models appear to achieve
qualitatively similar scrambled states at late times.

\subsection{General state}

Consider a general state $\rho$, such that
$\SumKD{\rho}$ assumes the general form in Eq.~\eqref{eq:ProjTrick}.
We still assume that $\mathcal{W}=\sigma_1^z$ and $V = \sigma_2^z$.
However, the results will, in general, now depend on these choices
via the initial condition $\rho$.
We still expect that, at late times,
the results will not depend on the precise choices.
Below, we use the notation $\langle . \rangle \equiv \text{Tr}( \rho \, . )$.

We must consider 16 terms again.
The general case involves fewer simplifications. The terms are
\begin{enumerate}[(1)]
  \item $\id \id \id \id $: $ 1 $,

  \item $\mathcal{W} \id \id \id$, $ \id V \id  \id $,
  $ \id  \id \mathcal{W} \id $, $ \id  \id  \id V$:
  $\langle \sigma_1^z(t) \rangle $ , $\langle \sigma_2^z \rangle$,

  \item $\mathcal{W}V \id  \id $, $\mathcal{W} \id  \id V$,
  $ \id V\mathcal{W} \id $, $ \id  \id \mathcal{W}V$: \\
  $ \langle \sigma_1^z(t)  \,  \sigma_2^z \rangle $,
  $\langle \sigma_2^z  \,   \sigma_1^z(t) \rangle$,

  \item $\mathcal{W} \id \mathcal{W} \id $, $ \id V \id V$: $ 1 $,

  \item  \label{item:ThreeIs}
  $\mathcal{W}V\mathcal{W} \id $, $\mathcal{W}V \id V$,
  $\mathcal{W} \id \mathcal{W}V$, $ \id V\mathcal{W}V$:
  $\langle \sigma_1^z(t)  \,   \sigma_2^z  \,   \sigma_1^z(t)  \rangle$, $
  \langle \sigma_1^z(t)  \,   \rangle$,
  $\langle \sigma_2^z \rangle$,
  $\langle \sigma_2^z  \,   \sigma_1^z(t)  \,   \sigma_2^z \rangle$,
  \quad \text{and}

  \item $\mathcal{W}V\mathcal{W}V$:
  $ \langle \sigma_1^z(t)  \,   \sigma_2^z  \,   \sigma_1^z(t)  \,   \sigma_2^z \rangle = F(t) $.

\end{enumerate}

Consider first the terms of the form
$\mathfrak{q}_i(t) := \mathbf{E}_B\{\langle \sigma_i^z(t) \rangle\}$.
The time derivative is
\begin{align}
&\frac{d \mathfrak{q}_i}{dt} = - \Sites \mathfrak{q}_i  \\ \nonumber
& + \frac{1}{8(\Sites-1)} \sum_{j<k} \sum_{\alpha_j,\alpha_k} \mathbf{E}_B \{ \langle \sigma_j^{\alpha_j} \sigma_k^{\alpha_k} U(t) \sigma_i^z U(t)^\dagger \sigma_j^{\alpha_j} \sigma_k^{\alpha_k} \rangle \}.
\end{align}
To simplify the second term, we use a trick. Since
\begin{align}
\label{eq:BrwnTrick}
\sigma_j^{\alpha_j} \sigma_k^{\alpha_k} \sigma_m^{\alpha_m} \sigma_n^{\alpha_n} \sigma_j^{\alpha_j} \sigma_k^{\alpha_k} = \pm \sigma_m^{\alpha_m} \sigma_n^{\alpha_n},
\end{align}
we may pass the factors of $\sigma_j^{\alpha_j} \sigma_k^{\alpha_k}$
through $U(t)$, at the cost of changing some Brownian weights.
We must consider a different set of $dB$'s,
related to the originals by minus signs.
This alternative set of Brownian weights
has the original set's ensemble probability.
Hence the ensemble average gives the same result. Therefore,
\begin{align}
& \mathbf{E}_B\{\langle \sigma_j^{\alpha_j} \sigma_k^{\alpha_k}
U(t) \, \sigma_i^z U(t)^\dagger \sigma_j^{\alpha_j} \sigma_k^{\alpha_k} \rangle \} \nonumber \\
& = \mathbf{E}_B\{\langle U(t) \sigma_j^{\alpha_j} \sigma_k^{\alpha_k}  \sigma_i^z \sigma_j^{\alpha_j} \sigma_k^{\alpha_k}  U(t)^\dagger \rangle \}.
\end{align}

If $i = j$ and/or $i = k$,
the sum over $\alpha_{j}$
and/or the sum over $\alpha_k$ vanishes.
If $i$ equals neither $j$ nor $k$, the Pauli operators commute.
The term reduces to $\mathfrak{q}_i$.
$i$ equals neither $j$ nor $k$ in
$(\Sites-1)(\Sites-2)/2$ terms.
A factor of $16$ comes from the sums over $\alpha_j$ and $\alpha_k$. Hence
\begin{align}
\frac{d \mathfrak{q}_i}{dt} = - \Sites \mathfrak{q}_i + (\Sites-2) \mathfrak{q}_i = - 2 \mathfrak{q}_i.
\end{align}

Consider the terms of the form $\mathfrak{q}_{ij}(t) := \langle \sigma_i^z(t) \sigma_j^z \rangle$.
Note that $\langle \sigma_j^z \sigma_i^z(t) \rangle = \mathfrak{q}_{ij}^*$.
We may reuse the trick introduced above.
[This trick fails only when more than two copies of $U$ appear, as in $F(t)$].
To be precise,
\begin{align}
& \mathbf{E}_B\{\langle \sigma_m^{\alpha_m} \sigma_n^{\alpha_n} U(t) \sigma_i^z U(t)^\dagger \sigma_m^{\alpha_m} \sigma_n^{\alpha_n} \sigma_j^z \rangle \} \nonumber \\
& = \mathbf{E}_B\{\langle U(t) \sigma_m^{\alpha_m} \sigma_n^{\alpha_n}  \sigma_i^z \sigma_m^{\alpha_m} \sigma_n^{\alpha_n} U(t)^\dagger \sigma_j^z \rangle \}.
\end{align}
As before, the sums over $\alpha$
kill the relevant term in
the time derivative of $\mathfrak{q}_{ij}$, unless $i \neq m,n$. Hence
\begin{align}
\frac{d \mathfrak{q}_{ij}}{dt} = - 2 \mathfrak{q}_{ij} ,
\end{align}
as at infinite temperature.

Item~\ref{item:ThreeIs}, in the list above, concerns
products of three $\W$'s and $V$'s.
We must consider four expectation values of Pauli products.
As seen above, two of these terms reduce to $\mathfrak{q}_i$ terms.
By the trick used earlier,
\begin{align}
& \mathbf{E}_B \{ \langle \sigma_2^z U(t) \sigma_1^z U(t)^\dagger \sigma_2^z \} \nonumber \\
& = \mathbf{E}_B \{ \langle U(t) \sigma_2^z  \sigma_1^z \sigma_2^z U(t)^\dagger  \} = \mathfrak{q}_1(t).
\end{align}
The other term we must consider is
$\mathbf{E}_B \{ \langle \sigma_i^z(t) \sigma_j^z \sigma_i^z(t) \rangle \}
=:  \mathfrak{f}_{ij}$.
Our trick will not work, because
there are multiple copies of $U(t)$
that are not all simultaneously switched
as operators are moved around.
At early times, when $\sigma_i^z(t)$ and $\sigma_j^z$ approximately commute, this term approximately equals $\langle \sigma_j^z \rangle = \mathfrak{q}_j(0)$.
At later times, including around the scrambling time,
this term decays to zero.

The general expression for $\mathfrak{A}$ becomes
\begin{eqnarray}
  16  \,  \mathfrak{A} &=& 1 + w_3 w_2 + v_1 v_2 \nonumber \\
    &+& (w_3 +w_2)  \, \mathfrak{q}_1(t)
    + (v_1 + v_2)  \, \mathfrak{q}_2(0) \nonumber \\
    &+& (w_3 v_2 + w_3 v_1 + w_2 v_1)  \,  \mathfrak{q}_{12}(t)
    + v_2 w_2  \,  \mathfrak{q}_{12}(t)^* \nonumber \\
    &+& w_3 v_2 w_2  \,  \mathfrak{f}_{12}(t)
    + (w_3 v_1 v_2 + w_2 v_1 v_2)  \,  \mathfrak{q}_1(t)
    \nonumber \\ \label{eq:Brown_help}
    &+& w_3 w_2 v_1  \,  \mathfrak{q}_2(0)
    + w_3 w_2 v_1 v_2  \,  \mathfrak{F}(t).
\end{eqnarray}
All these $\mathfrak{q}$ functions obey known differential equations.
The functions decay after a time of order one.
We do not have explicit expressions for the $\mathfrak{f}$ functions that appear.
They are expected to vary after a time $\sim  \log \Sites$.

\subsubsection{Special case: $\sigma_2^z$ eigenstate}

In a concrete example, we suppose that $\rho$ is
a $+1$ eigenstate of $\sigma_2^z$.
Expressions simplify:
\begin{align}
\mathfrak{q}_2(0) = 1,
\end{align}
\begin{align}
\mathfrak{q}_{12}(t) = \mathfrak{q}_1(t) = \mathfrak{q}_{12}(t)^*,
\end{align}
and
\begin{align}
\mathfrak{f}_{12} = \mathfrak{F}.
\end{align}
Hermiticity of the Pauli operators implies that $\mathfrak{f}_{12}$ is real.
Hence the ensemble-averaged OTOC $\mathfrak{F}$ is real for this choice of $\rho$.
The ensemble-averaged $\OurKD{\rho}$ has the form
\begin{align}
\label{eq:Brown_ex0}
\mathfrak{A} = \frac{k_1 + k_2 \mathfrak{q}_1 + k_3 \mathfrak{F}  }{16} ,
\end{align}
wherein
\begin{align}
k_1 = (1+v_1)(1 + v_2  + w_3 w_2),
\end{align}
\begin{align}
k_2 =  (1+v_1)(w_3 +w_2)(1+v_2),
\end{align}
and
\begin{align}
\label{eq:Brown_ex3}
k_3 =  (1+v_1) w_3 v_2 w_2.
\end{align}
Equations~\eqref{eq:Brown_ex0}--\eqref{eq:Brown_ex3} imply that
$\mathfrak{A} = 0$ unless $v_1=1$.

The time scale after which $\mathfrak{q}_1$ decays
is order-one.
The time required for $\mathfrak{F}$ to decay is of order $\log \Sites$ (although not necessarily exactly the same as for the $T = \infty$ state).
Therefore, the late-time value of $\mathfrak{A}$ is well approximated by
\begin{align}
\mathfrak{A}_{t \gg 1} = \frac{k_1 + k_3  \,  \mathfrak{F}}{16}.
\end{align}

\subsection{Summary}

This study has the following main messages.
\begin{enumerate}[(1)]
\item In this model, the ensemble-averaged quasiprobability varies
on two time scales.
The first time scale is an order-one relaxation time.
At later times, the OTOC controls the physics entirely.
$F(t)$ varies after a time of order $\log \Sites$.
\item While the late-time physics of $\SumKD{\rho}$ is controlled entirely by
the ensemble-averaged $F(t)$,
the negative values of $\SumKD{\rho}$ show a nonclassicality
that might not be obvious from $F(t)$ alone.
Furthermore, we computed only the first moment of $\SumKD{\rho}$.
The higher moments are likely not determined by $F(t)$ alone.
\item For $T = \infty$, the late-time physics
is qualitatively similar to the late-time physics of
the geometrically local spin chain in Sec.~\ref{section:Numerics}.
\item Nonclassicality, as signaled by negative values of $\SumKD{\rho}$,
is extremely robust.
It survives the long-time limit and the ensemble average.
One might have expected thermalization and interference
to stamp out nonclassicality.
On the other hand, we expect the circuit average
to suppress the imaginary part of $\SumKD{\rho}$ rapidly.
We have no controlled examples in which
$\Im \left( \SumKD{\rho} \right)$ remains nonzero at long times.
Finding further evidence for or against this conjecture
remains an open problem.
\end{enumerate}

\section{Theoretical study of $\OurKD{\rho}$}
\label{section:Theory}

We have discussed experimental measurements,
numerical simulations, and analytical calculations
of the OTOC quasiprobability $\OurKD{\rho}$.
We now complement these discussions with
mathematical properties and physical interpretations.
First, we define an \emph{extended Kirkwood-Dirac distribution}
exemplified by $\OurKD{\rho}$.
We still denote by $\mathcal{B} ( \Hil )$ the set of
bounded operators defined on $\Hil$.

\begin{definition}[$\Ops$-extended Kirkwood-Dirac quasiprobability]
\label{definition:Extend_KD}
Let $\Set{ \ket{a} },  \ldots,  \Set{ \ket{k} }$
and $\Set{ \ket{f} }$ denote orthonormal bases
for the Hilbert space $\mathcal{H}$.
Let $\Oper \in \mathcal{B}( \mathcal{H} )$ denote
a bounded operator defined on $\Hil$.
A \emph{$\Ops$-extended Kirkwood-Dirac quasiprobability}
for $\Oper$ is defined as\footnote{
\label{footnote:ImplicitTime}
Time evolutions may be incorporated into the bases.
For example, Eq.~\eqref{eq:KD_rho_2} features
the 1-extended KD quasiprobability
$\langle f' | a \rangle  \langle a | \rho' | f' \rangle$.
The $\rho'  :=  U_{t'} \rho U_{t'}^\dag$ results from
time-evolving a state $\rho$.
The $\ket{ f' }  :=  U_{ t'' - t' }^\dag  \ket{f}$ results from
time-evolving an eigenket $\ket{f}$ of $F = \sum_f  f  \ketbra{f}{f}$.
We label~\eqref{eq:KD_rho_2} as
$\OurKD{\rho}^\1 ( \rho, a , f )$,
rather than as $\OurKD{\rho}^\1 ( \rho', a , f' )$.
Why? One would measure~\eqref{eq:KD_rho_2}
by preparing $\rho$, evolving the system,
measuring $\A$ weakly, inferring outcome $a$,
evolving the system, measuring $F$, and obtaining outcome $f$.
No outcome $f'$ is obtained.
Our notation is that in~\cite{Dressel_15_Weak}
and is consistent with the notation in~\cite{YungerHalpern_17_Jarzynski}.}
\begin{align}
   \label{eq:Extend_KD}
   \OurKD{\Oper}^\ParenK ( a,  \ldots,  k, f )  :=
   \langle f | k \rangle \langle k | \ldots | a \rangle \langle a | \Oper | f \rangle \, .
\end{align}
\end{definition}

This quasiprobability can be measured
via an extension of the protocol in Sec.~\ref{section:Intro_weak_meas}.
Suppose that $\Oper$ denotes a density matrix.
In each trial, one prepares $\Oper$,
weakly measures the bases sequentially
(weakly measures $\Set{ \ket{a} }$, and so on,
until weakly measuring $\Set{ \ket{k} }$),
then measures $\ketbra{f}{f}$ strongly.

We will focus mostly on density operators
$\Oper = \rho \in \mathcal{D} ( \Hil )$.
One infers $\OurKD{\rho}^\ParenK$ by performing
$2 \Ops - 1$ weak measurements,
and one strong measurement, per trial.
The order in which the bases are measured is
the order in which the labels $a , \ldots, k, f$ appear in
the argument of $\OurKD{\Oper}^\ParenK ( . )$.
The conventional KD quasiprobability is 1-extended.
The OTOC quasiprobability $\OurKD{\rho}$ is 3-extended.

Our investigation parallels the exposition, in Sec.~\ref{section:Intro_to_KD},
of the KD distribution.
First, we present basic mathematical properties.
$\OurKD{\rho}$, we show next,
obeys an analog of Bayes' Theorem.
Our analog generalizes the known analog~\eqref{eq:CondQuasi}.
Our theorem reduces exponentially (in system size)
the memory needed to compute weak values, in certain cases.
Third, we connect $\OurKD{\rho}$ with
the operator-decomposition argument in Sec.~\ref{section:KD_Coeffs}.
$\OurKD{\rho}$ consists of coefficients
in a decomposition of an operator $\rho'$
that results from asymmetrically decohering $\rho$.
Summing $\OurKD{\rho} ( . )$ values yields
a KD representation for $\rho$.
This sum can be used, in experimental measurements
of $\OurKD{\rho}$ and the OTOC, to evaluate
how accurately the desired initial state was prepared.
Fourth, we explore the relationship between
out-of-time ordering and quasiprobabilities.
Time-ordered correlators are moments of quasiprobabilities
that clearly reduce to classical probabilities.
Finally, we generalize beyond the OTOC,
which encodes $\Ops = 3$ time reversals.
Let $\Opsb  :=  \frac{1}{2} ( \Ops + 1 )$.
A \emph{$\Opsb$-fold OTOC} $F^\ParenKB(t)$
encodes $\Ops$ time reversals~\cite{Roberts_16_Chaos,Hael_17_Classification}.
The quasiprobability behind $F^\ParenKB(t)$, we find,
is $\Ops$-extended.

Recent quasiprobability advances involve out-of-time ordering,
including in correlation functions~\cite{Manko_00_Lyapunov,Bednorz_13_Nonsymmetrized,Oehri_16_Time,Hofer_17_Quasi,Lee_17_On}.
Merging these works with the OTOC framework
offers an opportunity for further research (Sec.~\ref{section:Outlook}).

\subsection{Mathematical properties of $\OurKD{\rho}$}
\label{section:TA_Props}

$\OurKD{\rho}$ shares some of its properties
with the KD quasiprobability (Sec.~\ref{section:KDProps}).
Properties of $\OurKD{\rho}$ imply properties of $P(W, W')$,
presented in Appendix~\ref{section:P_Properties}.

\begin{property}  \label{prop:TA_Complex}
The OTOC quasiprobability is a map
$\OurKD{\rho}  \:  :  \:
\mathcal{D} ( \mathcal{H} )  \times
\Set{ v_1 }    \times   \Set{ \DegenV_{v_1} }    \times
\Set{ w_2 }    \times  \Set{ \DegenW_{w_2} }    \times
\Set{ v_2 }    \times   \Set{ \DegenV_{v_2} }      \times
\Set{ w_3 }    \times   \Set{ \DegenW_{w_3} }    \times
\to \mathbb{C} \, .$
The domain is a composition of
the set $\mathcal{D} ( \Hil )$ of density operators defined on $\Hil$
and eight sets of complex numbers.
The range is not necessarily real:
$\mathbb{C}  \supset \mathbb{R}$.
\end{property}

$\OurKD{\rho}$ depends on $H$ and $t$ implicitly through $U$.
The KD quasiprobability in~\cite{Dressel_15_Weak}
depends implicitly on time similarly
(see Footnote~\ref{footnote:ImplicitTime}).
Outside of OTOC contexts, $\mathcal{D} ( \Hil )$
may be replaced with $\mathcal{B} ( \Hil )$.
$\Ops$-extended KD distributions represent bounded operators,
not only quantum states.
$\mathbb{C}$, not necessarily $\mathbb{R}$,
is the range also of
the $\Ops$-fold generalization $\OurKD{\rho}^\ParenK$.
We expound upon the range's complexity
after discussing the number of arguments of $\OurKD{\rho}$.

\emph{Five effective arguments of $\OurKD{\rho}$:}
On the left-hand side of Eq.~\eqref{eq:TADef},
semicolons separate four tuples.
Each tuple results from a measurement, e.g., of $\NondegW$.
We coarse-grained over the degeneracies
in Sections~\ref{section:ProjTrick}--\ref{section:Brownian}.
Hence each tuple often functions as one degree of freedom.
We treat $\OurKD{\rho}$ as a function of
four arguments (and of $\rho$).
The KD quasiprobability has just two arguments (apart from $\Oper$).
The need for four arises from
the noncommutation of $\W(t)$ and $V$.

\emph{Complexity of $\OurKD{\rho}$:}
The ability of $\OurKD{\rho}$ to assume nonreal values
mirrors Property~\ref{prop:Complex} of the KD distribution.
The Wigner function, in contrast, is real.
The OTOC quasiprobability's real component, $\Re ( \OurKD{\rho} )$,
parallels the Terletsky-Margenau-Hill distribution.
We expect nonclassical values of $\OurKD{\rho}$
to reflect nonclassical physics,
as nonclassical values of the KD quasiprobability do
(Sec.~\ref{section:Intro_to_KD}).

Equations~\eqref{eq:TADef} and~\eqref{eq:TAForm}
reflect the ability of $\OurKD{\rho}$ to assume nonreal values.
Equation~\eqref{eq:TADef} would equal
a real product of probabilities
if the backward-process amplitude $A_\rho^*$
and the forward-process amplitude $A_\rho$
had equal arguments.
But the arguments typically do not equal each other.
Equation~\eqref{eq:TAForm} reveals conditions under which
$\OurKD{\rho} ( . )  \in  \mathbb{R}$ and $\not\in \mathbb{R}$.
We illustrate the $\in$ case with two examples
and the $\not\in$ case with one example.

\begin{example}[Real $\OurKD{\rho}$ \#1:
$t = 0$, shared eigenbasis, arbitrary $\rho$]
\label{ex:Real1}
Consider $t  =  0$, at which $U  =  \id$.
The operators $\W(t) = \W$ and $V$ share an eigenbasis,
under the assumption that $[ \W ,  \,  V ]  =  0$:
$\Set{ \ket{ w_\ell,  \DegenW_{w_\ell} } }
=  \Set{ \ket{ v_\ell,  \DegenV_{v_\ell} } }$.
With respect to that basis,
\begin{align}
   \label{eq:Ex1}
   & \OurKD{\rho} ( v_1,  \DegenV_{v_1}  ;  w_2,  \DegenW_{w_2} ;
   v_2,  \DegenV_{v_2}  ;  w_3,  \DegenW_{w_3}  )
   \nonumber \\ &
   =  \left(  \delta_{w_3 v_2}  \delta_{ \DegenW_{w_3} \DegenV_{v_2} }  \right)
   \left(  \delta_{v_2 w_2}  \delta_{ \DegenV_{v_2}  \DegenW_{w_2} }  \right)
   \left(  \delta_{w_2  v_1}  \delta_{ \DegenW_{w_2} \DegenV_{v_1} }   \right)
   \nonumber \\ &  \quad  \times
   \sum_j  p_j
   | \langle w_3,  \DegenW_{w_3}  |  j  \rangle |^2
    \nonumber \\ &
   \in  \mathbb{R}  \, .
\end{align}
We have substituted into Eq.~\eqref{eq:TAForm}.
We substituted in for $\rho$ from Eq.~\eqref{eq:Rho}.
\end{example}

Example~\ref{ex:Real1} is consistent with
the numerical simulations in Sec.~\ref{section:Numerics}.
According to Eq.~\eqref{eq:Ex1}, at $t = 0$,
$\sum_{\rm degeneracies} \OurKD{\rho}  =:  \SumKD{\rho}  \in \mathbb{R}$.
In Figures~\ref{fig:T1_zz_AI},~\ref{fig:rand_zz_AI}, and~\ref{fig:xup_zz_AI},
the imaginary parts $\Im ( \SumKD{\rho} )$
clearly vanish at $t = 0$.
In Fig.~\ref{fig:TInf_zz_AI}, $\Im ( \SumKD{\rho} )$ vanishes
to within machine precision.\footnote{
The $\Im ( \SumKD{\rho} )$ in Fig.~\ref{fig:TInf_zz_AI}
equals zero identically, if $w_2 = w_3$ and/or if $v_1  =  v_2$.
For general arguments,
\begin{align}
   \label{eq:Imag_zero}
   \Im  \LParen  \SumKD{\rho} ( v_1, w_2, v_2, w_3 ) \RParen
   & = \frac{1}{ 2 i }  \:  \Big[  \OurKD{\rho} ( v_1, w_2, v_2, w_3 )
   \nonumber \\ & \qquad
   -  \OurKD{\rho}^* ( v_1, w_2, v_2, w_3 )  \Big]  \, .
\end{align}
The final term equals
\begin{align}
   & \left[  \Tr \left(  \ProjWt{ w_3 }  \ProjV{ v_2 }  \ProjWt{ w_2 }
                     \ProjV{ v_1 }  \right)  \right]^*
   =  \Tr \left(  \ProjV{ v_1 }  \ProjWt{ w_2 }  \ProjV{ v_2 }
                      \ProjWt{ w_3 }  \right)
   \\ & \quad  \label{eq:Imag_zero_help}
   =  \Tr \left(  \ProjWt{ w_2 }  \ProjV{ v_2 } \ProjWt{ w_3 }
                     \ProjV{ v_1 }  \right)
   =  \SumKD{\rho}  ( v_1 , w_3 , v_2 , w_2 )  \, .
\end{align}
The first equality follows from projectors' Hermiticity;
and the second, from the trace's cyclicality.
Substituting into Eq.~\eqref{eq:Imag_zero} shows that
$\SumKD{\rho} ( . )$ is real if $w_2 = w_3$.
$\SumKD{\rho} ( . )$ is real if $v_1 = v_2$,
by an analogous argument.}

Consider a $\rho$ that lacks coherences relative to
the shared eigenbasis, e.g., $\rho  =  \id / \Dim$.
Example~\ref{ex:Real1} implies that
$\Im \left( \OurKD{( \id / \Dim)} \right)$ at $t = 0$.
But $\Im \left( \OurKD{( \id / \Dim)} \right)$ remains zero for all $t$
in the numerical simulations.
Why, if time evolution deforms the $\W(t)$ eigenbasis
from the $V$ eigenbasis?
The reason appears to be a cancellation, as in Example~\ref{ex:Real2}.

Example~\ref{ex:Real2} requires more notation.
Let us focus on a chain of $\Sites$ spin-$\frac{1}{2}$
degrees of freedom.
Let $\sigma^\alpha$ denote the $\alpha = x, y, z$ Pauli operator.
Let $\ket{ \sigma^\alpha,  \pm }$ denote
the $\sigma^\alpha$ eigenstates, such that
$\sigma^\alpha  \ket{ \sigma^\alpha,  \pm }
=  \pm  \ket{ \sigma^\alpha,  \pm }$.
$\Sites$-fold tensor products are denoted by
$\ket{ \bm{ \sigma^\alpha,  \pm } }  :=
\ket{ \sigma^\alpha,  \pm }^{\otimes \Sites }$.
We denote by $\sigma_j^\alpha$
the $\alpha^\th$ Pauli operator
that acts nontrivially on site $j$.

\begin{example}[Real $\OurKD{\rho}$ \#2:
$t = 0$, nonshared eigenbases, $\rho = \id / \Dim$]
\label{ex:Real2}
Consider the spin chain at $t = 0$, such that $U  =  \id$.
Let $\W = \sigma_1^z$ and $V  =  \sigma_\Sites^y$.
Two $\W$ eigenstates are $\ket{ \bm{ \sigma^z, \pm } }$.
Two $V$ eigenstates are
$\ket{ \bm{ \sigma^y,  +  } }
=  \left[  \frac{1}{ \sqrt{2} }  \:
\left( \ket{ \sigma^z, + }
+  i  \ket{ \sigma^z, - }  \right)  \right]^{ \otimes \Sites }$
and  $\ket{ \bm{ \sigma^y,  -  } }
=  \left[  \frac{1}{ \sqrt{2} }  \:
\left( \ket{ \sigma^z, + }
-  i  \ket{ \sigma^z, - }  \right)  \right]^{ \otimes \Sites }$.
The overlaps between the $\W$ eigenstates
and the $V$ eigenstates are
\begin{align}
   \label{eq:InnerPs}
   & \langle \bm{ \sigma^z, + }  |  \bm{ \sigma^y, + }  \rangle
   =  \left(  \frac{ 1 }{ \sqrt{2} }  \right)^\Sites \, ,
   \nonumber \\ &
   \langle \bm{ \sigma^z, + }  |  \bm{ \sigma^y, - }  \rangle
   =  \left(  \frac{ 1 }{ \sqrt{2} }  \right)^\Sites  \, ,
   \nonumber \\ &
   \langle \bm{ \sigma^z, - }  |  \bm{ \sigma^y, + }  \rangle
   = \left(  \frac{ i }{ \sqrt{2} }  \right)^\Sites \, ,
   \; \text{and} \nonumber \\ &
   \langle \bm{ \sigma^z, - }  |  \bm{ \sigma^y, - }  \rangle
   = \left(  \frac{ - i }{ \sqrt{2} }  \right)^\Sites \, .
\end{align}

Suppose that $\rho = \id / \Dim$.
$\OurKD{ ( \id / \Dim) } ( . )$ would have a chance of being nonreal
only if some $\ket{ v_\ell ,  \DegenV_{v_\ell} }$ equaled
$\ket{ \bm{ \sigma^z,  - } }$.
That $\ket{ \bm{ \sigma^z,  - } }$ would introduce
an $i$ into Eq.~\eqref{eq:TAForm}.
But $\bra{ \bm{ \sigma^z,  - } }$ would introduce another $i$.
The product would be real.
Hence $\OurKD{ ( \id / \Dim) } ( . )  \in  \mathbb{R}$.
\end{example}
\noindent $\OurKD{\rho}$ is nonreal in the following example.

\begin{example}[Nonreal $\OurKD{\rho}$:
$t = 0$, nonshared eigenbases, $\rho$ nondiagonal relative to both]
\label{ex:Nonreal1}
Let $t$, $\W$, $V$, $\Set{ \ket{ w_\ell,  \DegenW_{w_\ell} } }$,
and $\Set{ \ket{ v_m,  \DegenV_{v_m} } }$
be as in Example~\ref{ex:Real2}.

Suppose that $\rho$ has coherences relative to
the $\W$ and $V$ eigenbases.
For instance, let
$\rho  =  \ketbra{ \bm{ \sigma^x, + } }{ \bm{ \sigma^x, + } }$.
Since $\ket{ \sigma^x, + }  =  \frac{1}{ \sqrt{2} }  \:
( \ket{ \sigma^z, + }  +  \ket{ \sigma^z, - } )$,
\begin{align}
   \rho  & =  \frac{1}{ 2^\Sites }  \:
   (  \ketbra{ \sigma^z, + }{ \sigma^z, + }
   +  \ketbra{ \sigma^z, + }{ \sigma^z, - }
   \nonumber \\ & \qquad
   +  \ketbra{ \sigma^z, - }{ \sigma^z, + }
   +  \ketbra{ \sigma^z, - }{ \sigma^z, - }  )^{ \otimes \Sites }  \, .
\end{align}

Let $\ket{ w_3, \DegenW_{w_3} } =  \ket{ \bm{ \sigma^z, - } }$,
such that its overlaps with $V$ eigenstates can contain $i$'s.
The final factor in Eq.~\eqref{eq:TAForm} becomes
\begin{align}
   \langle  v_1,  \DegenV_{v_1}  |  \rho  |  w_3,  \DegenW_{w_3}  \rangle
   & =  \frac{1}{ 2^\Sites }  \Big[
   \langle  v_1,  \DegenV_{v_1}  |
   \left(  \ket{ \sigma^z,  +  }^{ \otimes \Sites }  \right)
   \nonumber \\ & \qquad
   +   \langle  v_1,  \DegenV_{v_1}  |
   \left(  \ket{ \sigma^z,  -  }^{ \otimes \Sites }  \right)
   \Big]  \, .
\end{align}
The first inner product evaluates to
$\left(  \frac{ 1 }{ \sqrt{2} }  \right)^\Sites$,
by Eqs.~\eqref{eq:InnerPs}.
The second inner product evaluates to
$\left( \pm  \frac{ i }{ \sqrt{2} }  \right)^\Sites$.
Hence
\begin{align}
   \langle  v_1,  \DegenV_{v_1}  |  \rho  |  w_3,  \DegenW_{w_3}  \rangle
   =  \frac{1}{ 2^{2 \Sites} }
   \left[  1  +  \left(  \pm i \right)^\Sites  \right] \, .
\end{align}
This expression is nonreal if $\Sites$ is odd.
\end{example}

Example~\ref{ex:Nonreal1}, with the discussion after Example~\ref{ex:Real1},
shows how interference can eliminate nonreality from a quasiprobability.
In Example~\ref{ex:Nonreal1},
$\Im \left( \OurKD{\rho} \right)$ does not necessarily vanish.
Hence the coarse-grained $\Im \left( \SumKD{\rho} \right)$
does not obviously vanish.
But $\Im \left( \SumKD{\rho} \right) = 0$
according to the discussion after Example~\ref{ex:Real1}.
Summing Example~\ref{ex:Nonreal1}'s nonzero
$\Im \left( \OurKD{\rho}  \right)$ values
must quench the quasiprobability's nonreality.
This quenching illustrates how interference
can wash out quasiprobabilities' nonclassicality.
Yet interference does not always wash out nonclassicality.
Section~\ref{section:Numerics} depicts
$\SumKD{\rho}$'s that have nonzero imaginary components
(Figures~\ref{fig:T1_zz_AI},~\ref{fig:rand_zz_AI}, and~\ref{fig:xup_zz_AI}).

Example~\ref{ex:Nonreal1} resonates with a finding in~\cite{Solinas_15_Full,Solinas_16_Probing}.
Solinas and Gasparinetti's quasiprobability assumes nonclassical values
when the initial state
has coherences relative to the energy eigenbasis.

\begin{property}
\label{prop:MargOurKD}
Marginalizing $\OurKD{\rho} ( . )$ over all its arguments
except any one
yields a probability distribution.
\end{property}

   Consider, as an example, summing  Eq.~\eqref{eq:TAForm} over
   every tuple except $( w_3, \DegenW_{w_3} )$.
   The outer products become resolutions of unity, e.g.,
   $\sum_{ ( w_2, \DegenW_{w_2} ) }
      \ketbra{ w_2, \DegenW_{w_2} }{ w_2, \DegenW_{w_2} }
      = \id$.
   A unitary cancels with its Hermitian conjugate:
   $U^\dag U = \id$.
   The marginalization yields
   $\langle  w_3, \DegenW_{w_3}  |  U  \rho  U^\dag  |
                      w_3, \DegenW_{w_3}  \rangle$.
   This expression equals the probability that
   preparing $\rho$, time-evolving,
   and measuring the $\NondegW$ eigenbasis
   yields the outcome $( w_3, \DegenW_{w_3} )$.

This marginalization property,
with the structural and operational resemblances
between $\OurKD{\rho}$
and the KD quasiprobability,
accounts for our calling $\OurKD{\rho}$ an extended quasiprobability.
The general $\Ops$-extended $\OurKD{\rho}^\ParenK$
obeys Property~\ref{prop:MargOurKD}.

\begin{property}[Symmetries of $\OurKD{ ( \id / \Dim) }$]
\label{property:Syms}
Let $\rho$ be the infinite-temperature Gibbs state $\id / \Dim$.
The OTOC quasiprobability $\OurKD{ ( \id / \Dim) }$
has the following symmetries.
\begin{enumerate}[(A)]

   \item  \label{item:Sym1}
   $\OurKD{ ( \id / \Dim) } ( . )$ remains invariant under
   the simultaneous interchanges of
   $( w_2,  \DegenW_{w_2} )$ with $( w_3,  \DegenW_{w_3} )$
   and $( v_1,  \DegenV_{v_1} )$ with $( v_2,  \DegenV_{v_2} )$:
   $\OurKD{ ( \id / \Dim) } ( v_1,  \DegenV_{v_1} ;  w_2,  \DegenW_{w_2} ;
   v_2,  \DegenV_{v_2}  ;  w_3,  \DegenW_{w_3}  )
     =  \OurKD{ ( \id / \Dim) } ( v_2,  \DegenV_{v_2}  ;  w_3,  \DegenW_{w_3} ;
     v_1,  \DegenV_{v_1} ;  w_2,  \DegenW_{w_2} )$.

   \item  \label{item:Sym2}
   Let $t = 0$, such that
   $\Set{ \ket{ w_\ell,  \DegenW_{w_\ell} } }
   =  \Set{ \ket{ v_\ell,  \DegenV_{v_\ell} } }$
   (under the assumption that $[\W,  V]  =  0$).
   $\OurKD{ ( \id / \Dim) } ( . )$ remains invariant under
   every cyclic permutation of its arguments.

\end{enumerate}
\end{property}

Equation~\eqref{eq:TAForm} can be recast as a trace.
Property~\ref{property:Syms} follows from the trace's cyclicality.
Subproperty~\ref{item:Sym2} relies on the triviality of
the $t = 0$ time-evolution operator: $U = \id$.
The symmetries lead to degeneracies visible in numerical plots
(Sec.~\ref{section:Numerics}).

Analogous symmetries characterize a \emph{regulated} quasiprobability.
Maldacena \emph{et al.} regulated $F(t)$
to facilitate a proof~\cite{Maldacena_15_Bound}:\footnote{
The name ``regulated'' derives from quantum field theory.
$F(t)$ contains operators $\W^\dag(t)$ and $\W(t)$
defined at the same space-time point
(and operators $V^\dag$ and $V$ defined at the same space-time point).
Products of such operators encode divergences.
One can regulate divergences
by shifting one operator to another space-time point.
The inserted $\rho^{1/4}  =  \frac{1}{Z^{1/4} }  \;  e^{ - H / 4 T }$
shifts operators along an imaginary-time axis.}
\begin{align}
   \label{eq:RegOTOC_def}
   F_\reg (t)  :=  \Tr \left(  \rho^{1/4}  \W(t)  \rho^{1/4}  V
   \rho^{1/4}  \W(t)  \rho^{1/4}  V  \right)  \, .
\end{align}
$F_\reg(t)$ is expected to behave roughly like $F(t)$~\cite{Maldacena_15_Bound,Yao_16_Interferometric}.
Just as $F(t)$ equals a moment of
a sum over $\OurKD{\rho}$,
$F_\reg(t)$ equals a moment of a sum over
\begin{align}
   \label{eq:RegKD}
   &\OurKD{\rho}^\reg  ( v_1,  \DegenV_{v_1} ; w_2,  \DegenW_{w_2} ;
   v_2,  \DegenV_{v_2}  ;  w_3,  \DegenW_{w_3}  )
   \\ \nonumber &
   :=  \langle  w_3,  \DegenW_{w_3}  |  U  \rho^{1/4}  |
        v_2,  \DegenV_{v_2}  \rangle
        \langle  v_2,  \DegenV_{v_2}  |  \rho^{1/4}  U^\dag  |
        w_2,  \DegenW_{w_2}  \rangle
        \\ \nonumber  & \qquad \times
        \langle  w_2,  \DegenW_{w_2}  |  U  \rho^{1/4}  |
        v_1,  \DegenV_{v_1}  \rangle
        \langle  v_1,  \DegenV_{v_1}  |  \rho^{1/4}  U^\dag  |
        w_3,  \DegenW_{w_3}  \rangle \\
   &   \label{eq:RegKD2}
   \equiv \langle  w_3,  \DegenW_{w_3}  |  \tilde{U}  |
        v_2,  \DegenV_{v_2}  \rangle
        \langle  v_2,  \DegenV_{v_2}  |  \tilde{U}^\dag  |
        w_2,  \DegenW_{w_2}  \rangle
        \\ \nonumber  & \qquad \times
        \langle  w_2,  \DegenW_{w_2}  |  \tilde{U}  |
        v_1,  \DegenV_{v_1}  \rangle
        \langle  v_1,  \DegenV_{v_1}  |  \tilde{U}^\dag  |
        w_3,  \DegenW_{w_3}  \rangle  \, .
\end{align}
The proof is analogous to the proof of Theorem~1 in~\cite{YungerHalpern_17_Jarzynski}.
Equation~\eqref{eq:RegKD2} depends on
$\tilde{U}  :=  \frac{1}{Z}  \:  e^{ - i H \tau }$,
which propagates in the complex-time variable
$\tau :=  t  -  \frac{i}{4 T}$.
The Hermitian conjugate $\tilde{U}^\dag  =  \frac{1}{Z}  \:  e^{ i H \tau^* }$
propagates along $\tau^* = t + \frac{i }{ 4 T }$.

$\OurKD{ \left( e^{ - H / T } / Z \right) }^\reg$ has
the symmetries of $\OurKD{ ( \id / \Dim) }$
(Property~\ref{property:Syms}) for arbitrary $T$.
One might expect $\OurKD{\rho}^\reg$ to behave
similarly to $\OurKD{\rho}$,
as $F_\reg(t)$ behaves similarly to $F(t)$.
Numerical simulations largely support this expectation.
We compared $\SumKD{\rho} ( . )$ with
$\SumKD{ \rho }^\reg ( . )
:=  \sum_{\rm degeneracies} \OurKD{\rho}^{ \reg } ( . ) \, .$
The distributions vary significantly over similar time scales
and have similar shapes.
$\SumKD{\rho}^\reg$ tends to have a smaller imaginary component
and, as expected, more degeneracies.

The properties of $\OurKD{\rho}$ imply properties of $P(W, W')$.
We discuss these properties in Appendix~\ref{section:P_Properties}.

%
%
\subsection{Bayes-type theorem and retrodiction with $\OurKD{\rho}$}
\label{section:TA_retro}

We reviewed, in Sec.~\ref{section:KD_Retro},
the KD quasiprobability's role in retrodiction.
The KD quasiprobability $\OurKD{\rho}^\1$ generalizes the nontrivial part
$\Re ( \langle f' | a \rangle  \langle a | \rho' | f' \rangle )$
of a conditional quasiprobability $\tilde{p} ( a | \rho, f )$
used to retrodict about an observable $\A$.
Does $\OurKD{\rho}$ play a role similar to $\OurKD{\rho}^\1$?

It does. To show so, we generalize Sec.~\ref{section:KD_Retro}
to composite observables.
Let $\A, \B, \ldots, \K$ denote $\Ops$ observables.
$\K  \ldots  \B  \A$ might not be Hermitian but can be symmetrized.
For example, $\Gamma  :=  \K \ldots  \A  +  \A \ldots \K$
is an observable.\footnote{
So is $\tilde{\Gamma}  :=  i ( \K  \ldots  \A  -  \A  \ldots  \K )$.
An operator can be symmetrized in multiple ways.
Theorem~\ref{theorem:RetroK} governs $\Gamma$.
Appendix~\ref{section:RetroK2} contains
an analogous result about $\tilde{\Gamma}$.
Theorem~\ref{theorem:RetroK} extends trivially to
Hermitian (already symmetrized) instances of $\K \ldots \A$.
Corollary~\ref{corollary:RetroOurKD} illustrates this extension.}
Which value is most reasonably attributable to
$\Gamma$ retrodictively?
A weak value $\Gamma_\weak$ given by Eq.~\eqref{eq:WeakVal}.
We derive an alternative expression for $\Gamma_\weak$.
In our expression, $\Gamma$ eigenvalues
are weighted by $\Ops$-extended KD quasiprobabilities.
Our expression reduces exponentially, in the system's size,
the memory required to calculate weak values,
under certain conditions.
We present general theorems about $\OurKD{\rho}^\ParenK$,
then specialize to the OTOC $\OurKD{\rho}$.

%
%
\begin{theorem}[Retrodiction about composite observables]
\label{theorem:RetroK}
Consider a system $\Sys$ associated with a Hilbert space $\Hil$.
For concreteness, we assume that $\Hil$ is discrete.
Let $\A  =  \sum_a  a  \ketbra{a}{a} \, ,  \ldots ,
\K  =  \sum_k  k  \ketbra{k}{k}$ denote
$\Ops$ observables defined on $\mathcal{H}$.
Let $U_t$ denote the family of unitaries
that propagates the state of $\Sys$
along time $t$.

Suppose that $\Sys$ begins in the state
$\rho$ at time $t = 0$,
then evolves under $U_{t''}$ until $t = t''$.
Let $F = \sum_f  f  \ketbra{f}{f}$ denote an observable
measured at $t = t''$.
Let $f$ denote the outcome.
Let $t'  \in  (0,  t'')$ denote an intermediate time.
Define $\rho'  :=  U_{t'}  \rho  U_{t'}^\dag$   and
$\ket{f'}  :=  U^\dag_{t'' - t'}  \ket{f}$ as time-evolved states.

The value most reasonably attributable retrodictively to the time-$t'$
$\Gamma  :=  \K  \ldots  \A  +  \A \ldots \K$
is the weak value
\begin{align}
   \label{eq:GammaW}
   \Gamma_\weak ( \rho , f )
   & = \sum_{a, \ldots, k }  (a \ldots k)
   \Big[ \tilde{p}_\rightarrow (a, \ldots, k | \rho, f )
           \nonumber \\ & \qquad \qquad  +
           \tilde{p}_\leftarrow ( k, \ldots, a | \rho, f )
   \Big] \, .
\end{align}
The weights are joint conditional quasiprobabilities.
They obey analogs of Bayes' Theorem:
\begin{align}
   \tilde{p}_\rightarrow ( a, \ldots, k | \rho , f )
   & =  \label{eq:QuasiBayesLeft1}
   \frac{ \tilde{p}_\rightarrow ( a, \ldots, k , f | \rho ) }{
             p ( f | \rho ) }   \\
   & \equiv  \label{eq:QuasiBayesLeft2}
   \frac{ \Re ( \langle f' | k \rangle \langle k |  \ldots
                      | a \rangle \langle a | \rho' | f' \rangle ) }{
             \langle f'  | \rho' | f' \rangle }\, ,
\end{align}
and
\begin{align}
    \tilde{p}_\leftarrow ( k, \ldots, a | \rho , f )
    & =  \label{eq:QuasiBayesRt1}
            \frac{ \tilde{p}_\leftarrow ( k, \ldots, a, f | \rho ) }{
                      p ( f | \rho ) } \\
    & \equiv  \label{eq:QuasiBayesRt2}
    \frac{ \Re ( \langle f' | a \rangle \langle a |  \ldots
                      | k \rangle \langle k | \rho' | f' \rangle ) }{
             \langle f'  | \rho' | f' \rangle }    \, .
\end{align}
Complex generalizations of the weights' numerators,
\begin{align}
   \label{eq:Extend_KD2}
   \OurKD{ \rho ,  \rightarrow }^\ParenK ( a, \ldots, k, f )
   :=  \langle f' | k \rangle  \langle k |  \ldots  | a \rangle
         \langle a | \rho' | f' \rangle
\end{align}
and
\begin{align}
   \label{eq:Extend_KD1}
   \OurKD{ \rho ,  \leftarrow }^\ParenK  ( k, \ldots, a, f )
   :=  \langle f' | a \rangle  \langle a |  \ldots  | k \rangle
        \langle k |  \rho' | f'  \rangle   \, ,
\end{align}
are $\Ops$-extended KD distributions.
\end{theorem} \noindent
A rightward-pointing arrow $\rightarrow$
labels quantities in which the outer products,
$\ketbra{k}{k},  \ldots,  \ketbra{a}{a}$, are ordered analogously to
the first term $\K \ldots \A$ in $\Gamma$.
A leftward-pointing arrow $\leftarrow$
labels quantities in which
reading the outer products $\ketbra{a}{a}, \ldots, \ketbra{k}{k}$
backward---from right to left---parallels
reading $\K \ldots \A$ forward.

\begin{proof}
The initial steps come from~\cite[Sec. II A]{Dressel_15_Weak},
which recapitulates~\cite{Johansen_04_Nonclassical,Hall_01_Exact,Hall_04_Prior}.
For every measurement outcome $f$, we assume,
some number $\gamma_f$ is
the guess most reasonably attributable to $\Gamma$.
We combine these best guesses into the effective observable
$\Gest  :=  \sum_f  \gamma_f  \ketbra{f'}{f'}$.
We must optimize our choice of $\Set{ \gamma_f }$.
We should quantify the distance between
(1) the operator $\Gest$ we construct and
(2) the operator $\Gamma$ we wish to infer about.
We use the weighted trace distance
\begin{align}
   \label{eq:OprDist}
   \mathscr{D}_{\rho'} ( \Gamma,  \Gest )  =
   \Tr  \left(  \rho'  [  \Gamma  -  \Gest ]^2  \right) \, .
\end{align}
$\rho'$ serves as a ``positive prior bias''~\cite{Dressel_15_Weak}.

Let us substitute in for the form of $\Gest$.
Expanding the square, then invoking the trace's linearity, yields
\begin{align}
   \label{eq:DHelp1}
   & \mathscr{D}_{ \rho' } ( \Gamma,  \Gest )  =
  \Tr ( \rho' \Gamma^2 )
  +  \sum_f  \Big[
      \gamma_f^2    \langle  f'  |  \rho'  |  f'  \rangle
      \nonumber \\ & \qquad \quad
       -  \gamma_f  ( \langle f' | \rho' \Gamma | f' \rangle
           +  \langle f' | \Gamma \rho' | f' \rangle )
   \Big] \, .
\end{align}
The parenthesized factor equals
$2 \Re ( \langle f' | \Gamma \rho' | f' \rangle )$.
Adding and subtracting
$\sum_f  \langle f' | \rho' | f' \rangle
[  \Re ( \langle f' | \Gamma \rho' | f' \rangle )  ]^2$
to and from Eq.~\eqref{eq:DHelp1},
we complete the square:
\begin{align}
   \label{eq:DResult}
   & \mathscr{D}_{ \rho' } ( \Gamma,  \Gest )  =
   \Tr ( \rho' \Gamma^2 )
   -  \sum_f  \langle f' | \rho' | f' \rangle
   [  \Re ( \langle f' | \Gamma \rho' | f' \rangle )  ]^2
   \nonumber \\ &  \qquad \qquad
   +  \sum_f  \langle f' | \rho' | f' \rangle
   \Bigg(  \gamma_f  -
   \frac{ \Re ( \langle f' | \Gamma \rho' | f' \rangle ) }{
   \langle f' | \rho' | f' \rangle }    \Bigg)^2  .
\end{align}

Our choice of $\Set{ \gamma_f }$ should minimize the distance~\eqref{eq:DResult}.
We should set the square to zero:
\begin{align}
   \label{eq:Choose}
   \gamma_f  =
   \frac{ \Re ( \langle f' | \Gamma \rho' | f' \rangle ) }{
   \langle f' | \rho' | f' \rangle }  \, .
\end{align}

Now, we deviate from~\cite{Johansen_04_Nonclassical,Hall_01_Exact,Hall_04_Prior,Dressel_15_Weak}.
We substitute the definition of $\Gamma$ into Eq.~\eqref{eq:Choose}.
Invoking the linearity of $\Re$ yields
\begin{align}
   \label{eq:Choose2}
   \gamma_f  =  \frac{
   \Re ( \langle f' | \K  \ldots  \A \rho' | f' \rangle ) }{
   \langle f' | \rho' | f' \rangle }
   + \frac{  \Re ( \langle f' | \A  \ldots  \K \rho' | f' \rangle ) }{
   \langle f' | \rho' | f' \rangle }  \, .
\end{align}
We eigendecompose $\A, \ldots, \K$.
The eigenvalues, being real,
can be factored out of the $\Re$'s.
Defining the eigenvalues' coefficients
as in Eqs.~\eqref{eq:QuasiBayesLeft2} and~\eqref{eq:QuasiBayesRt2},
we reduce Eq.~\eqref{eq:Choose2} to the form in Eq.~\eqref{eq:GammaW}.
\end{proof}

Theorem~\ref{theorem:RetroK} reduces exponentially,
in system size, the space required
to calculate $\Gamma_\weak$, in certain cases.\footnote{
``Space'' means ``memory,'' or ``number of bits,'' here.}
For concreteness, we focus on a multiqubit system
and on $l$-local operators $\A, \ldots, \K$.
An operator $\mathcal{O}$ is \emph{$l$-local} if
$\mathcal{O} =  \sum_j  \mathcal{O}_j$,
wherein each $\mathcal{O}_j$ operates nontrivially on,
at most, $l$ qubits.
Practicality motivates this focus:
The lesser the $l$, the more easily
$l$-local operators can be measured.

We use asymptotic notation from computer science:
Let $f \equiv f( \Sites )$ and $g \equiv g (\Sites)$ denote
any functions of the system size.
If $g  =  O( f )$, $g$ grows no more quickly than
(is upper-bounded by)
a constant multiple of $f$
in the asymptotic limit, as $\Sites \to \infty$.
If $g  =  \Omega ( f )$, $g$ grows at least as quickly as
(is lower-bounded by) a constant multiple of $f$
in the asymptotic limit.
If $g  =  \Theta ( f )$, $g$ is upper- and lower-bounded by $f$:
$g  =  O ( f )$,  and  $g = \Omega ( f )$.
If $g  =  o ( f )$, $g$ shrinks strictly more quickly than $f$
in the asymptotic limit.

\begin{theorem}[Weak-value space saver]
\label{theorem:Space}

Let $\Sys$ denote a system of $\Sites$ qubits.
Let $\Hil$ denote the Hilbert space associated with $\Sys$.
Let $\ket{ f' }  \in  \Hil$ denote a pure state
and $\rho' \in \mathcal{D} ( \Hil )$ denote a density operator.
Let $\Basis$ denote any fixed orthonormal basis for $\Hil$
in which each basis element equals a tensor product
of $\Sites$ factors, each of which operates nontrivially on
exactly one site.
$\Basis$ may, for example, consist of tensor products of
$\sigma^z$ eigenstates.

Let $\Ops$ denote any polynomial function of $\Sites$:
$\Ops  \equiv  \Ops( \Sites )  =  {\rm poly}( \Sites )$.
Let $\A , \ldots, \K$ denote
$\Ops$ traceless $l$-local observables defined on $\Hil$,
for any constant $l $.
Each observable may, for example, be a tensor product of
$\leq l$ nontrivial Pauli operators and $\geq \Sites  -  l$ identity operators.
The composite observable $\Gamma :=  \A \ldots \K  +  \K \ldots \A$
is not necessarily $l$-local.
Let $\A = \sum_a a \ketbra{a}{a} \, ,  \ldots, \K = \sum_k k \ketbra{k}{k}$
denote eigenvalue decompositions of the local observables.
Let $\mathcal{O}_{ \Basis }$ denote the matrix that represents
an operator $\mathcal{O}$ relative to $\Basis$.

Consider being given the matrices $\A_\Basis, \ldots, \K_\Basis$,
$\rho'_\Basis$, and $\ket{ f' }_\Basis$.
From this information, the weak value $\Gamma_\weak$
can be computed in two ways:
\begin{enumerate}[(1)]

   \item \label{item:Conven}
            \textbf{Conventional method}
   \begin{enumerate}[(A)]
      \item Multiply and sum given matrices to form
      $\Gamma_{ \Basis }  =  \K_\Basis \ldots \A_\Basis
                                            +  \A_\Basis \ldots \K_\Basis$.
      \item Compute $\langle f' | \rho' | f' \rangle
               =  \langle f' |_\Basis  \: \rho'_\Basis  \: | f' \rangle_\Basis$.
      \item Substitute into $\Gamma_\weak  =  \Re \left(
      \frac{ \langle f' |_\Basis  \:  \Gamma_\Basis  \:   \rho'_\Basis  \:
               | f' \rangle_\Basis }{ \langle f' | \rho' | f' \rangle }
               \right) \, .$
   \end{enumerate}

   \item \label{item:KFac}
            \textbf{$\Ops$-factored method}
   \begin{enumerate}[(A)]
      \item Compute $\langle f' | \rho' | f' \rangle$.
      \item \label{eq:TermStep}
               For each nonzero term in Eq.~\eqref{eq:GammaW},
               \begin{enumerate}[(i)]
                  \item  calculate $\tilde{p}_\rightarrow ( . )$ and $\tilde{p}_\leftarrow ( . )$
               from Eqs.~\eqref{eq:QuasiBayesLeft2} and~\eqref{eq:QuasiBayesRt2}.
                  \item  substitute into Eq.~\eqref{eq:GammaW}.
               \end{enumerate}
   \end{enumerate}

\end{enumerate}

Let $\Sigma_{(n)}$ denote the space required to compute $\Gamma_\weak$,
aside from the space required to store $\Gamma_\weak$,
with constant precision,
using method $(n) = \ref{item:Conven}, \ref{item:KFac}$,
in the asymptotic limit.
Method~\ref{item:Conven} requires a number of bits at least
exponential in the number $\Ops$ of local observables:
\begin{align}
   \label{eq:Space1}
   \Sigma_{ \ref{item:Conven} }  =  \Omega \left( 2^\Ops \right) \, .
\end{align}
Method~\ref{item:KFac} requires a number of bits linear in $\Ops$:
\begin{align}
   \label{eq:Space2}
   \Sigma_{ \ref{item:KFac} }  =  O ( \Ops )  \, .
\end{align}
Method~\ref{item:KFac} requires exponentially---in $\Ops$ and so in $\Sites$---less memory than Method~\ref{item:Conven}.
\end{theorem}
\begin{proof}
Using Method~\ref{item:Conven}, one computes $\Gamma_\Basis$.
$\Gamma_\Basis$ is a $2^\Sites \times 2^\Sites$ complex matrix.
The matrix has $\Omega ( 2^\Ops )$ nonzero elements:
$\A , \ldots, \K$ are traceless,
so each of $\A_\Basis, \ldots, \K_\Basis$ contains
at least two nonzero elements.
Each operator at least doubles
the number of nonzero elements in $\Gamma_\Basis$.
Specifying each complex number with constant precision
requires $\Theta(1)$ bits.
Hence Method~\ref{item:Conven} requires $\Omega \left( 2^\Ops \right)$ bits.

Let us turn to Method~\ref{item:KFac}.
We can store $\langle f' | \rho' | f' \rangle$
in a constant number of bits.

Step~\ref{eq:TermStep} can be implemented with
a counter variable $C_\Oper$ for each local operator $\Oper$,
a running-total variable $G$, and a ``current term'' variable $T$.
$C_\Oper$ is used to iterate through
the nonzero eigenvalues of $\Oper$
(arranged in some fiducial order).
$\Oper$ has $O ( 2^l )$ nonzero eigenvalues.
Hence $C_\Oper$  requires $O ( l )$ bits.
Hence the set of $\Ops$ counters $C_\Oper$
requires $O ( l \Ops )  =  O ( \Ops )$ bits.

The following algorithm implements Step~\ref{eq:TermStep}:
\begin{enumerate}[(i)]

   \item   \label{item:CNormal}
   If $C_\K <$ its maximum possible value, proceed as follows:
   \begin{enumerate}[(a)]

      \item  For each $\Oper = \A, \ldots, \K$,
      compute the $(2^{ C_\Oper} )^\th$ nonzero eigenvalue
      (according to the fiducial ordering).

      \item  Multiply the eigenvalues to form $a \ldots k$.
      Store the product in $T$.

      \item  For each $\Oper  =  \A, \ldots, \K$,
      calculate the $(2^{ C_\Oper} )^\th$ eigenvector column
      (according to some fiducial ordering).

      \item  Substitute the eigenvector columns into
      Eqs.~\eqref{eq:QuasiBayesLeft2} and~\eqref{eq:QuasiBayesRt2},
      to compute $\tilde{p}_\rightarrow ( . )$ and $\tilde{p}_\leftarrow ( . )$.

      \item  Form $(a \ldots k) \Big[ \tilde{p}_\rightarrow (a, \ldots, k | \rho, f )
           +  \tilde{p}_\leftarrow ( k, \ldots, a | \rho, f )$.
           Update $T$ to this value.

      \item  Add $T$ to $G$.

      \item  Erase $T$.

      \item  Increment $C_\K$.
   \end{enumerate}

   \item If $C_\K$ equals its maximum possible value,
   increment the counter of the preceding variable, $\mathcal{J}$, in the list;
   reset $C_\K$ to one;
   and, if $\mathcal{J}$ has not attained its maximum possible value,
   return to Step~\ref{item:CNormal}.
   Proceed in this manner---incrementing counters;
   then resetting counters, incrementing preceding counters,
   and returning to Step~\ref{item:CNormal}---until $C_\A$ reaches its maximum possible value. Then, halt.
\end{enumerate}

The space needed to store $G$ is
the space needed to store $\Gamma_\weak$.
This space does not contribute to $\Sigma_{ \ref{item:KFac} }$.

How much space is needed to store $T$?
We must calculate $\Gamma_\weak$ with constant precision.
$\Gamma_\weak$ equals a sum of $2^{ l \Ops }$ terms.
Let $\varepsilon_j$ denote the error in term $j$.
The sum $\sum_{j = 1}^{ 2^{ l \Ops } }  \varepsilon_j$
must be $O(1)$.
This requirement is satisfied if
$2^{ l \Ops } \,  \left( \max_j | \varepsilon_j | \right)  =  o (1)$, which implies
$\max_j | \varepsilon_j |  =  o \left(  2^{ - l \Ops }  \right)$.
We can specify each term, with a small-enough roundoff error,
using $O ( l \Ops )  =  O ( \Ops )$ bits.

Altogether, the variables require $O( \Ops )$ bits.
As the set of variables does, so does the $\Oper$-factored method.
\end{proof}

Performing Method~\ref{item:KFac} requires slightly more time
than performing Method~\ref{item:Conven}.
Yet Theorem~\ref{theorem:Space} can benefit computations
about quantum many-body systems.
Consider measuring a weak value of a quantum many-body system.
One might wish to predict the experiment's outcome
and to compare the outcome with the prediction.
Alternatively, consider simulating quantum many-body systems
independently of laboratory experiments,
as in Sec.~\ref{section:Numerics}.
One must compute weak values numerically,
using large matrices.
The memory required to store these matrices
can limit computations.
Theorem~\ref{theorem:Space} can free up space.

Two more aspects of retrodiction deserve exposition:
related studies and the physical significance of $\K \ldots \A$.

\emph{Related studies:}
Sequential weak measurements have been proposed~\cite{Lundeen_12_Procedure}
and realized recently~\cite{Piacentini_16_Measuring,Suzuki_16_Observation,Thekkadath_16_Direct}.
Lundeen and Bamber proposed a ``direct measurement''
of a density operator~\cite{Lundeen_12_Procedure}.
Let $\rho$ denote a density operator
defined on a dimension-$\Dim$ Hilbert space $\Hil$.
Let $\Basis_a  :=  \Set{ \ket{ a_\ell } }$ and
$\Basis_b  :=  \Set{ \ket{ b_\ell } }$ denote orthonormal
\emph{mutually unbiased bases} (MUBs) for $\Hil$.
The interbasis inner products have constant magnitudes:
$| \langle a_\ell | b_m \rangle |  =  \frac{1}{ \sqrt{ \Dim } }
\;  \forall \ell, m$.
Consider measuring $\Basis_a$ weakly,
then $\Basis_b$ weakly, then $\Basis_a$ strongly,
in each of many trials.
One can infer (1) a KD quasiprobability for $\rho$
and (2) a matrix that represents $\rho$ relative to $\Basis_a$~\cite{Lundeen_12_Procedure}.

KD quasiprobabilities are inferred from experimental measurements in~\cite{Piacentini_16_Measuring,Thekkadath_16_Direct}.
Two weak measurements are performed sequentially also in~\cite{Suzuki_16_Observation}.
Single photons are used in~\cite{Piacentini_16_Measuring,Suzuki_16_Observation}.
A beam of light is used in~\cite{Thekkadath_16_Direct}.
These experiments indicate the relevance
of Theorem~\ref{theorem:RetroK}
to current experimental capabilities.
Additionally, composite observables $\A \B  +  \B \A$
accompany KD quasiprobabilities in e.g.,~\cite{Halliwell_16_Leggett}.

%
\emph{Physical significance of $\K \ldots \A$:}
Rearranging Eq.~\eqref{eq:GammaW} offers insight
into the result:
\begin{align}
   \label{eq:GammaW2}
   \Gamma_\weak ( \rho , f )    & =
   \sum_{ k, \ldots, a }  ( k \ldots a )
   \tilde{p}_\rightarrow  ( k, \ldots, a | \rho , f )
   \nonumber \\ & \qquad   +
   \sum_{a, \ldots, k }  (a \ldots k)
   \tilde{p}_\leftarrow ( a, \ldots, k | \rho , f )   \, .
\end{align}
Each sum parallels the sum in Eq.~\eqref{eq:WeakVal3}.
Equation~\eqref{eq:GammaW2} suggests that we are retrodicting
about $\K \ldots \A$ independently of $\A \ldots \K$.
But neither $\K \ldots \A$ nor $\A \ldots \K$ is Hermitian.
Neither operator seems measurable.
Ascribing a value to neither
appears to have physical significance, \emph{prima facie}.

Yet non-Hermitian products $\B \A$
have been measured weakly~\cite{Piacentini_16_Measuring,Suzuki_16_Observation,Thekkadath_16_Direct}.
Weak measurements associate a value with
the supposedly unphysical $\K \ldots \A$,
just as weak measurements enable us to infer
supposedly unphysical probability amplitudes $\Amp_\rho$.
The parallel between $\K \ldots \A$ and $\Amp_\rho$
can be expanded.
$\K \ldots \A$ and $\A \ldots \K$, being non-Hermitian,
appear to lack physical significance independently.
Summing the operators forms an observable.
Similarly, probability amplitudes $\Amp_\rho$ and $\Amp_\rho^*$
appear to lack physical significance independently.
Multiplying the amplitudes forms a probability.
But $\Amp_\rho$ and $\K \ldots \A$
can be inferred individually from weak measurements.

We have generalized Sec.~\ref{section:KD_Retro}.
Specializing to $k = 3$,
and choosing forms for $\A , \ldots \K$,
yields an application of $\OurKD{\rho}$ to retrodiction.

%
%
\begin{corollary}[Retrodictive application of $\OurKD{\rho}$]
   \label{corollary:RetroOurKD}
   Let $\Sys$, $\Hil$, $\rho$, $\W(t)$, and $V$ be defined
   as in Sec.~\ref{section:SetUp}.
   Suppose that $\Sys$ is in state $\rho$ at time $t = 0$.
   Suppose that the observable
   $F = \W
   =  \sum_{ w_3, \DegenW_{w_3} }  w_3
       \ketbra{ w_3, \DegenW_{w_3} }{ w_3, \DegenW_{w_3} }$
   of $\Sys$ is measured at time $t'' = t$.
   Let $( w_3, \DegenW_{w_3} )$ denote the outcome.
   Let $\A  =  V  =  \sum_{ v_1,  \DegenV_{v_1} }  v_1
                       \ketbra{ v_1,  \DegenV_{v_1} }{ v_1,  \DegenV_{v_1} }$,
   $\B  =  \W(t)  =  \sum_{ w_2,  \DegenW_{w_2} }  w_2  \,
                            U^\dag \ketbra{ w_2,  \DegenW_{w_2} }{
                                                      w_2,  \DegenW_{w_2} }  U$,
   and $\C  =  V  =  \sum_{ v_2,  \DegenV_{v_2} }  v_2
                              \ketbra{ v_2,  \DegenV_{v_2} }{ v_2,  \DegenV_{v_2} }  \, .$
   Let the composite observable $\Gamma =  \A \B \C  =  V  \W(t)  V$.
   The value most reasonably attributable to
   $\Gamma$ retrodictively is the weak value
   \begin{align}
      \label{eq:OTOC_retro1}
      & \Gamma_\weak ( \rho ; w_3, \DegenW_{w_3} )
      =  \sum_{ ( v_1 ,  \DegenV_{v_1} ) ,  ( v_2 ,  \DegenV_{v_2} ),
                      ( w_2 ,  \DegenW_{w_2} ) }
      v_1  w_2  v_2
      \nonumber \\ & \times
      \tilde{p}_{\leftrightarrow}
      ( v_2,  \DegenV_{v_2} ; w_2,  \DegenW_{w_2}  ;
        v_1,  \DegenV_{v_1}   |
      \rho  ;  w_3 , \DegenW_{w_3} )  \, .
   \end{align}
   The weights are joint conditional quasiprobabilities
   that obey an analog of Bayes' Theorem:
   \begin{align}
      \label{eq:OTOC_retro2}
      & \tilde{p}_{\leftrightarrow}
      ( v_1,  \DegenV_{v_1} ; w_2,  \DegenW_{w_2}  ;
        v_2,  \DegenV_{v_2}  |
      \rho  ;  w_3, \DegenW_{w_3}  )
      \nonumber \\ &
      =  \frac{ \tilde{p}_\leftrightarrow
      ( v_1,  \DegenV_{v_1} ; w_2,  \DegenW_{w_2}  ;
        v_2,  \DegenV_{v_2}  ;  w_3, \DegenW_{w_3}   |  \rho )   }{
      p ( w_3, \DegenW_{w_3}  | \rho ) }  \\
      & \equiv  \Re (
      \langle w_3, \DegenW_{w_3} | U | v_2,  \DegenV_{v_2} \rangle
      \langle v_2,  \DegenV_{v_2} | U^\dag | w_2,  \DegenW_{w_2} \rangle
      \nonumber \\ & \qquad \times
      \langle w_2,  \DegenW_{w_2} | U | v_1,  \DegenV_{v_1} \rangle
      \langle v_1,  \DegenV_{v_1} | \rho U^\dag |
        w_3, \DegenW_{w_3} \rangle )
      \nonumber \\ & \qquad  \;  /
      \langle w_3, \DegenW_{w_3} | \rho | w_3, \DegenW_{w_3} \rangle \, .
   \end{align}
   A complex generalization of the weight's numerator
   is the OTOC quasiprobability:
   \begin{align}
      \label{eq:OTOC_retro3}
      & \OurKD{\rho, \leftrightarrow }^\3 (
      v_1,  \DegenV_{v_1}  ;  w_2,  \DegenW_{w_2}  ;
      v_2,  \DegenV_{v_2}  ;  w_3 , \DegenW_{w_3}  )
      \nonumber \\ &
      =  \OurKD{\rho}
      ( v_1,  \DegenV_{v_1}  ;  w_2,  \DegenW_{w_2}  ;
      v_2,  \DegenV_{v_2}  ;  w_3 , \DegenW_{w_3} )  \, .
   \end{align}
\end{corollary}

The OTOC quasiprobability, we have shown,
assists with Bayesian-type inference,
similarly to the KD distribution.
The inferred-about operator is $V \W(t) V$,
rather than the $\W(t) V \W(t) V$ in the OTOC.
The missing $\W(t)$ plays the role of $F$.
This structure parallels the weak-measurement scheme in
the main text of~\cite{YungerHalpern_17_Jarzynski}:
$V$, $\W(t)$, and $V$ are measured weakly.
$\W(t)$ is, like $F$, then measured strongly.

\subsection{$\OurKD{\rho} ( . )$ values as coefficients
in an operator decomposition}
\label{section:TA_Coeffs}

Let $\Basis$ denote any orthonormal operator basis for $\mathcal{H}$.
Every state $\rho  \in  \mathcal{D} ( \mathcal{H} )$ can be decomposed
in terms of $\Basis$, as in Sec.~\ref{section:KD_Coeffs}.
The coefficients form a KD distribution.
Does $\OurKD{\rho}$ consist of
the coefficients in a state decomposition?

Summing $\OurKD{\rho} ( . )$ values yields
a coefficient in a decomposition of an operator $\rho'$.\footnote{
This $\rho'$ should not be confused with
the $\rho'$ in Theorem~\ref{theorem:RetroK}.}
$\rho'$ results from asymmetrically ``decohering'' $\rho$.
This decoherence relates to time-reversal asymmetry.
We expect $\rho'$ to tend to converge to $\rho$
after the scrambling time $t_*$.
By measuring $\OurKD{\rho}$ after $t_*$,
one may infer how accurately one prepared
the target initial state.

\begin{theorem}   \label{theorem:TA_Decomp}
Let
\begin{align}
   \label{eq:RhoPrime}
   \rho'  & :=  \rho  -  \sum_{ \substack{
   ( v_2,  \DegenV_{v_2} ) ,  ( w_3, \DegenW_{w_3} )
   \: : \:   \\
   \langle w_3, \DegenW_{w_3} | U | v_2,  \DegenV_{v_2} \rangle
   \neq 0 } }
   \ketbra{ v_2,  \DegenV_{v_2} }{ w_3, \DegenW_{w_3} }  U
   \nonumber \\ & \qquad \qquad \qquad \qquad \quad \times
   \langle v_2,  \DegenV_{v_2} |  \rho  U^\dag  |
                w_3, \DegenW_{w_3}  \rangle
\end{align}
denote the result of removing, from $\rho$,
the terms that connect the ``input state''
$U^\dag  \ket{ w_3, \DegenW_{w_3} }$
to the ``output state'' $\ket{ v_2,  \DegenV_{v_2} }$.
We define the set
\begin{align}
   \Basis  :=  \Set{  \frac{
   \ketbra{ v_2,  \DegenV_{v_2} }{ w_3, \DegenW_{w_3} }  U }{
   \langle  w_3, \DegenW_{w_3}  |  U  |  v_2,  \DegenV_{v_2}  \rangle }
   }_{ \langle w_3, \DegenW_{w_3} | U | v_2,  \DegenV_{v_2} \rangle
         \neq 0 }
\end{align}
of trace-one operators.
$\rho'$ decomposes in terms of $\Basis$ as
\begin{align}
   \sum_{ \substack{
   ( v_2,  \DegenV_{v_2} ) ,  ( w_3, \DegenW_{w_3} )
   \: : \:   \\
   \langle w_3, \DegenW_{w_3} | U | v_2,  \DegenV_{v_2} \rangle
   \neq 0 } }
   C^{ ( w_3, \DegenW_{w_3} ) }_{ ( v_2,  \DegenV_{v_2} ) }  \;
   \frac{
   \ketbra{ v_2,  \DegenV_{v_2} }{ w_3, \DegenW_{w_3} }  U }{
   \langle  w_3, \DegenW_{w_3}  |  U  |  v_2,  \DegenV_{v_2}  \rangle }  \, .
\end{align}
The coefficients follow from summing values
of the OTOC quasiprobability:
\begin{align}
   \label{eq:TA_decomp_coeff}
   & C^{ ( w_3, \DegenW_{w_3} ) }_{ ( v_2,  \DegenV_{v_2} ) }
   :=
   \sum_{ \substack{ ( w_2,  \DegenW_{w_2} ), \\  ( v_1,  \DegenV_{v_1} ) } }
   \OurKD{\rho}
   ( v_1,  \DegenV_{v_1}  ;  w_2,  \DegenW_{w_2} ;
     v_2,  \DegenV_{v_2}  ;  w_3,  \DegenW_{w_3} )  \, .
\end{align}
\end{theorem}

\begin{proof}
We deform the argument in Sec.~\ref{section:KD_Coeffs}.
Let the $\Set{ \ket{a} }$ in Sec.~\ref{section:KD_Coeffs} be
$\Set{ \ket{ v_2,  \DegenV_{v_2} } }$.
Let the $\Set{ \ket{f} }$ be
$\Set{ U^\dag \ket{ w_3,  \DegenW_{w_3} } }$.
We sandwich $\rho$ between resolutions of unity:
$\rho = \left( \sum_a  \ketbra{a}{a}  \right)  \rho
\left(  \sum_f  \ketbra{f}{f}  \right)$.
Rearranging yields
\begin{align}
   \label{eq:TA_Decomp_1}
   \rho  & =  \sum_{ ( v_2,  \DegenV_{v_2} ) ,  ( w_3,  \DegenW_{w_3} ) }
   \ketbra{ v_2,  \DegenV_{v_2} }{ w_3,  \DegenW_{w_3} } U
   \nonumber \\ & \qquad \qquad \qquad \qquad \times
   \langle  v_2,  \DegenV_{v_2}  |  \rho U^\dag  |
               w_3,  \DegenW_{w_3}  \rangle \, .
\end{align}

We wish to normalize the outer product,
by dividing by its trace.
We assumed, in Sec.~\ref{section:KD_Coeffs}, that
no interbasis inner product vanishes.
But inner products could vanish here.
Recall Example~\ref{ex:Real1}: When $t = 0$,
$\W(t)$ and $V$ share an eigenbasis.
That eigenbasis can have orthogonal states
$\ket{\psi}$ and $\ket{\phi}$.
Hence $\langle  w_3,  \DegenW_{w_3}  |  U  |
                          v_2,  \DegenV_{v_2}  \rangle$
can equal  $\langle \psi | \phi \rangle  =  0$.
No such term in Eq.~\eqref{eq:TA_Decomp_1}
can be normalized.

We eliminate these terms from the sum with
the condition $\langle  w_3,  \DegenW_{w_3}  |  U  |
                          v_2,  \DegenV_{v_2}  \rangle  \neq 0$.
The left-hand side of Eq.~\eqref{eq:TA_Decomp_1}
is replaced with the $\rho'$ in Eq.~\eqref{eq:RhoPrime}.
We divide and multiply by
the trace of each $\Basis$ element:
\begin{align}
   \rho' & =
   \sum_{ \substack{
   ( v_2,  \DegenV_{v_2} )  ,   ( w_3, \DegenW_{w_3} )  \: : \:   \\
   \langle w_3, \DegenW_{w_3} | U | v_2,  \DegenV_{v_2} \rangle
   \neq 0 } }
   \frac{
   \ketbra{ v_2,  \DegenV_{v_2} }{ w_3, \DegenW_{w_3} }  U }{
   \langle  w_3, \DegenW_{w_3}  |  U  |  v_2,  \DegenV_{v_2}  \rangle }
   \nonumber \\ &  \quad \times
   \langle  w_3, \DegenW_{w_3}  |  U  |  v_2,  \DegenV_{v_2}  \rangle
   \langle  v_2,  \DegenV_{v_2}  |  \rho  U^\dag  |
                w_3, \DegenW_{w_3}  \rangle \, .
\end{align}
The coefficients are KD-quasiprobability values.

Consider inserting, just leftward of the $\rho$,
the resolution of unity
\begin{align}
   \id  & =
   \left(  U^\dag   \sum_{ w_2, \DegenW_{w_2} }
            \ketbra{ w_2, \DegenW_{w_2} }{ w_2, \DegenW_{w_2} }  U  \right)
   \nonumber \\ & \qquad \times
   \left(  \sum_{ v_1,  \DegenV_{v_1} }
            \ketbra{ v_1,  \DegenV_{v_1} }{ v_1,  \DegenV_{v_1} }  \right) \, .
\end{align}
In the resulting $\rho'$ decomposition,
the $\sum_{ w_2, \DegenW_{w_2} }
\sum_{ v_1,  \DegenV_{v_1} }$
is pulled leftward, to just after the
$\frac{
   \ketbra{ v_2,  \DegenV_{v_2} }{ w_3, \DegenW_{w_3} }  U }{
   \langle  w_3, \DegenW_{w_3}  |  U  |  v_2,  \DegenV_{v_2}  \rangle }$.
This double sum becomes a sum of $\OurKD{\rho}$'s.
The $\rho'$ weights have the form in Eq.~\eqref{eq:TA_decomp_coeff}.
\end{proof}

Theorem~\ref{theorem:TA_Decomp} would hold
if $\rho$ were replaced with
any bounded operator $\Oper \in \mathcal{B} ( \Hil )$.
Four more points merit discussion.
We expect that, after the scrambling time $t_*$,
there tend to exist parameterizations
$\Set{ \DegenW_{w_\ell} }$ and $\Set{ \DegenV_{v_m} }$
such that $\Basis$ forms a basis.
Such a tendency could facilitate error estimates:
Suppose that $\OurKD{\rho}$ is measured after $t_*$.
One can infer the form of the state $\rho$ prepared
at the trial's start.
The target initial state may be difficult to prepare, e.g., thermal.
The preparation procedure's accuracy can be assessed at a trivial cost.
Third, the physical interpretation of $\rho'$ merits investigation.
The asymmetric decoherence relates to time-reversal asymmetry.
Fourth, the sum in Eq.~\eqref{eq:TA_decomp_coeff} relates to
a sum over trajectories,
a marginalization over intermediate-measurement outcomes.

\emph{Relationship between scrambling and
completeness of $\Basis$:}
The $\Set{ \frac{ \ketbra{ a }{ f } }{ \langle f | a \rangle } }$
in Sec.~\ref{section:KD_Coeffs} forms a basis for $\mathcal{D} ( \Hil )$.
But suppose that $\rho'  \neq  \rho$.
$\Basis$ fails to form a basis.

What does this failure imply about $\W(t)$ and $V$?
The failure is equivalent to the existence of a vanishing
$\xi  :=  | \langle w_3, \DegenW_{w_3} | U |
               v_2,  \DegenV_{v_2} \rangle |$.
Some $\xi$ vanishes if
some degenerate eigensubspace $\mathcal{H}_0$ of $\W(t)$
is a degenerate eigensubspace of $V$:
Every eigenspace of every Hermitian operator
has an orthogonal basis.
$\mathcal{H}_0$ therefore has an orthogonal basis.
One basis element can be labeled $U^\dag \ket{ w_3, \DegenW_{w_3} }$;
and the other, $\ket{ v_2,  \DegenV_{v_2} }$.

The sharing of an eigensubspace is equivalent to
the commutation of some component of $\W(t)$ with
some component of $V$.
The operators more likely commute
before the scrambling time $t_*$
than after.
Scrambling is therefore expected to magnify the similarity between
the OTOC quasiprobability $\OurKD{\rho}$
and the conventional KD distribution.

Let us illustrate with an extreme case.
Suppose that all the $\xi$'s lie as far from zero as possible:
\begin{align}
   \label{eq:MUB}
   \xi  =  \frac{1}{ \sqrt{ \Dim } }  \quad  \forall \xi \, .
\end{align}
Equation~\eqref{eq:MUB} implies that
$\mathcal{W}(t)$ and $V$ eigenbases are
\emph{mutually unbiased biases} (MUBs)~\cite{Durt_10_On}.
MUBs are eigenbases of operators
that maximize the lower bound
in an uncertainty relation~\cite{Coles_15_Entropic}.
If you prepare any eigenstate of one operator
(e.g., $U^\dag  \ket{ w_\ell,  \DegenW_{w_\ell} }$)
and measure the other operator (e.g., $V$),
all the possible outcomes have equal likelihoods.
You have no information with which to predict the outcome;
your ignorance is maximal.
$\W(t)$ and $V$ are maximally incompatible,
in the quantum-information (QI) sense of entropic uncertainty relations.
Consistency between this QI sense
of ``mutually incompatible'' and the OTOC sense
might be expected:
$\W(t)$ and $V$ eigenbases might be expected
to form MUBs after the scrambling time $t_*$.
We elaborate on this possibility in Sec.~\ref{section:TheoryOpps}.

KD quasiprobabilities are typically evaluated
on MUBs, such as position and momentum eigenbases~\cite{Lundeen_11_Direct,Lundeen_12_Procedure,Thekkadath_16_Direct}.
One therefore might expect $\OurKD{\rho}$
to relate more closely the KD quasiprobability
after $t_*$ than before.
The OTOC motivates a generalization of KD studies
beyond MUBs.

\emph{Application: Evaluating a state preparation's accuracy:}
Experimentalists wish to measure the OTOC $F(t)$
at each of many times $t$.
One may therefore wish to measure $\OurKD{\rho}$ after $t_*$.
Upon doing so, one may be able to infer not only $F(t)$,
but also the accuracy with which one prepared the target initial state.

Suppose that, after $t_*$, some $\Basis$
that forms a basis for $\Hil$.
Consider summing late-time $\OurKD{\rho}( . )$ values
over $( w_2,  \DegenW_{w_2} )$
and $( v_1, \DegenV_{v_1} )$.
The sum equals a KD quasiprobability for $\rho$.
The quasiprobability encodes all the information in $\rho$~\cite{Lundeen_11_Direct,Lundeen_12_Procedure}.
One can reconstruct the state that one prepared~\cite{Piacentini_16_Measuring,Suzuki_16_Observation,Thekkadath_16_Direct}.

The prepared state $\rho$ might differ from
the desired, or target, state $\rho_\target$.
Thermal states $e^{ - H / T } / Z$
are difficult to prepare, for example.
How accurately was $\rho_\target$ prepared?
One may answer by comparing $\rho_\target$
with the KD quasiprobability $\OurKD{\rho}$ for $\rho$.

Reconstructing the KD quasiprobability requires a trivial sum
over already-performed measurements
[Eq.~\eqref{eq:TA_decomp_coeff}].
One could reconstruct $\rho$ independently
via conventional quantum-state tomography~\cite{Paris_04_Q_State_Estimation}.
The $\rho$ reconstruction inferred from $\OurKD{\rho}$
may have lower precision,
due to the multiplicity of weak measurements
and to the sum.
But independent tomography would likely require extra measurements,
exponentially many in the system size.
Inferring $\OurKD{\rho}$ requires
exponentially many measurements, granted.\footnote{
One could measure, instead of $\OurKD{\rho}$,
the coarse-grained quasiprobability
$\SumKD{\rho}  =:  \sum_{\rm degeneracies} \OurKD{\rho}$
(Sec.~\ref{section:ProjTrick}).
From $\SumKD{\rho}$, one could infer the OTOC.
Measuring $\SumKD{\rho}$ would require
exponentially fewer measurements.
But from $\SumKD{\rho}$, one could not infer the KD distribution.
One could infer a coarse-grained KD distribution, akin to
a block-diagonal matrix representation for $\rho$.}
But, from these measurements,
one can infer $\OurKD{\rho}$, the OTOC, and $\rho$.
Upon reconstructing the KD distribution for $\rho$,
one can recover a matrix representation for $\rho$
via an integral transform~\cite{Lundeen_12_Procedure}.

\emph{The asymmetrically decohered $\rho'$:}
What does the decomposed operator $\rho'$ signify?
$\rho'$ has the following properties:
The term subtracted off in Eq.~\eqref{eq:RhoPrime}
has trace zero.
Hence $\rho'$ has trace one, like a density operator.
But the subtracted-off term is not Hermitian.
Hence $\rho'$ is not Hermitian,
unlike a density operator.
Nor is $\rho'$ anti-Hermitian, necessarily unitarity,
or necessarily anti-unitary.

$\rho'$ plays none of the familiar roles---of
state, observable, or time-evolution operator---in quantum theory.
The physical significance of $\rho'$ is not clear.
Similar quantities appear in weak-measurement theory:
First, non-Hermitian products $\B \A$ of observables
have been measured weakly
(see Sec.~\ref{section:TA_retro} and~\cite{Piacentini_16_Measuring,Suzuki_16_Observation,Thekkadath_16_Direct}).
Second, nonsymmetrized correlation functions
characterize quantum detectors
of photon absorptions and emissions~\cite{Bednorz_13_Nonsymmetrized}.
Weak measurements imbue these examples with physical significance.
We might therefore expect $\rho'$ to have physical significance.
Additionally, since $\rho'$ is non-Hermitian,
non-Hermitian quantum mechanics
might offer insights~\cite{Moiseyev_11_Non}.

The subtraction in Eq.~\eqref{eq:RhoPrime}
constitutes a removal of coherences.
But the subtraction is not equivalent to a decohering channel~\cite{NielsenC10},
which outputs a density operator.
Hence our description of the decoherence as asymmetric.

The asymmetry relates to the breaking time-reversal invariance.
Let $U^\dag  \ket{  w_3,  \DegenW_{w_3}  }
=:  \ket{  \tilde{w}_3  }$
be fixed throughout the following argument
(be represented, relative to any given basis, by a fixed list of numbers).
Suppose that $\rho = e^{ - H / T } / Z$.
The removal of $\langle v_2,  \DegenV_{v_2}  |  \rho  |
\tilde{w}_3  \rangle$     terms from $\rho$
is equivalent to the removal of
$\langle v_2,  \DegenV_{v_2}  |  H  |  \tilde{w}_3  \rangle$
terms from $H$:
$\rho  \mapsto  \rho'  \;  \Leftrightarrow  \;  H  \mapsto  H'$.
Imagine, temporarily, that $H'$ could represent a Hamiltonian
without being Hermitian.
$H'$ would generate a time evolution under which
$\ket{ \tilde{w}_3 }$ could not evolve into
$\ket{ v_2,  \DegenV_{v_2} }$.
But $\ket{ v_2,  \DegenV_{v_2} }$ could evolve into
$\ket{ \tilde{w}_3 }$.
The forward process would be allowed;
the reverse would be forbidden.
Hence $\rho \mapsto \rho'$ relates to
a breaking of time-reversal symmetry.

\emph{Interpretation of the sum in Eq.~\eqref{eq:TA_decomp_coeff}:}
Summing $\OurKD{\rho} ( . )$ values, in Eq.~\eqref{eq:TA_decomp_coeff},
yields a decomposition coefficient $C$ of $\rho'$.
Imagine introducing that sum into Eq.~\eqref{eq:OTOC_retro3}.
The OTOC quasiprobability $\OurKD{\rho} ( . )$
would become a KD quasiprobability.
Consider applying this summed Eq.~\eqref{eq:OTOC_retro3}
in Eq.~\eqref{eq:OTOC_retro1}.
We would change from retrodicting about $V \W(t) V$
to retrodicting about the leftmost $V$.

\subsection{Relationship between out-of-time ordering
and quasiprobabilities}
\label{section:OTOC_TOC}

The OTOC has been shown to equal
a moment of the complex distribution $P(W, W')$~\cite{YungerHalpern_17_Jarzynski}.
This equality echoes Jarzynski's~\cite{Jarzynski_97_Nonequilibrium}.
Jarzynski's equality governs
out-of-equilibrium statistical mechanics.
Examples include a quantum oscillator
whose potential is dragged quickly~\cite{An_15_Experimental}.
With such nonequilibrium systems, one can associate
a difficult-to-measure, but useful,
free-energy difference $\Delta F$.
Jarzynski cast $\Delta F$ in terms of
the characteristic function $\expval{ e^{ - \beta W } }$
of a probability distribution $P(W)$.\footnote{
Let $P(W)$ denote a probability distribution
over a random variable $W$.
The characteristic function $\Charac (s)$
equals the Fourier transform:
$\Charac (s)  :=  \int dW  \;  e^{ i s W }$.
Defining $s$ as an imaginary-time variable,
$is \equiv - \beta$, yields $\expval{ e^{ - \beta W } }$.
Jarzynski's equality reads,
$\expval{ e^{ - \beta W } }  =  e^{ - \beta \Delta F }$.}
Similarly, the difficult-to-measure, but useful, OTOC $F(t)$
has been cast in terms of
the characteristic function $\expval{ e^{ - (\beta W + \beta' W') } }$
of the summed quasiprobability $P(W, W')$~\cite{YungerHalpern_17_Jarzynski}.

Jarzynski's classical probability must be replaced with a quasiprobability
because $[\W(t),  V]  =  0$.
This replacement appeals to intuition:
Noncommutation and quasiprobabilities reflect nonclassicality
as commuting operators and probabilities do not.
The OTOC registers quantum-information scrambling
unregistered by \emph{time-ordered correlators} (TOCs).
One might expect TOCs to equal
moments of coarse-grained quasiprobabilities
closer to probabilities than $\OurKD{\rho}$ is.

We prove this expectation.
First, we review the TOC $\TOC (t)$.
Then, we introduce the TOC analog $\Amp_\rho^\toc$ of
the probability amplitude $\Amp_\rho$ [Eq.~\eqref{eq:Amp}].
$\Amp_\rho$ encodes no time reversals, as expected.
Multiplying a forward amplitude $\Amp_\rho^\toc$
by a backward amplitude $\left( \Amp_\rho^\toc \right)^*$
yields the TOC quasiprobability $\TOCKD{\rho}$.
Inferring $\TOCKD{\rho}$ requires
only two weak measurements per trial.
$\TOCKD{\rho}$ reduces to a probability
if $\rho = \rho_{V}$ [Eq.~\eqref{eq:WRho}].
In contrast, under no known condition on $\rho$
do all $\OurKD{\rho}( . )$ values
reduce to probability values.
Summing $\TOCKD{\rho}$ under constraints
yields a complex distribution $P_\toc (W, W')$.
The TOC $\TOC (t)$ equals a moment of $P_\toc (W, W')$.

%
%
%
\subsubsection{Time-ordered correlator $\TOC(t)$}
\label{section:TOC_def}

The OTOC equals a term in
the expectation value $\expval{ . }$ of
the squared magnitude $| . |^2$ of a commutator $[ . \, , \: . ]$~\cite{Kitaev_15_Simple,Maldacena_15_Bound},
\begin{align}
   \label{eq:CommSquared}
   C(t)  & :=  \expval{ [ \W(t) ,  V ]^\dag  [ \W(t),  V ] } \\
   \nonumber \\ &
   =  - \expval{ \W^\dag (t)  V^\dag  V   \W (t)  }
   -  \expval{ V^\dag   \W^\dag (t)  \W(t)  V }
   \nonumber \\ &  \qquad
   +  2 \Re \LParen F(t) \RParen \, .
\end{align}
The second term is a time-ordered correlator (TOC),
\begin{align}
   \label{eq:TOC_def}
   \TOC (t)  :=   \expval{ V^\dag   \W^\dag (t)  \W(t)  V } \, .
\end{align}
The first term, $\expval{ \W^\dag (t)  V^\dag  V   \W (t)  }$,
exhibits similar physics.
Each term evaluates to one if $\W$ and $V$ are unitary.
If $\W$ and $V$ are nonunitary Hermitian operators,
the TOC reaches its equilibrium value by the dissipation time $t_d < t_*$
(Sec.~\ref{section:OTOC_review}).
The TOC fails to reflect scrambling,
which generates the OTOC's Lyapunov-type behavior
at $t \in (t_d,  \,  t_*)$.

\subsubsection{TOC probability amplitude $\Amp_\rho^\toc$}
\label{section:TOC_Amp}

We define
\begin{align}
   \label{eq:TOC_Amp}
   & \Amp_\rho^\toc  ( j ;  v_1,  \DegenV_{v_1} ;  w_1, \DegenW_{w_1} )
   \nonumber \\ &
   :=  \langle  w_1,  \DegenW_{w_1}  | U  |  v_1   \DegenV_{v_1}  \rangle
   \langle  v_1   \DegenV_{v_1}  |  j  \rangle \,
   \sqrt{ p_j }
\end{align}
as the \emph{TOC probability amplitude}.
$\Amp_\rho^\toc$ governs a quantum process $\ProtocolA^\toc$.
Figure~\ref{fig:TOC_amp1}, analogous to Fig.~\ref{fig:Protocoll_Trial1},
depicts $\ProtocolA^\toc$,
analogous to the $\ProtocolA$ in Sec.~\ref{section:Review_A}:
\begin{enumerate}[(1)]

   \item Prepare $\rho$.

   \item Measure the $\rho$ eigenbasis, $\Set{ \ketbra{j}{j} }$.

   \item Measure $\NondegV$.

   \item Evolve the system forward in time under $U$.

   \item Measure $\NondegW$.

\end{enumerate}
Equation~\eqref{eq:TOC_Amp} represents
the probability amplitude associated with
the measurements' yielding the outcomes
$j,  ( v_1,  \DegenV_{v_1} )$, and $( w_1,  \DegenW_{w_1} )$,
in that order.
All the measurements are strong.
$\ProtocolA^\toc$ is not a protocol for measuring $\Amp_\rho^\toc$.
Rather, $\ProtocolA^\toc$ facilitates
the physical interpretation of $\Amp_\rho^\toc$.

%
%
\begin{figure}[h]
\centering
\begin{subfigure}{0.4\textwidth}
\centering
\includegraphics[width=.9\textwidth]{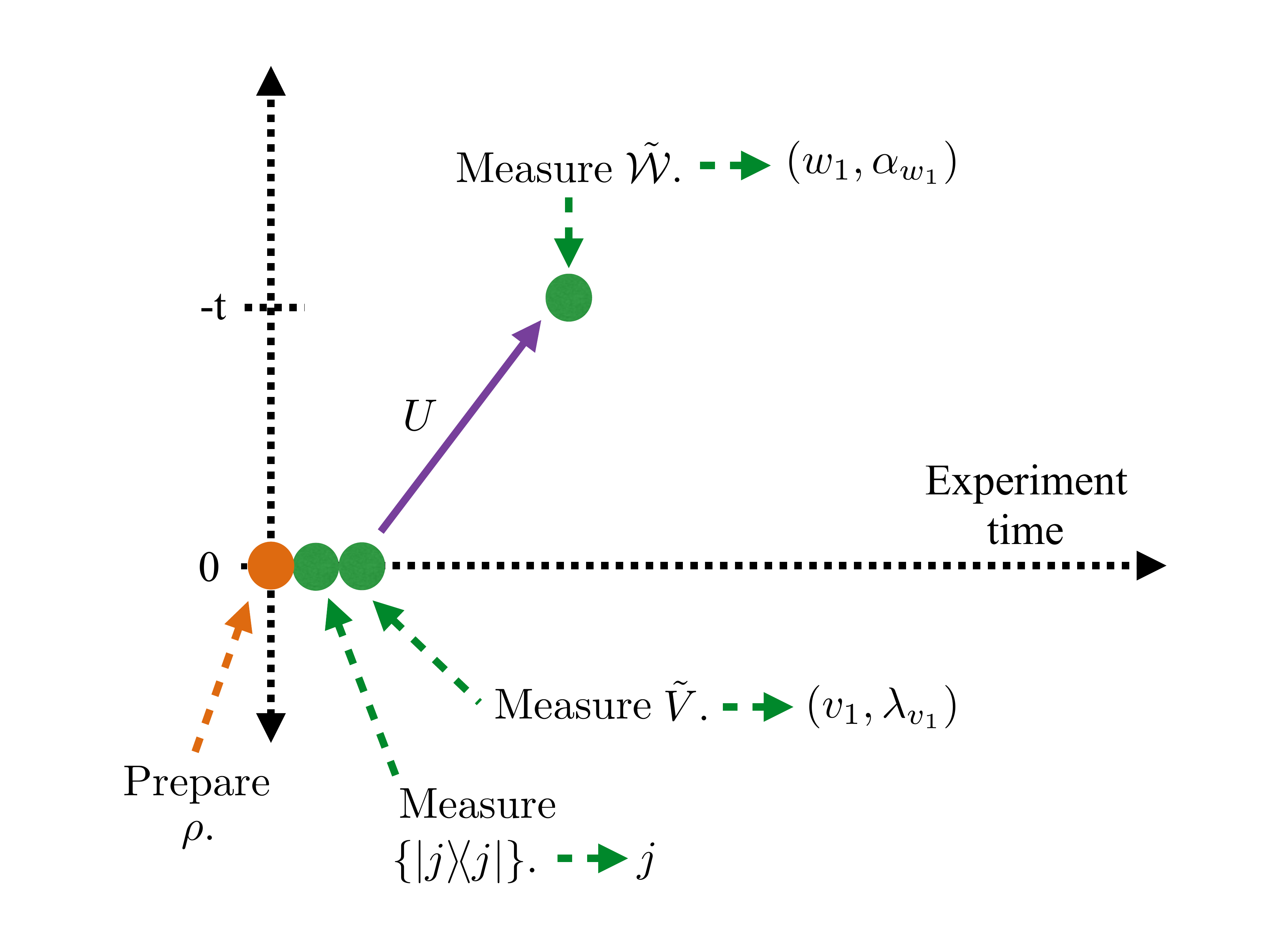}
\caption{}
\label{fig:TOC_amp1}
\end{subfigure}
\begin{subfigure}{.4\textwidth}
\centering
\includegraphics[width=.9\textwidth]{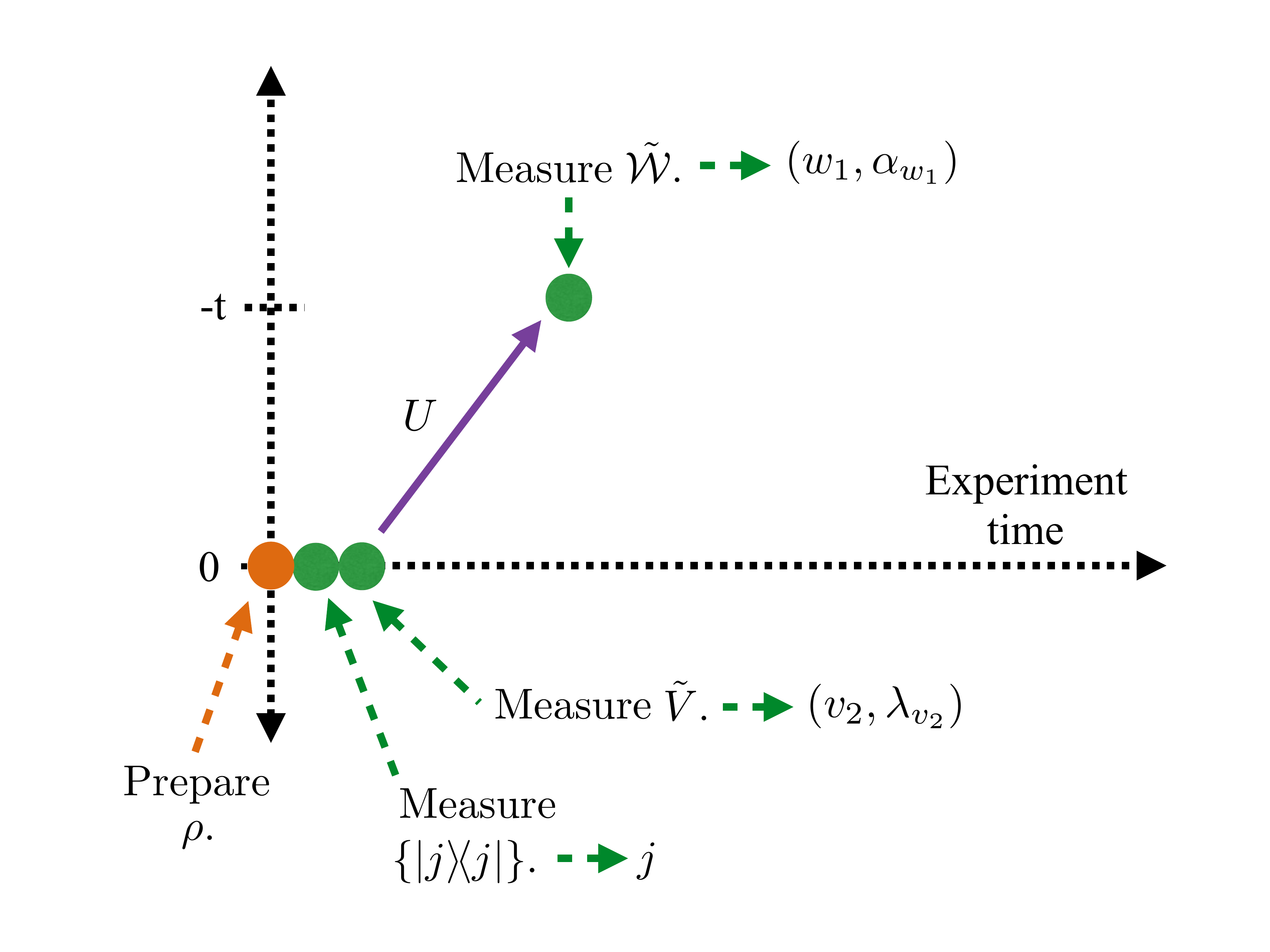}
\caption{}
\label{fig:TOC_amp2}
\end{subfigure}
\caption{\caphead{Quantum processes described by
the probability amplitudes $\Amp_\rho^\toc$
in the time-ordered correlator (TOC) $\TOC (t)$:}
$\TOC (t)$, like $F(t)$, equals a moment of
a summed quasiprobability (Theorem~\ref{theorem:TOC_Jarz}).
The quasiprobability, $\OurKD{\rho}^\toc$, equals
a sum of multiplied probability amplitudes $\Amp_\rho^\toc$
[Eq.~\eqref{eq:TOCKD_def}].
Each product contains two factors:
$\Amp_\rho^\toc ( j  ;  v_1,  \DegenV_{v_1}  ;  w_1,  \DegenW_{w_1} )$
denotes the probability amplitude
associated with the ``forward'' process in Fig.~\ref{fig:TOC_amp1}.
The system, $\Sys$, is prepared in a state $\rho$.
The $\rho$ eigenbasis $\Set{ \ketbra{j}{j} }$ is measured,
yielding outcome $j$.
$\NondegV$ is measured,
yielding outcome $( v_1,  \DegenV_{v_1} )$.
$\Sys$ is evolved forward in time under the unitary $U$.
$\NondegW$ is measured,
yielding outcome $( w_1,  \DegenW_{w_1} )$.
Along the abscissa runs the time
measured by a laboratory clock.
Along the ordinate runs the $t$ in $U := e^{ - i H t }$.
The second factor in each $\OurKD{\rho}^\toc$ product is
$\Amp_\rho^\toc ( j  ;  v_2,  \DegenV_{v_2}  ;  w_1,  \DegenW_{w_1} )^*$.
This factor relates to the process in Fig.~\ref{fig:TOC_amp2}.
The operations are those in Fig.~\ref{fig:TOC_amp1}.
The processes' initial measurements yield the same outcome.
So do the final measurements.
The middle outcomes might differ.
Complex-conjugating $\Amp_\rho^\toc$ yields
the probability amplitude associated with the \emph{reverse} process.
Figures~\ref{fig:TOC_amp1} and~\ref{fig:TOC_amp2}
depict no time reversals.
Each analogous OTOC figure
(Fig.~\ref{fig:Protocoll_Trial1} and Fig.~\ref{fig:Protocoll_Trial2}) depicts two.}
\label{fig:TOC_amps}
\end{figure}

$\ProtocolA^\toc$ results from eliminating, from $\ProtocolA$,
the initial $U$, $\NondegW$ measurement, and $U^\dag$.
$\Amp_\rho$ encodes two time reversals.
$\Amp_\rho^\toc$ encodes none,
as one might expect.

%
%
%
\subsubsection{TOC quasiprobability $\TOCKD{\rho}$}
\label{section:TOC_KD}

Consider a $\ProtocolA^\toc$ implementation that yields the outcomes
$j$, $( v_2,  \DegenV_{v_2} )$, and $( w_1,  \DegenW_{w_1} )$.
Such an implementation appears in Fig.~\ref{fig:TOC_amp2}.
The first and last outcomes [$j$ and $( w_1,  \DegenW_{w_1} )$]
equal those in Fig.~\ref{fig:TOC_amp1}, as in the OTOC case.
The middle outcome can differ.
This process corresponds to the probability amplitude
\begin{align}
   \label{eq:Amp_TOC_rev}
   & \Amp_\rho^\toc ( j ; v_2,  \DegenV_{v_2}  ;  w_1,  \DegenW_{w_1} )
   \nonumber \\ & \quad
   =  \langle w_1 , \DegenW_{w_1}  |  U  |  v_2,  \DegenV_{v_2}
   \rangle \langle   v_2,  \DegenV_{v_2}  |  j  \rangle  \:
   \sqrt{p_j} \, .
\end{align}
Complex conjugation reverses the inner products,
yielding the reverse process's amplitude.

We multiply this reverse amplitude by the forward amplitude~\eqref{eq:TOC_Amp}.
Summing over $j$ yields the \emph{TOC quasiprobability}:
\begin{align}
   \label{eq:TOCKD_def}
   & \TOCKD{\rho}
   ( v_1 ,  \DegenV_{v_1}  ;  w_1,  \DegenW_{w_1} ;  v_2 ,  \DegenV_{v_2}  )
   \nonumber \\ &
   :=  \sum_j
   \Amp_\rho^\toc ( j ; v_2,  \DegenV_{v_2}  ;  w_1,  \DegenW_{w_1} )^*
    \Amp_\rho^\toc  ( j ; v_1,  \DegenV_{v_1}  ;  w_1,  \DegenW_{w_1}  ) \\
   \label{eq:TOCKD_form}
   & =  \langle  v_2 ,  \DegenV_{v_2}  |  U^\dag  |  w_1,  \DegenW_{w_1}  \rangle
   \langle  w_1,  \DegenW_{w_1}  |  U  |  v_1 ,  \DegenV_{v_1}
   \rangle
   \nonumber \\ & \qquad \times
   \langle v_1 ,  \DegenV_{v_1} |  \rho  |   v_2 ,  \DegenV_{v_2}  \rangle  \, .
\end{align}

Like $\OurKD{\rho}$, $\TOCKD{\rho}$ is
an extended Kirkwood-Dirac quasiprobability.
$\TOCKD{\rho}$ is 2-extended, whereas $\OurKD{\rho}$ is 3-extended.
$\TOCKD{\rho}$ can be inferred from
a weak-measurement protocol $\Protocol^\toc$:
\begin{enumerate}[(1)]
   \item Prepare $\rho$.

   \item Measure $\NondegV$ weakly.

   \item Evolve the system forward under $U$.

   \item Measure $\NondegW$ weakly.

   \item Evolve the system backward under $U^\dag$.

   \item Measure $\NondegV$ strongly.

\end{enumerate}
$\Protocol^\toc$ requires just two weak measurements.
The weak-measurement protocol $\Protocol$
for inferring $\OurKD{\rho}$ requires three.
$\Protocol^\toc$ requires one time reversal;
$\Protocol$ requires two.

In a simple case, every $\Amp_\rho^\toc ( . )$ value
reduces to a probability value.
Suppose that $\rho$ shares the $\NondegV$ eigenbasis,
as in Eq.~\eqref{eq:WRho}.
The $( v_2 ,  \DegenV_{v_2} )$ in Eq.~\eqref{eq:TOCKD_form}
comes to equal $( v_1,  \DegenV_{v_1} )$;
Figures~\ref{fig:TOC_amp1} and~\ref{fig:TOC_amp2}
become identical.
Equation~\eqref{eq:TOCKD_form} reduces to
\begin{align}
   & \Amp_{ \rho_{V} }^\toc
   ( v_1,  \DegenV_{v_1}  ;  w_1,  \DegenW_{w_1} ;  v_2,  \DegenV_{v_2} ) \\
   & =  |  \langle  w_1,  \DegenW_{w_1}  |  U  |
             v_1,  \DegenV_{v_1}  \rangle  |^2  \,
   p_{ v_1,  \DegenV_{v_1} }  \,
   \delta_{ v_1  v_2 }  \,  \delta_{ \DegenV_{v_1}  \DegenV_{v_2} }  \\
   &  =  p (  w_1,  \DegenW_{w_1}  |  v_1,  \DegenV_{v_1} )  \,
   p_{ v_1,  \DegenV_{v_1} }  \,
  \delta_{ v_1  v_2 }  \,  \delta_{ \DegenV_{v_1}  \DegenV_{v_2} }  \\
   & =  p ( v_1,  \DegenV_{v_1} ;  w_1,  \DegenW_{w_1} )  \,
   \delta_{ v_1  v_2 }  \,  \delta_{ \DegenV_{v_1}  \DegenV_{v_2} }  \, .
\end{align}
The $p( a | b)$ denotes the conditional probability that,
if $b$ has occurred, $a$ will occur.
$p( a ; b )$ denotes the joint probability that $a$ and $b$ will occur.

All values $\OurKD{\rho_{ V } }^\toc ( . )$ of the TOC quasiprobability
have reduced to probability values.
Not all values of $\OurKD{ \rho_V }$ reduce:
The values associated with
$( v_2 ,  \DegenV_{v_2 } )  =  ( v_1 ,  \DegenV_{v_1} )$ or
$( w_3 ,  \DegenW_{w_3} )  =  ( w_2 ,  \DegenW_{w_2} )$
reduce to products of probabilities.
[See the analysis around Eq.~\eqref{eq:Reduce_to_p}.]
The OTOC quasiprobability encodes nonclassicality---violations
of the axioms of probability---more resilient than the TOC quasiprobability's.

\subsubsection{Complex TOC distribution $P_\toc (W_\toc, W'_\toc)$}
\label{section:P_TOC}

Let $W_\toc$ and $W'_\toc$ denote random variables
analogous to thermodynamic work.
We fix the constraints $W_\toc  =  w_1  v_2$ and
$W'_\toc  =  w_1  v_1$.
($w_1$ and $v_2$ need not be complex-conjugated because
they are real, as $\W$ and $V$ are Hermitian.)
Multiple outcome sextuples
$( v_2,  \DegenV_{v_2} ; w_1,  \DegenW_{w_1} ; v_1,  \DegenV_{v_1} )$
satisfy these constraints.
Each sextuple corresponds to a quasiprobability
$\TOCKD{\rho} ( . )$.
We sum the quasiprobabilities that satisfy the constraints:
\begin{align}
   \label{eq:P_TOC}
   & P_\toc ( W_\toc, W'_\toc )  :=
   \sum_{ ( v_1,  \DegenV_{v_1} ) ,  ( w_1,  \DegenW_{w_1} ) ,
                ( v_2,  \DegenV_{v_2} ) }
   \nonumber \\ &  \times
   \OurKD{\rho}^\toc ( v_1,  \DegenV_{v_1}  ;  w_1,  \DegenW_{w_1} ;
                                   v_2,  \DegenV_{v_2} )  \,
   \delta_{ W ( w_1^* v_2^* ) }  \,  \delta_{W' ( w_1 v_1 ) }  \, .
\end{align}

$P_\toc$ forms a complex distribution.
Let $f$ denote any function of $W_\toc$ and $W'_\toc$.
The $P_\toc$ average of $f$ is
\begin{align}
   \label{eq:TOC_avg}
   & \expval{ f ( W_\toc ,  W'_\toc ) }
   \\ \nonumber &
   :=  \sum_{ W_\toc ,  W'_\toc }  f ( W_\toc ,  W'_\toc )  P_\toc ( W_\toc, W'_\toc )  \, .
\end{align}

\subsubsection{TOC as a moment of the complex distribution}
\label{eq:TOC_Jarz}

The TOC obeys an equality
analogous to Eq.~(11) in~\cite{YungerHalpern_17_Jarzynski}.
%
%
\begin{theorem}[Jarzynski-like theorem for the TOC]
   \label{theorem:TOC_Jarz}
   The time-ordered correlator~\eqref{eq:TOC_def}
   equals a moment of the complex distribution~\eqref{eq:P_TOC}:
   \begin{align}
      \label{eq:TOC_Jarz}
      \TOC (t)  =  \frac{ \partial^2 }{ \partial \beta \, \partial \beta' }
      \expval{ e^{ - ( \beta W_\toc  +  \beta' W'_\toc ) } }
      \Bigg\rvert_{ \beta, \beta' = 0 }   ,
   \end{align}
   wherein $\beta, \beta' \in \mathbb{R}$.
\end{theorem}
%
%
\begin{proof}
The proof is analogous to the proof of Theorem~1 in~\cite{YungerHalpern_17_Jarzynski}.
\end{proof}
\noindent
Equation~\eqref{eq:TOC_Jarz} can be recast as
$\TOC (t)  =  \expval{ W_\toc  W'_\toc }  \, ,$
along the lines of Eq.~\eqref{eq:RecoverF2}.

\subsection{Higher-order OTOCs as moments of
longer (summed) quasiprobabilities}
\label{section:HigherOTOCs}

Differentiating a characteristic function again and again
yields higher- and higher-point correlation functions.
So does differentiating $P(W, W')$ again and again.
But each resulting correlator encodes
just $\Ops = 3$ time reversals.
Let $\Opsb  =  \frac{1}{2} ( \Ops + 1 )  =  2, 3, \ldots$,
for $\Ops = 3, 5, \ldots$
A \emph{$\Opsb$-fold OTOC} has been defined~\cite{Roberts_16_Chaos,Hael_17_Classification}:
\begin{align}
   \label{eq:k_OTOC}
   F^\ParenKB (t)  :=
   \langle \underbrace{ \W(t)  V  \ldots  \W(t) V }_{2 \Opsb } \rangle
   \equiv  \Tr \LParen  \rho
   \underbrace{ \W(t)  V  \ldots  \W(t) V }_{2 \Opsb }  \RParen \, .
\end{align}
Each such correlation function contains $\Opsb$
Heisenberg-picture operators $\W(t)$
interleaved with $\Opsb$ time-0 operators $V$.
$F^\ParenKB (t)$ encodes $2 \Opsb - 1  =  \Ops$ time reversals,
illustrated in Fig.~\ref{fig:k_OTOC}.
We focus on Hermitian $\W$ and $V$,
as in~\cite{Maldacena_15_Bound,HosurYoshida_16_Chaos},
for simplicity.

The conventional OTOC corresponds to $\Ops = 3$ and $\Opsb = 2$:
$F(t)  =  F^\2 (t)$.
If $\Ops < 3$, $F^\ParenKB (t)$ is not OTO.

%
%
\begin{figure}[h]
\centering
\includegraphics[width=.45\textwidth]{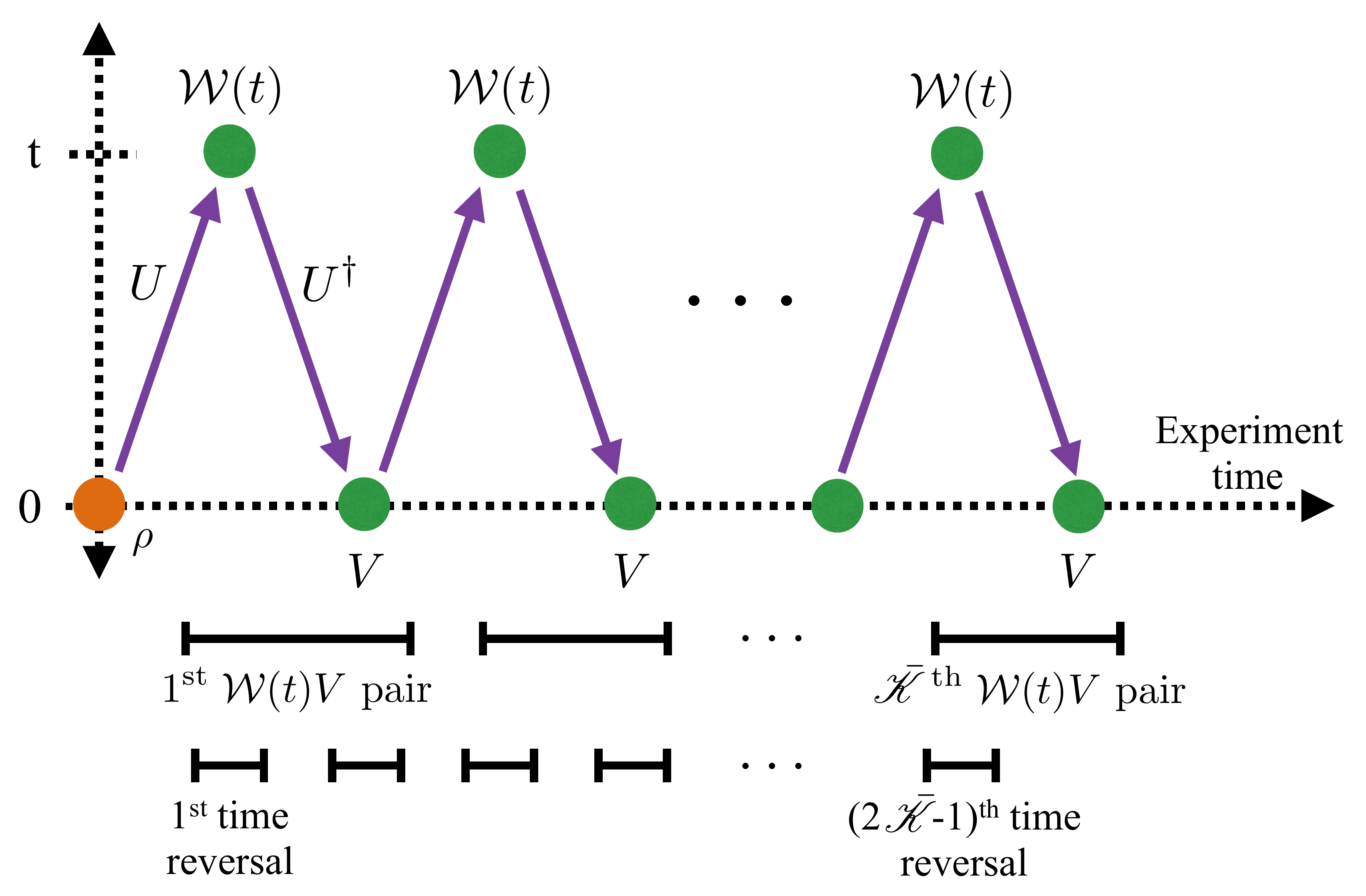}
\caption{\caphead{$\Opsb$-fold out-of-time-ordered correlator (OTOC):}
The conventional OTOC [Eq.~\eqref{eq:OTOC_Def}],
encodes just three time reversals.
The \emph{$\Opsb$-fold OTOC} $F^\ParenKB (t)$ encodes
$2 \Opsb - 1  =  \Ops   =  3, 5, \ldots$ time reversals.
The time measured by a laboratory clock
runs along the abscissa.
The ordinate represents the time parameter $t$,
which may be inverted in experiments.
The orange, leftmost dot represents
the state preparation $\rho$.
Each green dot represents a $\W(t)$ or a $V$.
Each purple line represents a unitary time evolution.
The diagram, scanned from left to right,
represents $F^\ParenKB (t)$, scanned from left to right.
}
\label{fig:k_OTOC}
\end{figure}

The greater the $\Ops$, the longer the distribution $P^\ParenK$
of which $F^\ParenKB (t)$ equals a moment.
We define $P^\ParenK$ in three steps:
We recall the $\Ops$-extended quasiprobability $\OurKD{\rho}^\ParenK$
[Eq.~\eqref{eq:Extend_KD}].
We introduce measurable random variables
$W_\ell$ and $W'_{\ell'}$.
These variables participate in constraints
on sums of $\OurKD{\rho}^\ParenK ( . )$ values.

Let us evaluate Eq.~\eqref{eq:Extend_KD}
on particular arguments:
\begin{align}
   \label{eq:k_OTO_quasi}
   & \OurKD{\rho}^\ParenK (
   v_1,  \DegenV_{v_1}  ;  w_2 ,  \DegenW_{w_2} ;
   \ldots ; v_{\Opsb} ,  \DegenV_{ v_{\Opsb} } ;
   w_{\Opsb + 1},  \DegenW_{ w_{\Opsb + 1} } )
   \nonumber \\ &
   =  \langle  w_{\Opsb + 1},  \DegenW_{w_{\Opsb + 1}}  |  U  |
                    v_{\Opsb}, \DegenV_{v_{\Opsb}}  \rangle
   \langle  v_{\Opsb}, \DegenV_{v_{\Opsb}}  |  U^\dag  |
               w_{\Opsb},  \DegenW_{w_{\Opsb}}  \rangle
   \nonumber \\ &  \times \ldots  \times
   \langle w_2,  \DegenW_{w_2}  |  U  |
               v_1,  \DegenV_{v_1}  \rangle
   \langle  v_1,  \DegenV_{v_1}  |  \rho  U^\dag  |
               w_{\Opsb + 1},  \DegenW_{w_{\Opsb+1}}  \rangle  \, .
\end{align}
One can infer $\OurKD{ \rho }^\ParenK$
from the interferometry scheme in~\cite{YungerHalpern_17_Jarzynski}
and from weak measurements.
Upon implementing one batch of the interferometry trials,
one can infer $\OurKD{\rho}^\ParenK$ for all $\Ops$-values:
One has measured all the inner products $\langle a | \U | b \rangle$.
Multiplying together arbitrarily many inner products yields
an arbitrarily high-$\Ops$ quasiprobability.
Having inferred some $\OurKD{\rho}^\ParenK$,
one need not perform new experiments
to infer $\OurKD{\rho}^{( \Ops + 2 )}$.
To infer $\OurKD{\rho}^\ParenK$ from weak measurements,
one first prepares $\rho$.
One performs $\Ops  =  2 \Opsb - 1$ weak measurements
interspersed with unitaries.
(One measures $\NondegV$ weakly,
evolves with $U$, measures $\NondegW$ weakly,
evolves with $U^\dag$, etc.)
Finally, one measures $\NondegW$ strongly.
The strong measurement corresponds to
the anomalous index $\Opsb + 1$ in
$( w_{\Opsb + 1},    \DegenW_{w_{\Opsb + 1}} )$.

We define $2 \Opsb$ random variables
\begin{align}
   \label{eq:W_ell}
   & W_\ell  \in \{ w_\ell \}   \qquad  \forall \ell =  2, 3, \ldots, \Opsb +1
   \qquad \text{and} \\
   & W'_{\ell'}  \in  \{ v_{\ell'} \}  \qquad \forall  \ell'  =  1, 2, \ldots, \Opsb \, .
\end{align}
Consider fixing the values of the $W_\ell$'s and the $W'_{\ell'}$'s.
Certain quasiprobability values $\OurKD{\rho}^\ParenK ( . )$
satisfy the constraints $W_\ell = w_\ell$ and $W'_{\ell'}  =  v_{\ell'}$
for all $\ell$ and $\ell'$.
Summing these quasiprobability values yields
\begin{align}
   \label{eq:P_k}
   & P^\ParenK ( W_2,  W_3,  \ldots,  W_{\Opsb +1},
   W'_1,  W'_2,  \ldots,  W'_{\Opsb} )
   \\ \nonumber &
   :=  \sum_{ W_2,  W_3,  \ldots,  W_{\Opsb +1} }
   \sum_{ W'_1,  W'_2,  \ldots,  W'_{\Opsb} } \\ &
   \OurKD{\rho}^\ParenK (
   v_1,  \DegenV_{v_1}  ;  w_2 ,  \DegenW_{w_2} ;
   \ldots ; v_{\Opsb} ,  \DegenV_{v_{\Opsb}} ;  w_{\Opsb+1} ,  \DegenW_{w_{\Opsb+1}  } )
   \nonumber \\ &  \nonumber  \times
   \left(  \delta_{ W_2 w_2 }  \times  \ldots
            \times  \delta_{ W_{\Opsb+1} w_{\Opsb+1} }  \right)
   \left(  \delta_{ W'_1  v_1 }  \times  \ldots
            \times  \delta_{ W'_{\Opsb}  v_{\Opsb} }  \right)    .
\end{align}


\begin{theorem}[The $\Opsb$-fold OTOC as a moment]
\label{theorem:k_OTOC}
The $\Opsb$-fold OTOC equals a $2\Opsb^\th$ moment of
the complex distribution~\eqref{eq:P_k}:
\begin{align}
   \label{eq:k_OTOC_thm}
   & F^\ParenKB (t)  =
   \frac{ \partial^{2\Opsb} }{
            \partial \beta_2  \ldots  \partial \beta_{\Opsb + 1 }  \,
            \partial \beta'_1  \ldots  \partial \beta'_{\Opsb} }
   \nonumber \\ &
   \expval{ \exp \left(  -  \left[
                 \sum_{ \ell = 2}^{\Opsb + 1 }  \beta_\ell  W_\ell
                 +   \sum_{ \ell' = 1 }^{\Opsb}  \beta'_{\ell'}  W'_{\ell'}
                 \right]  \right)  }
   \Bigg\lvert_{ \beta_\ell, \beta'_{\ell'}   =  0  \;  \forall \ell, \ell' }  \, ,
\end{align}
wherein $\beta_\ell, \beta'_\ell  \in  \mathbb{R}$.
\end{theorem}

\begin{proof}
The proof proceeds in analogy with
the proof of Theorem~1 in~\cite{YungerHalpern_17_Jarzynski}.
   %
   %
\end{proof}

The greater the $\Ops$, the ``longer''
the quasiprobability $\OurKD{\rho}^\ParenK$.
The more weak measurements are required to infer
$\OurKD{\rho}^\ParenK$.
Differentiating $\OurKD{\rho}^\ParenK$ more
does not raise the number of time reversals encoded in the correlator.

Equation~\eqref{eq:k_OTOC_thm} can be recast as
$F^\ParenKB (t)  =  \expval{
\left(  \prod_{ \ell = 2}^{\Opsb + 1 }  W_\ell  \right)
\left(  \prod_{ \ell' = 1 }^{\Opsb}  W'_{\ell'}   \right)
}  \, ,$
along the lines of Eq.~\eqref{eq:RecoverF2}.

\section{Outlook}
\label{section:Outlook}

We have characterized the quasiprobability $\OurKD{\rho}$
that ``lies behind'' the OTOC $F(t)$.
$\OurKD{\rho}$, we have argued, is an extension of
the Kirkwood-Dirac distribution used in quantum optics.
We have analyzed and simplified measurement protocols for $\OurKD{\rho}$,
calculated $\OurKD{\rho}$ numerically and on average over Brownian circuits,
and investigated mathematical properties.
This work redounds upon quantum chaos,
quasiprobability theory, and weak-measurement physics.
As the OTOC equals a combination of $\OurKD{\rho}( . )$ values,
$\OurKD{\rho}$ provides more-fundamental information about scrambling.
The OTOC motivates generalizations of,
and fundamental questions about, KD theory.
The OTOC also suggests a new application
of sequential weak measurements.

At this intersection of fields lie many opportunities.
We classify the opportunities by the tools
that inspired them: experiments, calculations, and abstract theory.

\subsection{Experimental opportunities}

We expect the weak-measurement scheme for $\OurKD{\rho}$ and $F(t)$
to be realizable in the immediate future.
Candidate platforms include
superconducting qubits, trapped ions, ultracold atoms,
and perhaps NMR.
Experimentalists have developed key tools required to implement the protocol~\cite{Bollen_10_Direct,Lundeen_11_Direct,Lundeen_12_Procedure,Bamber_14_Observing,Mirhosseini_14_Compressive,White_16_Preserving,Hacohen_16_Quantum,Browaeys_16_Experimental,Piacentini_16_Measuring,Suzuki_16_Observation,Thekkadath_16_Direct}.


Achievable control and dissipation must be compared with
the conditions needed to infer the OTOC.
Errors might be mitigated with
tools under investigation~\cite{Swingle_Resilience}.

\subsection{Opportunities motivated by calculations}

Numerical simulations and analytical calculations
point to three opportunities.

Physical models' OTOC quasiprobabilities may be evaluated.
The Sachdev-Ye-Kitaev model, for example, scrambles quickly~\cite{Sachdev_93_Gapless,Kitaev_15_Simple}.
The quasiprobability's functional form
may suggest new insights into chaos.
Our Brownian-circuit calculation (Sec.~\ref{section:Brownian}),
while a first step, involves averages over unitaries.
Summing quasiprobabilities can cause interference
to dampen nonclassical behaviors~\cite{Dressel_15_Weak}.
Additionally, while unitary averages model chaotic evolution,
explicit Hamiltonian evolution might provide different insights.
Explicit Hamiltonian evolution would also preclude the need
to calculate higher moments of the quasiprobability.

In some numerical plots, the real part $\Re ( \SumKD{\rho} )$ bifurcates.
These bifurcations resemble classical-chaos pitchforks~\cite{Strogatz_00_Non}.
Classical-chaos plots bifurcate when
a differential equation's equilibrium point
branches into three.
The OTOC quasiprobability $\OurKD{\rho}$ might be recast
in terms of equilibria.
Such a recasting would strengthen the parallel between
classical chaos and the OTOC.

Finally, the Brownian-circuit calculation has untied threads.
We calculated only the first moment of $\SumKD{\rho}$.
Higher moments may encode physics
less visible in $F(t)$.
Also, evaluating certain components of $\SumKD{\rho}$
requires new calculational tools.
These tools merit development,
then application to $\SumKD{\rho}$.
An example opportunity is discussed after Eq.~\eqref{eq:Brown_help}.

\subsection{Fundamental-theory opportunities}
\label{section:TheoryOpps}

Seven opportunities concern
the mathematical properties
and physical interpretations of $\OurKD{\rho}$.

The KD quasiprobability prompts the question,
``Is the OTOC definition of `maximal noncommutation' consistent with
the mutually-unbiased-bases definition?''
Recall Sec.~\ref{section:TA_Coeffs}:
We decomposed an operator $\rho'$
in terms of a set $\Basis  =  \Set{
\frac{ \ketbra{ a }{ f } }{  \langle f | a \rangle } }_{
\langle f | a \rangle  \neq 0 }$
of operators.
In the KD-quasiprobability literature, the bases
$\Basis_a =  \Set{ \ket{a} }$ and $\Basis_f  =  \Set{ \ket{f} }$
tend to be mutually unbiased (MU):
$| \langle f | a \rangle |  =  \frac{1}{ \sqrt{ \Dim } }  \;  \forall a, f$.
Let $\A$ and $\B$ denote operators that have MU eigenbases.
Substituting $\A$ and $\B$ into an uncertainty relation
maximizes the lower bound on an uncertainty~\cite{Coles_15_Entropic}.
In this quantum-information (QI) sense,
$\A$ and $\B$ noncommute maximally.

In Sec.~\ref{section:TA_Coeffs},
$\Basis_a  =  \Set{ \ket{v_2,  \DegenV_{v_2} } }$, and
$\Basis_f  =  \Set{  U^\dag \ket{ w_3, \DegenW_{w_3} } }$.
These $\Basis$'s are eigenbases of $V$ and $\W(t)$.
When do we expect these eigenbases to be MU,
as in the KD-quasiprobability literature?
After the scrambling time $t_*$---after $F(t)$ decays to zero---when
$\W(t)$ and $V$ noncommute maximally in the OTOC sense.

The OTOC provides one definition of ``maximal noncommutation.''
MUBs provide a QI definition.
To what extent do these definitions overlap?
Initial results show that, in some cases,
the distribution over possible values of
$| \langle v_2,  \DegenV_{v_2} | U | w_3, \DegenW_{w_3} \rangle |$
peaks at $\frac{1}{ \sqrt{ \Dim } }$.
But the distribution approaches this form before $t_*$.
Also, the distribution's width seems constant in $\Dim$.
Further study is required.
The overlap between OTOC and two QI definitions of scrambling
have been explored already:
(1) When the OTOC is small, a tripartite information is negative~\cite{HosurYoshida_16_Chaos}.
(2) An OTOC-like function is proportional to
a \emph{frame potential} that quantifies pseudorandomness~\cite{Roberts_16_Chaos}.
The relationship between the OTOC
and a third QI sense of incompatibility---MUBs and entropic uncertainty relations---merits investigation.

Second, $\OurKD{\rho}$ effectively has four arguments, apart from $\rho$
(Sec.~\ref{section:TA_Props}).
The KD quasiprobability has two.
This doubling of indices parallels the Choi-Jamiolkowski (CJ) representation
of quantum channels~\cite{Preskill_15_Ch3}.
Hosur \emph{et al.} have, using the CJ representation,
linked $F(t)$ to the tripartite information~\cite{HosurYoshida_16_Chaos}.
The extended KD distribution might be linked
to information-theoretic quantities similarly.

Third, our $P(W, W')$ and weak-measurement protocol
resemble analogs in~\cite{Solinas_15_Full,Solinas_16_Probing}.
\{See~\cite{Alonso_16_Thermodynamics,Miller_16_Time,Elouard_17_Role}
for frameworks similar to Solinas and Gasparinetti's (S\&G's).\}
Yet~\cite{Solinas_15_Full,Solinas_16_Probing}
concern quantum thermodynamics, not the OTOC.
The similarity between the quasiprobabilities in~\cite{Solinas_15_Full,Solinas_16_Probing}
and those in~\cite{YungerHalpern_17_Jarzynski},
their weak-measurement protocol and ours,
and the thermodynamic agendas in~\cite{Solinas_15_Full,Solinas_16_Probing}
and~\cite{YungerHalpern_17_Jarzynski}
suggest a connection between the projects~\cite{Jordan_chat,Solinas_chat}.
The connection merits investigation and might yield new insights.
For instance, S\&G calculate
the heat dissipated by an open quantum system
that absorbs work~\cite[Sec. IV]{Solinas_15_Full}.
OTOC theory focuses on closed systems.
Yet experimental systems are open.
Dissipation endangers measurements of $F(t)$.
Solinas and Gasparinetti's toolkit might facilitate predictions about,
and expose interesting physics in, open-system OTOCs.

Fourth, $W$ and $W'$ suggest understudies for work
in quantum thermodynamics.
Thermodynamics sprouted during the 1800s,
alongside steam engines and factories.
How much work a system could output---how
much ``orderly'' energy one could reliably draw---held practical importance.
Today's experimentalists draw energy from power plants.
Quantifying work may be less critical
than it was 150 years ago.
What can replace work in the today's growing incarnation of thermodynamics,
quantum thermodynamics?
Coherence relative to the energy eigenbasis is being quantified~\cite{Lostaglio_15_Description,Narasimhachar_15_Low}.
The OTOC suggests alternatives:
$W$ and $W'$ are random variables, analogous to work,
natural to quantum-information scrambling.
The potential roles of $W$ and $W'$ within quantum thermodynamics
merit exploration.

Fifth, relationships amongst three ideas were identified recently:
\begin{enumerate}[(1)]
   \item
We have linked quasiprobabilities with the OTOC,
following~\cite{YungerHalpern_17_Jarzynski}.
   \item
Aleiner \emph{et al.}~\cite{Aleiner_16_Microscopic}
and Haehl \emph{et al.}~\cite{Haehl_16_Schwinger_I,Haehl_16_Schwinger_II}
have linked the OTOC with Schwinger-Keldysh path integrals.
   \item
Hofer has linked Schwinger-Keldysh path integrals with quasiprobabilities~\cite{Hofer_17_Quasi}.
\end{enumerate}
The three ideas---quasiprobabilities, the OTOC, and Schwinger-Keldysh path integrals---form the nodes of the triangle in Fig.~\ref{fig:OTOC_path_quasi}.
The triangle's legs were discovered recently;
their joinings can be probed further.
For example, Hofer focuses on single-timefold path integrals.
OTOC path integrals contain multiple timefolds~\cite{Aleiner_16_Microscopic,Haehl_16_Schwinger_I,Haehl_16_Schwinger_II}.
Just as Hofer's quasiprobabilities involve fewer timefolds
than the OTOC quasiprobability $\OurKD{\rho}$,
the TOC quasiprobability $\TOCKD{\rho}$~\eqref{eq:TOCKD_def}
can be inferred from fewer weak measurements than $\OurKD{\rho}$ can.
One might expect Hofer's quasiprobabilities to relate to
$\TOCKD{\rho}$.
Kindred works, linking quasiprobabilities with out-of-time ordering, include~\cite{Manko_00_Lyapunov,Bednorz_13_Nonsymmetrized,Oehri_16_Time,Hofer_17_Quasi,Lee_17_On}.

%
%
\begin{figure}[h]
\centering
\includegraphics[width=.45\textwidth]{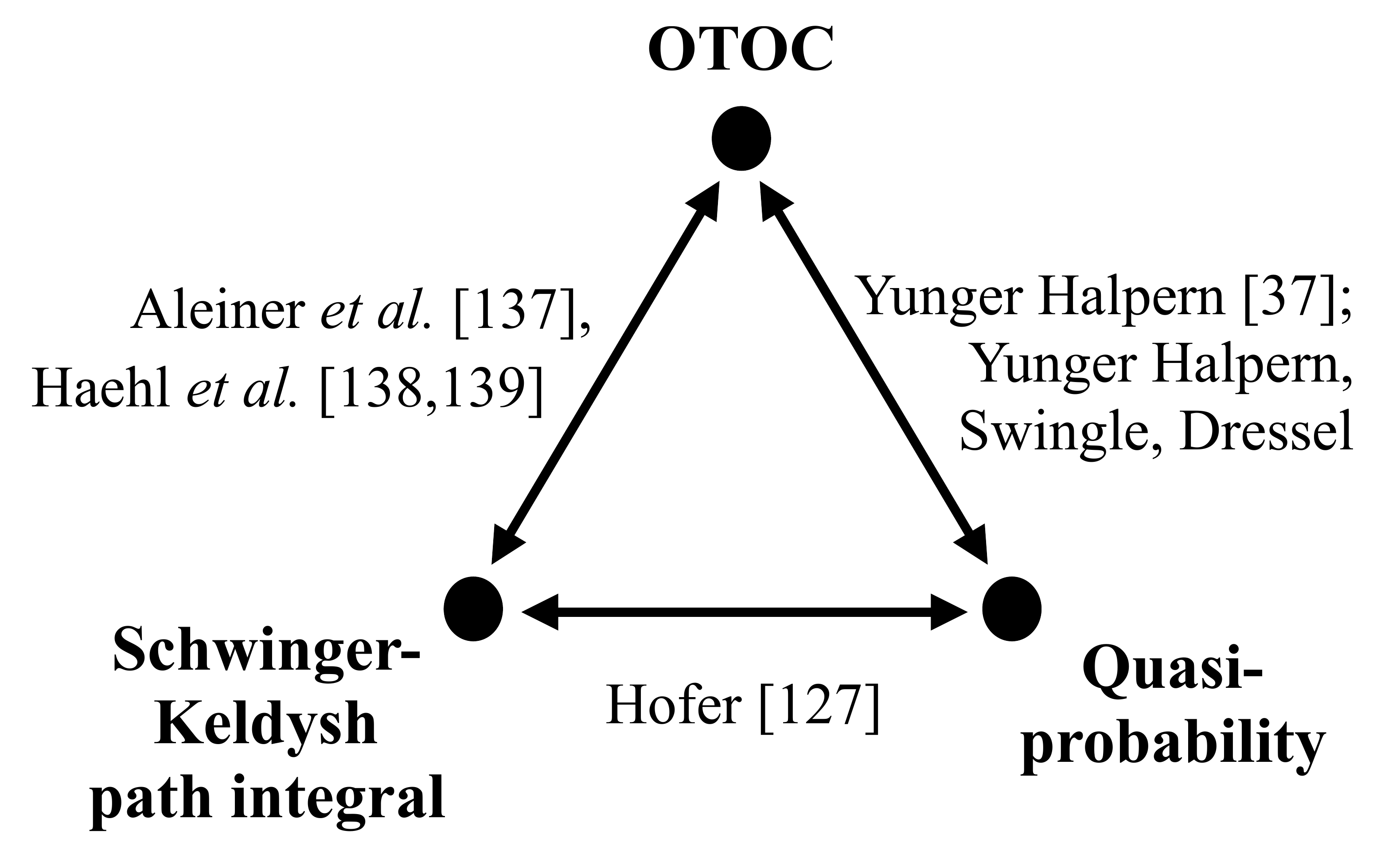}
\caption{\caphead{Three interrelated ideas:}
Relationships amongst the out-of-time-ordered correlator, quasiprobabilities,
and Schwinger-Keldysh path integrals
were articulated recently.}
\label{fig:OTOC_path_quasi}
\end{figure}

Sixth, the OTOC equals a moment of
the complex distribution $P( W, W')$~\cite{YungerHalpern_17_Jarzynski}.
The OTOC has been bounded with general-relativity
and Lieb-Robinson tools~\cite{Maldacena_15_Bound,Lashkari_13_Towards}.
A more information-theoretic bound might follow
from the Jarzynski-like equality in~\cite{YungerHalpern_17_Jarzynski}.

Finally, the KD distribution consists of the coefficients
in a decomposition of a quantum state
$\rho  \in  \mathcal{D} ( \mathcal{H} )$~\cite{Lundeen_11_Direct,Lundeen_12_Procedure}
(Sec.~\ref{section:KD_Coeffs}).
$\rho$ is decomposed in terms of a set
$\Basis  :=  \Set{  \frac{ \ketbra{ a }{ f } }{ \langle f | a \rangle }  }$
of operators.
$\Basis$ forms a basis for $\Hil$
only if $\langle f | a \rangle  \neq  0  \;  \forall a, f$.
The inner product has been nonzero in experiments,
because $\Set{ \ket{ a } }$ and $\Set{ \ket{f} }$
are chosen to be mutually unbiased bases (MUBs):
They are eigenbases of ``maximally noncommuting'' observables.
The OTOC, evaluated before the scrambling time $t = t_*$,
motivates a generalization beyond MUBs.
What if, $F(t)$ prompts us to ask,
$\langle f | a \rangle  =  0$ for some $a, f$
(Sec.~\ref{section:TA_Coeffs})?
The decomposition comes to be of
an ``asymmetrically decohered'' $\rho'$.
This decoherence's physical significance
merits investigation.
The asymmetry appears related to time irreversibility.
Tools from non-Hermitian quantum mechanics
might offer insight~\cite{Moiseyev_11_Non}.

%
%
\section*{Acknowledgements}
This research was supported by NSF grant PHY-0803371.
Partial support came from
the Walter Burke Institute for Theoretical Physics at Caltech. The Institute for Quantum Information and Matter (IQIM) is an NSF Physics Frontiers Center supported by the Gordon and Betty Moore Foundation.
NYH thanks Jordan Cotler and Paolo Solinas for pointing out the parallel with~\cite{Solinas_15_Full,Solinas_16_Probing};
David Ding for asking whether $\OurKD{\rho}$ represents a state;
Mukund Rangamani for discussing $\Opsb$-fold OTOCs;
Michele Campisi, Snir Gazit, John Goold, Jonathan Jones, Leigh Samuel Martin, Oskar Painter, and Norman Yao for discussing experiments;
and Christopher D. White and Elizabeth Crosson for discussing computational complexity.
Parts of this paper were developed while NYH was visiting
the Stanford ITP and UCL.
BGS is supported by the Simons Foundation,
as part of the It From Qubit collaboration;
through a Simons Investigator Award to Senthil Todadri;
and by MURI grant W911NF-14-1-0003 from ARO.
JD is supported by ARO Grant No. W911NF-15-1-0496.

\begin{appendices}


\renewcommand{\thesubsection}{\Alph{section}.\arabic{subsection}}
\renewcommand{\thesubsection}{\Alph{section}.\arabic{subsection}}
\renewcommand{\thesubsubsection}{\Alph{section}.\arabic{subsection}.\roman{subsubsection}}

\makeatletter\@addtoreset{equation}{section}
\def\theequation{\thesection\arabic{equation}}

\section{Mathematical properties of $P(W, W')$}
\label{section:P_Properties}

Summing $\OurKD{\rho}$, with constraints,
yields $P(W, W')$ [Eq.~\eqref{eq:PWWPrime}].
Hence properties of $\OurKD{\rho}$ (Sec.~\ref{section:TA_Props})
imply properties of $P(W, W')$.

\begin{property}
\label{property:P_Complex}
$P(W, W')$ is a map from
a composition of two sets of complex numbers
to the complex numbers:
$P \:  :  \:  \Set{ W }  \times  \Set{  W'  }  \to  \mathbb{C}$.
The range is not necessarily real:
$\mathbb{C}  \supset \mathbb{R}$.
\end{property}

Summing quasiprobability values can eliminate nonclassical behavior:
Interference can reduce quasiprobabilities' nonreality and negativity.
Property~\ref{prop:MargOurKD} consists of an example.
One might expect $P (W, W')$, a sum of $\OurKD{\rho} ( . )$ values,
to be real.
Yet $P(W, W')$ is nonreal in many numerical simulations
(Sec.~\ref{section:Numerics}).

\begin{property} \label{prop:MargP}

Marginalizing $P(W, W')$ over one argument
yields a probability if $\rho$ shares
the $\NondegV$ eigenbasis or the $\NondegW(t)$ eigenbasis.
\end{property}

Consider marginalizing Eq.~\eqref{eq:PWWPrime} over $W'$.
The $( w_2, \DegenW_{w_2} )$  and  $( v_1,  \DegenV_{v_1} )$
sums can be performed explicitly:
\begin{align}
   P(W)  & :=  \sum_{ W' }  P(W, W')
   \\ &  \label{eq:PW_Help1}
   = \sum_{ \substack{ ( v_2,  \DegenV_{v_2} ),  \\
                                     ( w_3,  \DegenW_{w_3} ) } }
   \langle w_3,  \DegenW_{w_3}  |  U  |  v_2,  \DegenV_{v_2}  \rangle
   \langle  v_2,  \DegenV_{v_2}  |  \rho  U^\dag  |
                w_3,  \DegenW_{w_3}  \rangle
   \nonumber \\ & \qquad \qquad \qquad   \times
   \delta_{W ( w_3^*  v_2^* ) }  \, .
\end{align}
The final expression is not obviously a probability.

But suppose that $\rho$ shares its eigenbasis with
$\NondegV$ or with $\NondegW(t)$.
Suppose, for example, that $\rho$ has
the form in Eq.~\eqref{eq:WRho}.
Equation~\eqref{eq:PW_Help1} simplifies:
\begin{align}
   \label{eq:PW_Help2}
   P(W)  & =
   \sum_{ \substack{ ( v_2,  \DegenV_{v_2} ),  \\
                                 ( w_3,  \DegenW_{w_3} ) } }
   p (  v_2,  \DegenV_{v_2} ;   w_3,  \DegenW_{w_3} )  \,
   \delta_{W ( w_3^*  v_2^* ) } \, .
\end{align}
The
\mbox{$p (  v_2,  \DegenV_{v_2} ;   w_3,  \DegenW_{w_3} )
:=  | \langle  w_3,  \DegenW_{w_3}  |  U  |
       v_2,  \DegenV_{v_2} \rangle |^2  \,
     p_{ v_2,  \DegenV_{v_2} }$}
denotes the joint probability that a $\NondegV$ measurement of $\rho$
yields $( v_2,  \DegenV_{v_2} )$
and, after a subsequent evolution under $U$,
a $\NondegW$ measurement yields $( w_3,  \DegenW_{w_3} )$.

Every factor in Eq.~\eqref{eq:PW_Help2} is nonnegative.
Summing over $W$ yields
a sum over the arguments of $\OurKD{\rho} ( . )$.
The latter sum equals one, by Property~\ref{prop:MargOurKD}:
$\sum_W  P(W)  =  1$.
Hence $P(W)  \in  [0, 1]$.
Hence $P(W)$ behaves as a probability.

We can generalize Property~\ref{prop:MargP} to arbitrary Gibbs states
$\rho = e^{ - H / T } / Z$,
using the regulated quasiprobability~\eqref{eq:RegKD2}.
The regulated OTOC~\eqref{eq:RegOTOC_def}
equals a moment of the complex distribution
\begin{align}
   \label{eq:RegP_def}
   & P_\reg ( W, W' )  :=
   \sum_{ \substack{ ( v_1,  \DegenV_{v_1} ),  ( w_2,  \DegenW_{w_2} ),
                                 ( v_2,  \DegenV_{v_2} )    ( w_3,  \DegenW_{w_3} ) } }
   \\  \nonumber  &
   \OurKD{\rho}^\reg  (  v_1,  \DegenV_{v_1} ;  w_2, \DegenW_{w_2} ;
   v_2,  \DegenV_{v_2}  ;  w_3,  \DegenW_{w_3} )  \,
   \delta_{ W ( w_3^* v_2^* ) }  \,  \delta_{ W' ( w_2  v_1 ) }  \, .
\end{align}
The proof is analogous to the proof of Theorem~1 in~\cite{YungerHalpern_17_Jarzynski}.

Summing over $W'$ yields
$P_\reg (W)  :=  \sum_{ W' }  P_\reg  (W, W')$.
We substitute in from Eq.~\eqref{eq:RegP_def},
then for $\OurKD{\rho}^\reg$ from Eq.~\eqref{eq:RegKD2}.
We perform the sum over $W'$ explicitly,
then the sums over $(w_2, \DegenW_{w_2} )$ and $( v_1 ,  \DegenV_{v_1} )$:
\begin{align}
   P_\reg (W)  =  \sum_{ \substack{ ( v_2,  \DegenV_{v_2} ) \\
                                                        ( w_3,  \DegenW_{w_3} ) } }
   | \langle  w_3,  \DegenW_{w_3}  |  \tilde{U} |
     v_2,  \DegenV_{v_2}  \rangle |^2  \,
   \delta_{ W ( w_3^*  v_2^* ) }  \, .
\end{align}
This expression is real and nonnegative.
$P_\reg(W)$ sums to one, as $P(W)$ does.
Hence $P_\reg(W)  \in  [ 0, \, 1 ]$ acts as a probability.

\begin{property}[Degeneracy of every $P(W, W')$
associated with $\rho = \id / \Dim$ and
with eigenvalue-$( \pm 1 )$ operators $\W$ and $V$]
\label{prop:P_Degen}

Let the eigenvalues of $\W$ and $V$ be $\pm 1$.
For example, let $\W$ and $V$ be Pauli operators.
Let $\rho = \id / \Dim$ be the infinite-temperature Gibbs state.
The complex distribution has the degeneracy
$P(1, -1 )  =  P(-1, 1)$.
\end{property}

Property~\ref{prop:P_Degen} follows from
(1) Eq.~\eqref{eq:SumKD_simple2} and
(2) Property~\ref{property:Syms} of $\OurKD{ ( \id / \Dim) }$.
Item (2) can be replaced with the trace's cyclicality.
We reason as follows:
$P(W, W')$ is defined in Eq.~\eqref{eq:PWWPrime}.
Performing the sums over the degeneracies
yields $\SumKD{ ( \id / \Dim) }$.
Substituting in from Eq.~\eqref{eq:SumKD_simple2} yields
\begin{align}
   \label{eq:P_Degen_Help1}
   P(W, W')  & =   \frac{1}{ \Dim }
   \sum_{ v_1 , w_2 , v_2 , w_3 }
   \Tr \left(  \ProjWt{w_3}  \ProjV{v_2}  \ProjWt{w_2}  \ProjV{v_1}  \right)
   \nonumber \\ & \qquad \qquad \qquad \; \times
   \delta_{W ( w_3^*  v_2^* ) }  \delta_{W' ( w_2  v_1 ) } \, .
\end{align}

Consider inferring $\OurKD{ ( \id / \Dim) }$ or $\SumKD{ ( \id / \Dim) }$
from weak measurements.
From one trial, we infer about four random variables:
$v_1,  w_2,  v_2$ and $w_3$.
Each variable equals $\pm 1$.
The quadruple $(v_1,  w_2,  v_2 ,  w_3)$ therefore
assumes one of sixteen possible values.
These four ``base'' variables are multiplied to form
the composite variables $W$ and $W'$.
The tuple $(W, W')$ assumes one of four possible values.
Every $(W, W')$ value can be formed from
each of four values of $(v_1,  w_2,  v_2 ,  w_3)$.
Table~\ref{table:P_Degen} lists the tuple-quadruple correspondences.

%
%
\begin{table*}[t]
\begin{center}
\begin{tabular}{|c|c|}
   \hline
        $(W, W')$
   &   $( v_1,  w_2,  v_2,  w_3 )$
   \\  \hline \hline
        $(1, 1)$
   &   $(1, 1, 1, 1),  (1, 1, -1, -1),  (-1, -1, 1, 1),  (-1, -1, -1, -1)$
   \\  \hline
        $(1, -1)$
   &   $(-1, 1, 1, 1),  (-1, 1, -1, -1),  (1, -1, 1, 1),  (1, -1, -1, -1)$
   \\  \hline
        $(-1, 1)$
   &   $(1, 1, -1, 1),  (1, 1, 1, -1),  (-1, -1, -1, 1),  (-1, -1, 1, -1)$
   \\  \hline
        $(-1, -1)$
   &   $(-1, 1, -1, 1),  (-1, 1, 1, -1),  (1, -1, -1, 1),  (1, -1, 1, -1)$
   \\  \hline
\end{tabular}
\caption{\caphead{Correspondence between
tuples of composite variables
and quadruples of ``base'' variables:}
From each weak-measurement trial, one learns about
a quadruple $( v_1,  w_2,  v_2,  w_3 )$.
Suppose that the out-of-time-ordered-correlator operators $\W$ and $V$
have the eigenvalues $w_\ell,  v_m  =  \pm 1$.
For example, suppose that $\W$ and $V$ are Pauli operators.
The quadruple's elements are combined into
$W := w_3^* v_2^*$ and $W'  :=  w_2 v_1$.
Each $(W, W')$ tuple can be formed from
each of four quadruples.}
\label{table:P_Degen}
\end{center}
\end{table*}

Consider any quadruple associated with
$(W, W')  =  (1, -1)$, e.g., $(-1, 1,  1,  1)$.
Consider swapping $w_2$ with $w_3$
and swapping $v_1$ with $v_2$.
The result, e.g., $(1, 1, -1, 1)$, leads to $(W, W')  =  (-1, 1)$.
This double swap amounts to a cyclic permutation
of the quadruple's elements.
This permutation is equivalent to
a cyclic permutation of the argument of
the~\eqref{eq:P_Degen_Help1} trace.
This permutation preserves the trace's value
while transforming the trace into $P(-1, 1)$.
The trace originally equaled $P(1, -1)$.
Hence $P(1, -1)  =  P(-1, 1)$.

\section{Retrodiction about the symmetrized composite observable $\tilde{\Gamma}  :=  i ( \K \ldots \A  -  \A \ldots \K )$}
\label{section:RetroK2}

Section~\ref{section:TA_retro} concerns retrodiction about
the symmetrized observable $\Gamma  :=  \K \ldots \A + \A \ldots \K$.
The product $\K \ldots \A$ is symmetrized also in
$\tilde{\Gamma}  :=  i ( \K \ldots \A  -  \A \ldots \K )$.
One can retrodict about $\tilde{\Gamma}$,
using $\Ops$-extended KD quasiprobabilities $\OurKD{\rho}^\ParenK$,
similarly to in Theorem~\ref{theorem:RetroK}.

The value most reasonably attributable retrodictively to
the time-$t'$ value of $\tilde{\Gamma}$ is
given by Eqs.~\eqref{eq:GammaW},~\eqref{eq:QuasiBayesLeft1},
and~\eqref{eq:QuasiBayesRt1}.
The conditional quasiprobabilities
on the right-hand sides of
Eqs.~\eqref{eq:QuasiBayesLeft2} and~\eqref{eq:QuasiBayesRt2} become
\begin{align}
   \label{eq:QuasiBayesLeft2Tilde}
   \tilde{p}_\rightarrow ( a, \ldots, k, f | \rho )
   =  \frac{ - \Im ( \langle f' | k \rangle  \langle k |  \ldots
   | a \rangle  \langle a |  \rho'  | f' \rangle ) }{
   \langle f' | \rho' | f' \rangle }
\end{align}
and
\begin{align}
   \label{eq:QuasiBayesRt2Tilde}
   \tilde{p}_\leftarrow ( k, \ldots, a, f | \rho )
   =  \frac{ \Im  ( \langle f' | a \rangle  \langle a |  \ldots
   | k \rangle  \langle k | \rho' | f' \rangle ) }{
   \langle f' | \rho' | f' \rangle }  \, .
\end{align}
The extended KD distributions become
\begin{align}
   \label{eq:Extend_KD_Left_Tilde}
   \OurKD{ \rho, \rightarrow }^\ParenK  ( \rho, a, \ldots, k , f )
   =  i  \langle f' | k \rangle  \langle k |  \ldots
   | a \rangle  \langle a |  \rho'  | f' \rangle
\end{align}
and
\begin{align}
   \label{eq:Extend_KD_Rt_Tilde}
   \OurKD{ \rho, \leftarrow }^\ParenK  ( \rho, k, \ldots, a , f )
   =  - i \langle f' | a \rangle  \langle a |  \ldots
   | k \rangle  \langle k | \rho | f' \rangle \, .
\end{align}

To prove this claim, we repeat the proof of Theorem~\ref{theorem:RetroK}
until reaching Eq.~\eqref{eq:Choose2}.
The definition of $\tilde{\Gamma}$ requires that
an $i$ enter the argument of the first $\Re$
and that a $-i$ enter the argument of the second $\Re$.
The identity $\Re ( i z )  =  - \Im (z)$, for $z \in \mathbb{C}$,
implies Eqs.~\eqref{eq:QuasiBayesLeft2Tilde}--\eqref{eq:Extend_KD_Rt_Tilde}.

\end{appendices}

%
%
\bibliographystyle{h-physrev}
\bibliography{OTOC_FT_bib}

\end{document}